\documentclass[%
 reprint,
 superscriptaddress,
 amsmath,amssymb,
 aps,
 prx
]{revtex4-2}
\usepackage[T1]{fontenc}
\usepackage{graphicx}% Include figure files
\usepackage{dcolumn}% Align table columns on decimal point
\usepackage{bm}% bold math
\usepackage[colorlinks=true, allcolors=blue]{hyperref}
\usepackage{braket}
\usepackage[table]{xcolor}
\usepackage{xfrac}
\usepackage{comment}
\usepackage[normalem]{ulem}
\usepackage{mathtools}
\usepackage{layouts}

\usepackage{tabularx}   % Tabulars with adjustable-width columns
\usepackage{booktabs}   % Publication quality tables
\usepackage{longtable}  % Allow tables to flow over page boundaries
\usepackage{multirow}  

\usepackage{quantikz}
\usepackage{tabularx}
\usepackage{multirow}
\usepackage{esint}
\usepackage{enumitem}

\usepackage{amsthm}% for theorems and stuff
\usepackage[capitalise]{cleveref}% load this as late as possible
\newcommand{\vect}[1]{\boldsymbol{\mathbf{#1}}}
\newcommand{\approptoinn}[2]{\mathrel{\vcenter{
  \offinterlineskip\halign{\hfil$##$\cr
    #1\propto\cr\noalign{\kern2pt}#1\sim\cr\noalign{\kern-2pt}}}}}

\newcommand{\dg}{\dagger}
\newcommand{\bo}[1]{\mathbf{#1}}
\providecommand{\boldsymbol}[1]{\boldsymbol{#1}}

\begin{document}
\title{Fault-Tolerant Non-Clifford GKP Gates using Polynomial Phase Gates and On-Demand Noise Biasing}
\author{Minh T. P. Nguyen}
\email{m.t.phamnguyen@tudelft.nl}
\affiliation{QuTech, Delft University of Technology, Lorentzweg 1, 2628 CJ Delft, The Netherlands}
\affiliation{Kavli Institute of Nanoscience, Delft University of Technology, Lorentzweg 1, 2628 CJ Delft, The Netherlands}
\author{Mackenzie H. Shaw}
\affiliation{QuTech, Delft University of Technology, Lorentzweg 1, 2628 CJ Delft, The Netherlands}
\affiliation{Delft Institute of Applied Mathematics, Delft University of Technology, Mekelweg 4, 2628 CD Delft, The Netherlands}
\newtheorem{theorem}{Theorem}[section]
\newtheorem{corollary}[theorem]{Corollary}
\newtheorem{lemma}[theorem]{Lemma}
\newtheorem{proposition}[theorem]{Proposition}
\newtheorem{criterion}{Criterion}
\newtheorem{appendixcriterion}{Criterion}[section]
\theoremstyle{definition}
\newtheorem{assumption}{Assumption}
\theoremstyle{remark}
\newtheorem*{remark}{Remark}
\theoremstyle{definition}
\newtheorem{definition}[theorem]{Definition}
%\numberwithin{equation}{section} %numbers equations by section, i.e. (sec.eqno)
\Crefname{criterion}{Criterion}{Criteria}
\crefname{criterion}{Crit.}{Crit.}

\begin{abstract}
The Gottesman-Kitaev-Preskill (GKP) error correcting code uses a bosonic mode to encode a logical qubit, and has the attractive property that its logical Clifford gates can be implemented using Gaussian unitary gates. In contrast, a direct unitary implementation of the ${T}$ gate using the cubic phase gate has been shown to have a logical error floor unless the GKP codestate has a biased noise profile. In this work, we propose a method for on-demand noise biasing based on a standard GKP error correction circuit. This on-demand biasing circuit can be used to bias the GKP codestate before a $T$ gate and return it to a non-biased state afterwards. With the on-demand biasing circuit, we prove that the logical error rate of the $T$ gate can be made arbitrarily small as the quality of the GKP codestates increases. We complement our proof with a numerical investigation of the cubic phase gate subject to a phenomenological noise model, showing that the ${T}$ gate can achieve average gate fidelities above $99\%$ with 12 dB of GKP squeezing without the use of postselection. Moreover, we develop a formalism for finding optimal unitary representations of logical diagonal gates in higher levels of the Clifford hierarchy that is based on a framework of ``polynomial phase stabilizers'' whose exponents are polynomial functions of one of the quadrature operators. This formalism naturally extends to multi-qubit logical gates and even to number-phase bosonic codes, providing a powerful algebraic tool for analyzing non-Clifford gates in bosonic quantum codes.

\end{abstract}

\maketitle

\section{Introduction}

Quantum error correction (QEC) is widely-believed to be necessary for large-scale, fault-tolerant quantum computing. In discrete variable (DV) QEC, logical qubits are encoded into many physical qubits \cite{Terhal_2015, Breuckmann2021LPDC}. However, the required resource overhead remains challenging for near-term hardware. On the other hand, continuous-variable (CV) QEC offers a more hardware-efficient alternative by encoding logical information into a single infinite-dimensional bosonic Hilbert space rather than multiple two-level systems \cite{albert_2025, Terhal_2020}. The logical information stored in such a code may already be low enough to perform an algorithm, or, more likely, the error rate would be low enough to concatenate with a code with lower overhead than would be possible using DV qubits.

Among the CV codes, the Gottesman–Kitaev–Preskill (GKP) code \cite{GKP2001} emerges as one of the leading candidates since all Clifford operations can be implemented unitarily using Gaussian operations \cite{GKP2001}. Moreover, when concatenated with a DV code, the decoded information given by the inner GKP code can be used to improve the decoder of the outer DV code \cite{Xu2023, Zhang2023, Vuillot2019nogo, Noh2022}, reducing the required hardware overhead. Approximate GKP states have been realized experimentally in multiple platforms, including superconducting circuits \cite{Sivak2023, brock2025quantum,nguyen2025superconductinggridstatesqubit}, trapped ions \cite{deNeeve2022, Matsos2024}, and integrated photonics \cite{Larsen2025}.

Nevertheless, achieving universal quantum computation requires the inclusion of a non-Clifford gate, typically the $T =\mathrm{diag}(1,e^{i\pi/4})$ gate. In the original GKP proposal \cite{GKP2001}, it was shown that the logical $T$ gate can be implemented unitarily using a \textit{cubic phase gate}, which is the exponential of a cubic polynomial of the CV position operator. However, Hastrup \textit{et al.}~\cite{Hastrup2021} recently showed that-in the absence of noise biasing-the fidelity of the cubic phase gate is upper bounded when applied to approximate GKP states, even in the limit of infinite squeezing and even when the states are decoded using ideal QEC. To circumvent this bound, one needs to perform \textit{on-demand} noise biasing, since Clifford gates (such as the Hadamard gate) can alter the noise bias during the algorithm.

\begin{figure}
    \centering
       \resizebox{0.45\textwidth}{!}{ \begin{quantikz}
            \lstick{$\cdots$}&\gate{\substack{q\text{-Steane}\\\text{QEC}}}\gategroup[1,steps=1,style={dashed,rounded corners,fill=blue!20,inner sep =0 pt},background,label style={label position=above,anchor=south,yshift=-0.2cm}]{$\substack{\text{On-demand}\\\text{biasing}\\\text{(\cref{section: mid-circuit noise biasing protocol})}}$} & \gate{T} \gategroup[1,steps=1,style={dashed,rounded corners,fill=green!20,inner sep =0 pt},background,label style={label position=below,anchor=south,yshift=-1.1cm}]{$\substack{\text{Optimized gate}\\\text{representation}\\\text{(\cref{section: optimized logical gate representations})}}$} &\gate{\substack{\text{Passive}\\\text{Knill}\\\text{QEC}}}&\rstick{$\cdots$}
        \end{quantikz}
        }
    \caption{Fault-tolerant circuit for implementing the logical $T$ gate using the cubic phase gate. The $q$-Steane QEC circuit~\cref{eq: q-Steane circuit} is applied before the $T$ gate to performs on-demand noise biasing of the GKP qubit. After the gate, the passive Knill QEC circuit~\cref{eq: passive Knill} resets the noise profile for subsequent computation. The implemented cubic phase gate can be optimized using our polynomial phase stabilizer framework to obtain a representation that minimizes error propagation.}
    \label{fig: T-gate with ondemand biasing}
\end{figure}
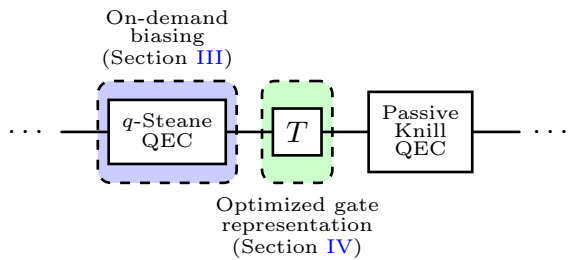

In this work we re-examine the possibility of implementing the logical $T$ gate in the GKP code unitarily. We show that the unitary $T$ gate is both fault tolerant and outperforms the vacuum-state injection method in a practical regime of interest. Key to our work is the observation that on-demand noise biasing can be performed using a standard GKP syndrome extraction circuit. This on-demand biasing circuit allows us to convert between a non-biased GKP code for Clifford gates and a biased GKP code for $T$ gates with the same overhead as is required for standard GKP QEC, as illustrated in Fig.~\ref{fig: T-gate with ondemand biasing}. With this noise biasing, we prove that the $T$ gate implemented is fault-tolerant, meaning that the logical gate fidelity goes to one as the quality of the GKP codestates state improves. We compare the unitary $T$ gate implementations with the vacuum state injection method numerically under a phenomenological noise model, and find that for the vacuum state method to achieve the same performance as the $T$ gate a post-selection probability below $20\%$ is necessary for GKP states with $\Delta < 0.25$ (which corresponds to 12 dB of squeezing).

Alongside our analysis of the cubic phase gate, we also develop an analytical framework of \textit{polynomial} phase gates that can implement gates in higher levels of the Clifford hierarchy. Central to this framework is the introduction of the \textit{polynomial phase stabilizer}; that is, any operator that acts trivially on the GKP codespace but whose exponent is a polynomial function of one of the quadrature operators. Using this, we develop a method to minimize the coefficients of the polynomial required to implement a given logical gate, enabling the systematic identification of optimal gate representations. For example, we derive a representation of the logical $T$ gate which first appeared in Ref.~\cite{Konno2021}, which we show to be optimal and to significantly outperform the original construction in Ref.~\cite{GKP2001}. We also numerically investigated optimized polynomial phase gates for the $T^{1/2}$, $T^{1/4}$ and $T^{1/8}$ gates, and found that these higher-order gates can achieve higher average gate fidelities than the $T$ gate in some regimes. Moreover, the framework applies to any gate whose exponent is a polynomial of one or more commuting position, momentum, or number operators, and therefore naturally generalizes to multi-qubit GKP gates as well as other bosonic codes, such as the multi-leg rotationally symmetric code \cite{Grimsmo2020} and the number-phase code \cite{hu2025generalized}.

The manuscript is organized as follows. In Section~\ref{section: preliminary}, we briefly review the square and rectangular GKP codes and introduce the noise models considered in this work. In Section~\ref{section: mid-circuit noise biasing protocol}, we explain why the fidelity of the cubic phase gate is fundamentally limited when applied to non-biased approximate GKP states \cite{Hastrup2021}. We then present the on-demand noise biasing circuit based on the Steane QEC circuit \cite{GKP2001}, which is applied immediately before the $T$ gate as shown in \cref{fig: T-gate with ondemand biasing}.
In Section~\ref{section: optimized logical gate representations}, we introduce the concept of polynomial phase stabilizers and use it to classify all representations of logical diagonal non-Clifford gates in the GKP code and identify the optimal representations. In Section~\ref{sec: error propagation and fault-tolerant theorem}, we analyze the propagation of errors when polynomial phase gates are applied to approximate GKP states, and prove our fault-tolerance theorem. Finally, in Section~\ref{sec: numerics}, we numerically compute the average gate fidelity of the polynomial phase gates using on-demand noise biasing, and compare their performance against the vacuum state method \cite{Baragiola2019}, before providing concluding remarks in \cref{sec: discussion}. While there remain significant challenges in performing our fault-tolerant unitary $T$ gate in experimental platforms such as the implementation of the unitary polynomial phase gates in photonics, or the suppression of DV qubit errors in superconducting platforms, our work provides a proof-of-principle that using polynomial phase gates can reduce the required overhead of implementing non-Clifford logical gates in the GKP code.

\subsection{Relationship to Existing Literature}

There are a variety of current approaches to realize fault-tolerant $T$ gates in the GKP code that we now briefly summarize and compare. The current leading approach to perform a ${T}$ gate is magic-state preparation by measuring the GKP stabilizers of a vacuum state \cite{Baragiola2019}, followed by injection and a non-linear feed-forward operation \cite{Konno2021}. The main advantage of this scheme is simplicity as it requires only Gaussian operations alongside standard GKP error correction. However, achieving high-fidelity magic states with the vacuum state method requires post-selection, significantly increasing the time overhead. In comparison, our approach instead assumes the additional capability to implement a unitary cubic phase gate. We show that the resulting magic-state fidelity exceeds that of the vacuum-state method with moderate post-selection rates, potentially reducing the subsequent cost of magic-state distillation.

An alternative approach is to implement a fault-tolerant logical $T$ through dissipative engineering. For example, Ref.~\cite{singh2025nonabelianquantumsignalprocessing} introduces a pieceable-circuit technique to realize logical phase gates by utilizing DV ancillas as bath, building on results of Refs.~\cite{campagne2020quantum,Royer2020finiteenergygate}. However, this approach is primarily tailored to superconducting platforms and typically requires multiple rounds --- along with biased-noise DV ancillas --- to achieve high-fidelity performance, which can lead to significant time overhead. 

In our work, we limit ourselves to approaches to performing the $T$ gate that do not use DV qubits for two reasons. First, some hardware platforms such as in optics do not have easy access to DV qubits. And second, we want to be able to prove the fault-tolerance of our approach in the limit of increasing GKP codestate quality. DV qubits complicate this picture because one has to also consider the effects of noise arising from the DV qubits themselves.

Finally, depending on the physical implementation of the GKP qubits, the dissipative bath may be realized directly. This idea has recently been explored in superconducting platforms \cite{Nathan2025, obrien2025exponentiallyrobustnoncliffordgate}, where a dissipative element (e.g.~a resistor) is coupled to the GKP mode through a controllable switch that can be used to implement a logical $T$ gate. We note that the optimal gate representations that we identified in Sec.~\ref{section: optimized logical gate representations} --- which realize logical $T$ and $\sqrt{T}$ gates using quartic and quadratic potentials --- are naturally implementable within the same scheme, as also pointed out by the authors of Ref.~\cite{obrien2025exponentiallyrobustnoncliffordgate}. The primary challenge here lies in realizing a picoseconds-scale switch which has not yet been demonstrated experimentally.

\section{Preliminaries}
\label{section: preliminary}
In this section, we introduce the notation and review the states, gates and QEC circuits for the single-mode square/rectangular GKP code.

For a single CV mode we use both the ladder operators ${a},{a}^{\dag}$ and the position ${q}$ and momentum ${p}$ operators that obey the canonical commutation relations $[{a},{a}^{\dagger}]=1$ and $[{q},{p}]=i$. We denote the position $x$-eigenstates as $\ket{x}_{q}$. We define the single-mode displacement operator ${W}(\bo v)$ where $\bo v = (v_q,v_p)\in\mathbb{R}^{2}$ as 
\begin{equation}
    {W}(\bo v) ={W}(v_{q},v_{p})= \exp\Big[ i \sqrt{2\pi}( v_p {q} - v_q {p} )\Big].  
\end{equation}
Note that our definition differs from the conventional one (e.g. in Ref.~\cite{Sakurai:1167961}) by a factor of $\sqrt{\pi}$.
The displacement operators form a group under multiplication, with the commutation relation
\begin{equation}
    \label{eq: displacement commutation relation}
    {W}(\bo u) {W}(\bo v) = e^{-i 2\pi(u_q v_p - u_p v_q)} {W}(\bo v) {W}(\bo u),
\end{equation}
and the composition rule 
\begin{equation}
    \label{eq: displacement composition rule}
     {W}(\bo u) {W}(\bo v) =  e^{-i \pi(u_q v_p - u_p v_q)} {W}(\bo u + \bo v). 
\end{equation} 
Displacement operators act on the quadrature operators via 

\begin{subequations}
    \begin{align}
        {W}(\bo u)^{\dagger}{q} {W}(\bo u)= {q} + \sqrt{2\pi} v_q, \\
        {W}(\bo u)^{\dagger}{p} {W}(\bo u)= {p} + \sqrt{2\pi} v_p.
    \end{align}
\end{subequations}

In this work, we will make use of both square and rectangular single-mode GKP codes. Similar to DV stabilizer codes, the GKP code can be defined in terms of its stabilizers \cite{GKP2001}
\begin{equation}\label{eq: rectangular stabilizers}
    {S}_p^{(\lambda)} = {W} (  \sqrt{2 \lambda},0 ), ~  {S}_q^{(\lambda)}  = {W} \Big(0,\sqrt{\frac{2}{\lambda}} \Big),
\end{equation}
where $\lambda>0$. When $\lambda=1$, this corresponds to the standard square GKP code, in which case we drop the $\lambda$ label from the stabilizers. In contrast to the square GKP code, the rectangular GKP code protects one quadrature at the expense of the other. As a result, it can offer improved performance under biased noise models, where errors are more likely to occur in one quadrature than the other. We will make use of this property later, where we will show that the logical ${T}$ gate causes an effectively asymmetric noise model to arise on the GKP code.

The logical Pauli operators are the displacement operators that commute with the stabilizers, given by
\begin{subequations}
    \begin{align}
            X^{(\lambda)} & = \sqrt{{S}_p^{(\lambda)}} = {W}\Big( \sqrt{\lambda/2},0 \Big), \\
        Z^{(\lambda)} & = \sqrt{{S}_q^{(\lambda)}} = {W}\Big( 0, 1/{\sqrt{2\lambda}} \Big). 
    \end{align}
\end{subequations}
The logical computational states $\ket{\bar 0}$ and $ \ket{\bar 1}$ are supported on even and odd integer multiples of $\sqrt{\lambda\pi}$ in the position basis:
\begin{equation}
    \label{eq: rectangular codewords}
    \ket{\bar 0}^{\lambda}  = \sum_{n \in \mathbb{Z} } \ket{2n \sqrt{\lambda\pi}}_q, ~ \ket{\bar 1}^{\lambda}  = \sum_{n \in \mathbb{Z}  } \ket{(2n+1) \sqrt{\lambda\pi}}_q.
\end{equation}

Logical Clifford gates in the GKP code are given by Gaussian unitary operators (i.e.~operators with an exponent that is a quadratic function of ${q}$ and ${p}$), and can be implemented by \cite{GKP2001}
\begin{subequations}
    \label{eq: GKP Clifford gates}
    \begin{align}
        {H}^{(\lambda)} & = \exp\Big(i\frac{\pi}{4}({q}^2/\lambda + \lambda{p}^2) \Big),\label{eq: logical H}\\
        {S}^{(\lambda)} & = \exp\Big(\frac{i}{2\lambda}{q}^2  \Big),\label{eq: logical S}\\
        {\rm CX}^{(\lambda)} & = \exp(-i {q}_1 {p}_2),\label{eq: logical CX}
    \end{align}
\end{subequations}
where we use subscripts to label the two modes participating in the ${\rm CX}^{(\lambda)}$ gate.
Throughout this work we will also make frequent use of two other Gaussian unitary operators that do not represent logical Clifford gates: the squeezing and beam-splitter operators defined by
\begin{subequations}
    \label{eq: squeeze and bs}
    \begin{align}
    {\text{Sq}}(\alpha) &= \exp\Big(i \frac{\ln(\alpha)}{2}( q  p +  p q)  \Big),\label{eq: squeeze}\\
    {\rm BS}(\theta) &= \exp\Big(i \theta({p}_1 {q}_2 - {q}_1 {p}_2 ) \Big).\label{eq: bs}
    \end{align}
\end{subequations}
The squeezing operator is named because it ``squeezes'' the quadrature operators
\begin{align}
    {\text{Sq}}(\alpha){q}{\text{Sq}}(\alpha)^{\dag}&=\alpha{q},&{\text{Sq}}(\alpha){p}{\text{Sq}}(\alpha)^{\dag}&={p}/\alpha.
\end{align}
Moreover, squeezing operators can be used to convert between square and rectangular GKP codes:
\begin{align}\label{eq: rectangular square conversion}
    \ket{\bar \psi}^{\lambda} &= {\text{Sq}}(\sqrt{\lambda})^{\dg} \ket{\bar \psi},&{O}^{(\lambda)}={\text{Sq}}(\sqrt{\lambda})^{\dg}\,{O}\,{\text{Sq}}(\sqrt{\lambda})
\end{align}
for any codestate $\ket{\bar{\psi}}$ and logical operator or stabilizer ${O}$.

To perform error correction on the GKP code, we need to measure the two stabilizers $ {S}_{q},{S}_p $, which is equivalent to measuring the position and momentum operators modulo $\sqrt{\lambda\pi}$ and $\sqrt{\pi/\lambda}$ (respectively). Then, given a measured modular position $q_{\text{m}}$ and momentum $p_{\text{m}}$, the state can then be returned to the GKP codespace by applying the corrective displacement ${W}(q_{\text{m}},p_{\text{m}})^{\dag}$.  
A displacement error ${W}(\bo v)$ acting on an ideal GKP state can therefore be detected and corrected if the displacement satisfies \cite{GKP2001,shaw2022}
\begin{align}\label{eq: correctable displacement}
    -\frac{\sqrt{\lambda}}{2\sqrt{2}}\leq v_{q}&<\frac{\sqrt{\lambda}}{2\sqrt{2}},&-\frac{1}{2\sqrt{2\lambda}}\leq v_{p}&<\frac{1}{2\sqrt{2\lambda}}.
\end{align}

A simple circuit that can be used to measure the stabilizers is the Steane QEC circuit. This consists of two subcircuits (which can be implemented either sequentially or in parallel) that measure the ${S}_{q}$ and ${S}_{p}$ stabilizers respectively. These are given explicitly by
\begin{subequations}
\label{eq: Steane circuit}
\begin{align}
\begin{quantikz}
        \lstick{$\ket{\bar \psi}^{\lambda}$} &  \gate[2]{ e^{-i {q}_1  {p}_2} }    &   \\
        \lstick{$\ket{\bar +}^{\lambda}$} &     &           \meter{{q}}
    \end{quantikz},
\label{eq: q-Steane circuit}\\
\begin{quantikz}
        \lstick{$\ket{\bar \psi}^{\lambda}$} & \gate[2]{ e^{i {p}_1  {q}_2} }    &   \\
        \lstick{$\ket{\bar 0}^{\lambda}$} &             & \meter{{p}}
    \end{quantikz},
\label{eq: p-Steane circuit}
\end{align}
\end{subequations}
where the measurements represent homodyne measurements of the ${q}$ or ${p}$ operators. We refer to \cref{eq: q-Steane circuit} as the $q$-Steane QEC circuit since it measures ${S}_{q}$, and to \cref{eq: p-Steane circuit} as the $p$-Steane QEC circuit. We will show in \cref{section: mid-circuit noise biasing protocol} that the $q$-Steane QEC circuit allows us to perform the on-demand biasing that is required to perform.
By tracking how a displacement operator ${W}(\mathbf{v})$ propagates through the ${\mathrm{CX}}$ gates, we find that an initial position displacement error on the data qubit propagates through the $q$-Steane QEC circuit to the ancilla and is detected by the position homodyne measurement (and likewise for momentum displacement errors in the $p$-Steane QEC circuit).

In our analysis, we also use another type of syndrome extraction circuit which we call the passive Knill QEC circuit, given by~\cite{Walshe2020}
\begin{equation}\label{eq: passive Knill}
    \begin{quantikz}
         &                  &  \gate[2]{\rm BS(\frac{\pi}{4})}   & \meter{{p}}  \\
        \lstick{$\ket{\emptyset}^{\lambda}$} &   \gate[2]{\rm BS(\frac{\pi}{4})}  &                 & \meter{{q}} \\
        \lstick{$\ket{\emptyset}^{\lambda}$} &                  &                 & 
    \end{quantikz}
\end{equation}
Here, the state $\ket{\emptyset}^{\lambda}$ is defined as 
\begin{equation}
    \ket{\emptyset}^{\lambda} = \sum_{n \in \mathbb{Z}} \ket{n \sqrt{2\lambda\pi}}_q.
\end{equation}
The passive Knill QEC circuit is convenient for a number of reasons: it uses only passive (i.e.~energy preserving) Gaussian unitary operators, it is more robust against noise on the ancilla modes than the Steane QEC circuit~\cite{Walshe2020}, and because the output state is guaranteed to be in the GKP codespace with only a logical correction needing to be applied after the circuit.

So far, we have only considered ideal GKP codewords \cref{eq: rectangular codewords}. In practice however, only approximate GKP states can be generated, since the ideal GKP codewords in \cref{eq: rectangular codewords} require infinite energy to realize. We define the approximate codestates as
\begin{equation}
    \ket{\bar{\psi}}_{\Delta}^{\lambda} = e^{-\Delta^2 ({q}^{2}+{p}^{2})/2} \ket{\bar \psi}^{\lambda},
\end{equation}
where we refer to $e^{-\Delta^{2} ({q}^{2}+{p}^{2})/2}=e^{-\Delta^{2}{a}^{\dag}{a}}=\mathrm{Env}_{\Delta}$ as the (non-biased) envelope operator.

In this work we also need to define \textit{biased} approximate GKP codestates, which are obtained by applying a biased envelope operator given by
\begin{subequations}\label{eq: biased envelope operator}
\begin{align}
    \text{Env}_{\Delta_{q},\Delta_{p}}&=e^{-(\Delta_{p}^{2}{q}^{2}+\Delta_{q}^{2}{p}^{2})/2}\\
    &={\text{Sq}}(\sqrt{\lambda})e^{-\Delta^{2}({q}^{2}+{p}^{2})/2}{\text{Sq}}(\sqrt{\lambda})^{\dag},
\end{align}
\end{subequations}
where $\Delta_{q}=\Delta/\sqrt{\lambda}$ and $\Delta_{p}=\sqrt{\lambda} \Delta$. The biased approximate codestate $\ket{\bar{\psi}}_{\Delta_{q},\Delta_{p}}^{\lambda}=\text{Env}_{\Delta_{q},\Delta_{p}}\ket{\bar{\psi}}^{\lambda}$ is the same as the ideal codestate but where in the position basis, each Dirac delta function in~\cref{eq: rectangular codewords} has been replaced by a Gaussian function of width $\sim\Delta_{q}$ and multiplied by a Gaussian envelope of width $\sim 1/\Delta_{p}$ (and vice versa in the momentum basis). This means that the approximate codestates are normalizable, but also that approximate codestates have inherent errors on them due to the width of the Gaussian functions.

Importantly, we can always map a \textit{biased} square GKP qubit to a \textit{non-biased} rectangular GKP qubit using a unitary squeezing operator
\begin{equation}\label{eq: relation between biased square GKP to non-biased rectangular GKP}
    \ket{\bar{\psi}}_{\Delta_{q},\Delta_{p}}=\text{Sq}(\sqrt{\lambda})\ket{\bar{\psi}}_{{\Delta}}^{\lambda},
\end{equation}
where $\lambda = \Delta_p/\Delta_q$ and $\Delta = \sqrt{\Delta_q \Delta_p}$. 
Therefore, in the absence of other noise, the biased square GKP qubit $\ket{\bar{\psi}}_{\Delta_{q},\Delta_{p}}$ is equivalent to the non-biased rectangular GKP qubit $\ket{\bar{\psi}}_{\Delta}^{\lambda}$ in terms of the logical information and error rate. 

Since the GKP code is designed to protect against displacement errors that are small enough as in~\cref{eq: correctable displacement}, a common way to analyze an error operator or channel is to decompose it into displacement operators. Because the displacement operators form a complete orthogonal basis for the operator space, any operator ${O}$ can be decomposed in terms of displacement operators as
\begin{equation}
    {O}   = \int_{} d \bo v~ \chi_{{O}}(\bo v)  {W}(\bo v),
\end{equation}
where the characteristic function $\chi_{{O}}$ of operator ${O}$ is given by
\begin{equation}
    \chi_{{O}}(\bo v) = \text{Tr}[{O}~ {W}^{\dagger}(\bo v) ]. 
\end{equation}
Analogously, any quantum channel $\mathcal{E}$ can be decomposed in terms of displacement operators as 
\begin{equation}\label{eq: channel chi function}
    \mathcal{E}(\rho) = \int d\bo v~ d \bo u~ \chi_{\mathcal{E}}(\bo v,\bo u) ~ {W}(\bo v) \rho {W}(\bo u)^{\dagger},
\end{equation}
where the $\chi$-function $\chi_{\mathcal{E}}(\bo u, \bo v)$ is a CV generalization of the DV $\chi$-matrix, and corresponds to the characteristic function of the superoperator representation of $\mathcal{E}$ \cite{shaw2022}. In our analysis, we frequently employ the twirling approximation, which retains only the diagonal components of \cref{eq: channel chi function}
\begin{equation}\label{eq: twirling definition}
    \text{Twirl}[\mathcal{E}](\rho) = \int d \bo v~ p_{\mathcal{E}}(\bo v) {W}(\bo v)\rho {W}(\bo v)^{\dg},
\end{equation}
for a more detailed discussion see Ref.~\cite{Conrad2021}. The twirling approximation is useful for us because \cref{eq: twirling definition} can be viewed as applying a displacement $W(\vect{v})$ randomly according to the probability density $p_{\mathcal{E}}(\vect{v})=\chi_{\mathcal{E}}(\vect{v},\vect{v})$. This is convenient because we can quickly determine the probability that the random displacement is correctable by integrating $p_{\mathcal{E}}(\vect{v})$ over the set of correctable displacements \cref{eq: correctable displacement}.

When $p_{\mathcal{E}}(\vect{v})$ is a Gaussian distribution, we call the twirled channel a Gaussian random displacement channel
\begin{equation}\label{eq: GRD channel definition}
    \mathcal{G}_{\Sigma}(\rho) = \int d \bo v~\mathcal{N}(\Sigma,\bo v) ~ {W}(\bo v)\rho{W}(\bo v)^{\dagger}, 
\end{equation}
where we use a slightly non-standard definition of the 2D Gaussian distribution $\mathcal{N}(\Sigma,\bo v)$ given by
\begin{equation}
    \mathcal{N}(\Sigma,\bo v) = \frac{ \exp\big[-\pi\bo v^{T}~ \Sigma^{-1}\bo v\big]  }{\sqrt{|\det \Sigma|}}. 
\end{equation}
Here, $\Sigma$ is the covariance matrix of the distribution (rescaled by a factor of $\sqrt{2\pi}$).

It was shown in Ref.~\cite{shaw2022} that the biased envelope operator $\mathrm{Env}_{\Delta_{q},\Delta_{p}}$ has a characteristic function given by a Gaussian function
\begin{subequations}\label{eq: biased envelope characteristic function}
    \begin{align}
        \chi_{\mathrm{Env}}(\vect{v}) & \propto \mathcal{N}\bigg[2\,\mathrm{tanh}\Big(\frac{\Delta_{q}\Delta_{p}}{2}\Big)\mathrm{Diag}\Big(\frac{\Delta_{q}}{\Delta_{p}},\frac{\Delta_{p}}{\Delta_{q}}\Big),\vect{v}\bigg]\\
        % =\int d  \bo v   \frac{ e^{ -\frac{\pi}{2} \coth(\frac{\tilde \Delta^2}{2}) (\lambda v_q^2 + v_p^2/\lambda)} }{1- e^{-\tilde{\Delta}^2}}   \hat W(\bo v)  \\
        & \approx \mathcal{N}\Big(\text{Diag}(\Delta_q^2,\Delta_p^2),\bo v\Big),\label{eq: envelope characteristic function approx}
    \end{align}
\end{subequations}
where \cref{eq: envelope characteristic function approx} is obtained by taking a first-order approximation of the $\mathrm{tanh}$ function. Under the twirling approximation it therefore becomes a Gaussian random displacement channel $\mathcal{G}_{\Sigma}$ with $\Sigma\approx \mathrm{Diag}(\Delta_{q}^{2},\Delta_{p}^{2})/2$. 

We remark that in physical devices, photon loss and dephasing channels are also significant sources of error. However, we do not include these noise channels in our analysis. This choice is made to clearly illustrate our core ideas using a realistic yet mathematically tractable model. Importantly, under the twirling approximation, these noise channels can also be converted into random displacement channels. For example, applying the twirling approximation to loss gives us directly the Gaussian random displacement channel~\cite{shaw2024logical}. One could therefore imagine extending the analysis presented in this work to account for photon loss or dephasing noise with some appropriate modifications.

\section{On-demand noise biasing circuits}
\label{section: mid-circuit noise biasing protocol}
In this section, we discuss why on-demand noise biasing is essential for achieving a high-fidelity implementation of the cubic phase $T$ gate and outline how on-demand biasing can be realized. In Section~\ref{section: Error propagation of cubic phase gate}, we provide an intuitive explanation of how the cubic phase gate amplifies noise from the envelope operator. Building on this intuition, we show how on-demand noise biasing can mitigate the resulting amplified noise. Then, in Section~\ref{section: biased syndrome measurement circuit}, we demonstrate that a fixed amount of on-demand noise biasing can be implemented using the standard $q$-Steane QEC circuit, and thus requires no additional resources beyond those of a conventional QEC circuit. General amounts of noise biasing can then be achieved by a simple generalization to the $q$-Steane QEC circuit that makes use either of a biased square GKP ancilla codestate or a non-biased rectangular codestate which can be prepared in advance.

\subsection{Error propagation of the cubic phase gate}
\label{section: Error propagation of cubic phase gate}

To understand the existing problems with the cubic phase gate, we consider a simple circuit to generate the ideal square GKP state $\ket{\bar{T}}=T\ket{\bar{+}}$
\begin{equation}\label{eq: simple T_GKP circuit}
    \begin{quantikz}
       \lstick{$\ket{\bar{+}}$} & \gate{ {T}_{\rm GKP} } &   
\end{quantikz}.
\end{equation}
Here, $T_{\rm GKP}$ is the specific realization of the cubic phase gate proposed in Ref.~\cite{GKP2001} for implementing the logical $T$ gate, defined as
\begin{equation}\label{eq: T_GKP}
    T_{\rm GKP} =  \exp\left\{i2\pi \Big[ \frac{1}{4} {\!\left(\frac{{q} }{\sqrt{\pi}}\right)\!}^3 + \frac{1}{8} {\!\left(\frac{{q} }{\sqrt{\pi}}\right)\!}^2 - \frac{1}{4} \!\left(\frac{{q} }{\sqrt{\pi}}\right)\!   \Big] \!\right\}\!. 
\end{equation}
If the input state to \cref{eq: simple T_GKP circuit} is ideal, then the output state is an ideal square GKP $\ket{\bar{T}}$ state. To see this, note that the ideal $\ket{\bar{0}}$ (and, respectively, $\ket{\bar{1}}$) codestates are superpositions of position eigenstates with positions at even (odd) multiples of $\sqrt{\pi}$~\cref{eq: rectangular codewords}. Meanwhile, note that the cubic polynomial
\begin{equation}\label{eq: T_GKP polynomial}
    P(x)=\frac{1}{4}x^{3}+\frac{1}{8}x^{2}-\frac{1}{4}x
\end{equation}
has the property that if $x$ is an even integer then $P(x)$ is an integer, and if $x$ is an odd integer then $P(x)$ is $n+1/8$ (for some integer $n$). Therefore, the cubic phase gate $T_{\mathrm{GKP}}=\mathrm{exp}\big(i2\pi P({q}/\sqrt{\pi})\big)$ acts as
\begin{subequations}
\begin{align}
    T_{\mathrm{GKP}}\ket{2n\sqrt{\pi}}_{q}&=\ket{2n\sqrt{\pi}}_{q},\\
    T_{\mathrm{GKP}}\ket{(2n+1)\sqrt{\pi}}_{q}&=e^{i\pi/4}\ket{(2n+1)\sqrt{\pi}}_{q},
\end{align}
\end{subequations}
for all integers $n$, and implements a logical $T$ gate.

However, as discussed in \cref{section: preliminary}, only approximate GKP states can be prepared in practice. Approximate codestates always introduce some logical errors due to the envelope operator, and the probability of a logical error goes to $0$ as $\Delta\rightarrow0$. However, it was shown in Ref.~\cite{Hastrup2021} that the ${T}_{\text{GKP}}$ gate amplifies the noise due to the envelope operator in such a way that the probability of a logical error no longer goes to $0$ as $\Delta\rightarrow0$, we will also rederive this result using the twirling approximation in \cref{sec: error propagation and fault-tolerant theorem}. 

One can understand this intuitively by considering a small perturbation to the position eigenstate $\ket{(n+\epsilon)\sqrt{\pi}}_{q}$ that the cubic phase gate acts upon. In particular, the polynomial $P(n+\epsilon)$ is evaluated at a position that is a small distance $\epsilon$ away from being an integer $n$. Roughly speaking, the error due to the $T_{\text{GKP}}$ gate comes from the difference in phase that is applied due to the perturbation $P(n+\epsilon)-P(n)$. To estimate how this error scales with $\Delta$, note that in the position basis, the approximate codestate is a superposition of Gaussian functions with width $\sim \Delta_{q}$ multiplied by an envelope of width $\sim 1/\Delta_{p}$. For an unbiased approximate codestate with $\Delta=\Delta_{q}=\Delta_{p}$, we therefore set $\epsilon=O(\Delta)$ and $n=O(1/\Delta)$ because larger values of $\epsilon$ and $n$ are exponentially suppressed by the Gaussian distribution. In this case, the error in the phase that is applied due to the perturbation $P(n)-P(n+\epsilon)$ contains terms that do not decrease as $\Delta\rightarrow0$:
\begin{equation}
    P(n+\epsilon)-P(n)=\frac{3}{4}n^{2}\epsilon+O(1)=O(\Delta^{-1}).
\end{equation}
We call this leading term the ``shear error'' and it is the leading source of error in the ${T}_{\text{GKP}}$ gate.
In particular, due to the shear error, the error in the phase of the resulting state does not decrease as $\Delta$ tends to zero - limiting the logical error rate that can be achieved \cite{Tzitrin2020}. Note, however, that this intuitive argument is not watertight and does not produce quantitatively correct predictions of the behavior of the ${T}_{\text{GKP}}$ gate (see footnote \footnote{Indeed, this argument incorrectly predicts the level of noise bias that would be optimal compared to our more detailed analysis in \cref{app: optimal bias estimate,sec: numerics}. Moreover, repeating the argument for the ${S}$ gate suggests that it suffers from the same problems as the ${T}$ gate; however, this is not the case, and the error spread due to the ${S}$ gate can be corrected for exactly if one has access to ideal QEC~\cite{shaw2024logical}}), for this reason we will more carefully derive the error induced by the $T$ gate in \cref{sec: error propagation and fault-tolerant theorem}.

The solution to this is to instead \textit{bias} the noise profile by using a biased envelope operator \cref{eq: biased envelope operator}, which improves gate fidelity by two ways. First, decreasing $\Delta_q$ reduces the typical displacement in the position quadrature $\epsilon$; and second, increasing $\Delta_p$ tightens the envelope of the codeword position-space wavefunction, effectively reducing $n$. Of course, increasing $\Delta_{p}$ also increases the error rate because it \textit{increases} the width of the envelope in the \textit{momentum} quadrature, introducing a trade-off between large and small noise bias.

An alternative, entirely equivalent solution is to use non-biased \textit{rectangular} codestates. This is equivalent because of the unitary relationship between biased square and non-biased rectangular codestates given by \cref{eq: relation between biased square GKP to non-biased rectangular GKP}; note that in this case one also has to use the rectangular $T_{\text{GKP}}^{(\lambda)}=\mathrm{Sq}(\sqrt{\lambda})^{\dag}\,T_{\text{GKP}}\,\mathrm{Sq}(\sqrt{\lambda})$ logical gate.

Regardless of whether one uses a biased square or non-biased rectangular codestate to perform a $T$ gate, it is crucial that the bias is prepared \textit{on-demand}. This is primarily because Clifford gates (for example the Hadamard gate~\cref{eq: logical H}) are not bias-preserving: in the square GKP code, performing a logical Hadamard gate causes the bias to be inverted $\Delta_{q}\leftrightarrow\Delta_{p}$. Moreover, a biased square GKP qubit is less protected while idling compared to a non-biased square GKP qubit, see for example Fig.~\ref{fig:all_gates}(a). Therefore, in the following section we will derive an on-demand biasing circuit that allows one to bias the codestate just before the $T$ gate and return to a non-biased codestate afterwards.

\subsection{On-demand biasing circuit}
\label{section: biased syndrome measurement circuit}
The on-demand biasing circuit is based on the $q$-Steane QEC circuit~\cref{eq: q-Steane circuit}. We will first begin by explaining how the $q$-Steane QEC circuit introduces a fixed amount of noise bias, and then we will generalize the $q$-Steane QEC circuit to be able to introduce an arbitrary amount of bias $\lambda$.

We begin with an identity from Ref.~\cite{Glancy2006}, which rewrites the $q$-Steane measurement circuit in terms of beam-splitters and single-mode squeezers~\cref{eq: squeeze and bs}:
\begin{subequations}\label{eq: Glancy Knill BS}
    \begin{align}
    \label{eq: Glancy Knill step 1}
& \hphantom{{}={}} \begin{quantikz}
     \lstick{$\ket{\bar \psi}_{\Delta}$}   &  \gate[2]{ e^{-i {q}_1  {p}_2} } & \gate[2]{ e^{i \frac{{p}_1 {q}_2}{2}}}  &    \\
     \lstick{$\ket{\bar +}_{\Delta}$}  &  &    & \meter{{q}}
\end{quantikz} \\
& = \begin{quantikz}
     \lstick{$\ket{\bar \psi}_{\Delta}$} &  \gate[2]{ \rm BS(\frac{\pi}{4}) } & \gate{\text{Sq}^{\dg}(\sqrt{2}) }   &   \\
     \lstick{$\ket{\bar +}_{\Delta}$} &                                & \gate{\text{Sq}^{\dg}(\frac{1}{\sqrt{2}})} & \meter{{q}}
\end{quantikz}.\label{eq: q-Steane rewritten}
\end{align}
\end{subequations}
Note that \cref{eq: Glancy Knill step 1} is the same as the $q$-Steane circuit depicted in \cref{eq: q-Steane circuit} except for the presence of the $\mathrm{exp}(i{p}_{1}{q}_{2}/2)$ gate. However, this gate can be compensated for by applying a corrective displacement $\exp(-i{p}_{1}q_{\text{m}}/2)$ based on the outcome $q_{\text{m}}$ of the measurement; this correction can be performed at the same time as the usual displacement correction that is done to return the GKP qubit to the codespace. \cref{eq: q-Steane rewritten} is simply a linear algebra rearrangement of the unitary operators, but it allows us to reason about the spread of noise as we now explain.

The beam-splitter $\rm BS(\theta)$ preserves the non-biased noise profile of both data and ancilla qubits, as shown by the commutation identity \cite{shaw2022,Walshe2020}
\begin{equation}
     \begin{quantikz}
      &[-0.28cm]\gate{\mathrm{Env}_{\Delta}} & \gate[2]{ \rm BS(\theta) }&[-0.28cm]    \\
      &[-0.28cm]\gate{\mathrm{Env}_{\Delta}} & &[-0.28cm]                             
\end{quantikz} = \begin{quantikz}
    &[-0.28cm]\gate[2]{ \rm BS(\theta) } & \gate{\mathrm{Env}_{\Delta}} &[-0.28cm]  \\
    &[-0.28cm] & \gate{\mathrm{Env}_{\Delta}} &[-0.28cm]
\end{quantikz}.
\end{equation}  
However, the single-mode squeezing operator on the data qubit after does modify the bias according to 
\begin{equation}\label{eq: BS envelope identity}
    \Delta_q = \Delta \to \frac{\Delta}{\sqrt{2}}, \quad \Delta_p = \Delta \to \sqrt{2} \Delta. 
\end{equation}
Therefore, after the displacement correction, the output codestate of the $q$-Steane QEC circuit (implemented either with Eq.~\eqref{eq: q-Steane circuit} or \eqref{eq: q-Steane rewritten}) will be a square GKP codestate $\ket{\bar{\psi}}_{\Delta/\sqrt{2},\sqrt{2}\Delta}$ with a bias of two. This shows that by measuring just one of the two GKP stabilizers, we have effectively turned a non-biased GKP codestate into a biased GKP codestate: in other words, we have performed on-demand noise biasing.

Intuitively, the reduction of $\Delta_q$ a factor of $\sqrt{2}$ is the direct consequence of measuring and correcting the position quadrature, which reduces its error variance. Conversely, the momentum quadrature becomes noisier because the momentum displacement errors arising from the envelope operator on the ancilla propagate to the data qubit through the $\text{CX}$ gate, adding to the noise in the momentum quadrature of the data qubit.

However, as $\Delta$ becomes smaller, we need increasing levels of noise bias to be able to perform the $\mathrm{GKP}$ gate with decreasing noise levels. Indeed, we will later see in \cref{sec: numerics} that for $\Delta<0.25$ (GKP squeezing above 12 dB) the optimal level of asymmetry is already larger than what can be achieved with the $q$-Steane QEC circuit alone. There are a number of ways to generate higher levels of bias using modifications of the $q$-Steane QEC circuit, which we now outline one at a time.

The most conceptually simple way to increase the bias is to perform multiple $q$-Steane QEC circuit in a row, which increases the bias of the square GKP data qubit each time. With $r$ applications of the $q$-Steane QEC circuit, one can prepare a biased square GKP codestate $\ket{\bar{\psi}}_{\Delta/\sqrt{r+1},\sqrt{r+1}\Delta}$.

However, a more time-efficient alternative is to use just a single $q$-Steane QEC circuit with a \textit{biased} square GKP ancilla
\begin{equation}
    \label{eq: biased noise q-Steane}
    \begin{quantikz}
        \ket{\bar \psi}_{\Delta } & \gate[2]{ e^{-i {q}_1  {p}_2} } &    \\
       \ket{ \bar +}_{\Delta/\sqrt{\lambda},\Delta \sqrt{\lambda}}  & &   \meter{{q}} 
    \end{quantikz}.
\end{equation}
In this circuit, due to the reduced variance $\Delta/\sqrt{\lambda}$ of the position displacement error, the syndrome measurement becomes more sensitive to position displacement errors coming from the data qubit, reducing the noise in the position quadrature. Conversely, momentum displacement errors on the ancilla are also amplified by the factor $\sqrt{\lambda}$ and propagate back to the data qubit through the CX gate. A straightforward calculation shows that the effective squeezing parameters of the data qubit are transformed as
\begin{equation}\label{eq: biased noise q-Steane resulting bias}
    \Delta_q = \Delta \to \frac{\Delta}{\sqrt{1+\lambda}}, \quad \Delta_p = \Delta \to \sqrt{1+\lambda} \Delta. 
\end{equation}
Note that even though \cref{eq: biased noise q-Steane} itself already uses a biased noise ancilla, this state can be prepared ``offline'' instead of on-demand.

Finally, an entirely equivalent alternative to \cref{eq: biased noise q-Steane} is to use a non-biased rectangular ancilla
\begin{equation}\label{eq: biased noise q-Steane with rect ancilla}
    \begin{quantikz}
        \ket{\bar \psi}_{\Delta } &  \gate[2]{ e^{-i \sqrt{\lambda} {q}_1  {p}_2} }    &  \\
       \ket{ \bar +}_{\Delta}^{\lambda}  &   & \meter{{q}}
    \end{quantikz},
\end{equation}
which also generates the same biasing as in \cref{eq: biased noise q-Steane resulting bias}. This can be derived from \cref{eq: biased noise q-Steane} by inserting squeezing operator identities $\mathrm{Sq}(\sqrt{\lambda})\mathrm{Sq}(\sqrt{\lambda})^{\dag}$ into the circuit.

After biasing the noise with either \cref{eq: biased noise q-Steane} or \cref{eq: biased noise q-Steane with rect ancilla} to the optimal amount of biasing for the given $\Delta$ and apply the corrective displacement, we proceed with the application of the cubic phase gate $T_{\rm GKP}$. Following the gate, the noise profile becomes highly asymmetric and distorted, see Fig.~\ref{fig: twirled_chi}. If left uncorrected, subsequent gates can further amplify these distortions, leading to rapid degradation of the logical state fidelity. Therefore, using the same idea as Ref.~\cite{matsuura2024continuous}, we return to the unbiased square GKP code using the passive Knill QEC circuit~\cite{Walshe2020}, which is guaranteed to output a fresh unbiased square GKP codestate. The passive Knill QEC circuit can be used exactly as shown in \cref{eq: passive Knill} because the input state is a square GKP codestate (even though it is biased).

In summary, to apply the cubic phase gate, we first perform noise biasing using Eq.~\eqref{eq: biased noise q-Steane} before the gate and then perform full round of error correction afterward using the passive Knill QEC circuit:
\begin{equation}
    \begin{quantikz}
        \lstick{$\cdots$}&\gate{\substack{q\text{-Steane}\\\text{QEC}}}\gategroup[1,steps=1,style={dashed,rounded corners,fill=blue!20,inner sep =0 pt},background,label style={label position=above,anchor=south,yshift=-0.2cm}]{$\substack{\text{On-demand}\\\text{biasing}}$}&\gate{T}&\gate{\substack{\text{Passive}\\\text{Knill}\\\text{QEC}}}&\rstick{$\cdots$}
    \end{quantikz}.\label{eq: T gate with on-demand biasing}
\end{equation}
In contrast, to apply any Clifford gate we simply apply the gate between two rounds of QEC:
\begin{equation}
    \begin{quantikz}
        \lstick{$\cdots$}&\gate{\substack{\text{Passive}\\\text{Knill}\\\text{QEC}}}&\gate{C}&\gate{\substack{\text{Passive}\\\text{Knill}\\\text{QEC}}}&\rstick{$\cdots$}
    \end{quantikz}.\label{eq: Clifford gate summary}
\end{equation}

An important practical question is how to generate the biased-noise square GKP ancilla $\ket{\bar +}_{\Delta/\sqrt{\lambda},\Delta \sqrt{\lambda}}$ required in Eq.~\eqref{eq: biased noise q-Steane} (or the non-biased rectangular ancilla in \cref{eq: biased noise q-Steane with rect ancilla}). Existing preparation protocols \cite{Le2019,Rymarz2021,hastrup2021measurement,Kolesnikow2024gottesman,Weigand2016} typically yield non-biased square GKP states. While these protocols likely could be adapted to generate rectangular states, we do not detail such modifications here. Instead, as a proof of principle, we show in Appendix~\ref{appendix: passive biasing circuit} how to use a slight modification of the $q$-Steane QEC circuits to convert non-biased square GKP states $\ket{\bar +}_{\Delta}$ into non-biased rectangular GKP states $\ket{\bar +}^{\lambda}_{\Delta}$ using only non-balanced beam splitters and displacement operations. The non-biased rectangular GKP state $\ket{\bar +}^{\lambda}_{\Delta}$ can be converted to the biased noise square GKP $\ket{\bar +}_{\Delta/\sqrt{\lambda},\Delta \sqrt{\lambda}}$ by a single-mode squeezer. The advantage of our procedure is that it is independent of the specific method used to generate the initial non-biased square GKP qubit.

We briefly remark that the above method applies the $T$ gate on a biased square GKP codestate. However, one can also apply the $T$ gate on an unbiased rectangular GKP codestate. To do this, one needs an on-demand \textit{morphing} circuit that ``morphs'' from the square GKP code to the rectangular GKP code while preserving a non-biased noise profile. We derive such an on-demand morphing circuit in \cref{appendix: mid-circuit morphing} and explain how it can be used to morph to the rectangular code for the $T$ gate and back to the square code for Clifford gates. This approach may be more attractive experimentally because it only uses a \textit{passive} Gaussian operation (a beam-splitter) instead of a CX gate, and because a non-biased rectangular codestate has a lower average photon number than the corresponding biased square codestate, reducing its susceptibility to loss and dephasing.

One may wonder if the on-demand biasing circuits that we propose can be improved upon to provide output states with better squeezing parameters. In particular, given an input state $\ket{\bar{\psi}}_{\Delta}$, what output state $\ket{\bar{\psi}}_{\Delta_{q},\Delta_{p}}$ is achievable when using Gaussian operations? In \cref{appendix: proof nogo theorem}, we prove that, when limited to circuits using any number of approximate GKP ancillas $\ket{\bar{\psi}}_{\Delta}$ and any deterministic Gaussian operations, the output GKP qubit will always have a bias that satisfies
\begin{equation}
    \label{eq: trade-off relation}
    \Delta_q \Delta_p \geq \Delta^2.
\end{equation}
In particular, the on-demand biasing circuits~\cref{eq: biased noise q-Steane,eq: biased noise q-Steane with rect ancilla} already saturate this bound, demonstrating that these circuits are indeed optimal among deterministic Gaussian circuits.

\section{Optimized logical gate representations}
\label{section: optimized logical gate representations}
In the previous section, we introduced the on-demand noise-biasing protocol, which allows us to generate bias that mitigates the effect of propagated errors induced by the cubic phase $T$ gate. In this section, we move beyond the cubic phase gate and develop a general formalism for characterizing polynomial phase representations of logical gates, enabling the identification of representations that minimize error propagation when acting on approximate GKP states. Central to this framework is the introduction of \textit{polynomial phase stabilizers} that preserve the GKP codestates but are not displacement operators. We show that these generalized stabilizers impose strong algebraic constraints on the allowed polynomial phase gates of logical gates. Finally, we emphasize that the formalism is broadly applicable: it extends naturally to other bosonic codes and can be generalized to multi-qubit GKP logical gates, as discussed in \cref{appendix: optimal representation of multi-qubit gates}.

\subsection{Polynomial phase stabilizers}

In this section, we introduce the framework of polynomial phase stabilizers, which will lead to a systematic method to construct alternative representations of logical gates in the GKP code that minimize error propagation.

As a motivating example, consider the following operator
\begin{equation}
    O=\exp\left\{ i 2\pi \Big[-\frac{1}{6}\Big( \frac{ q}{\sqrt{\pi}} \Big)^3 + \frac{1}{6}\Big( \frac{ q}{\sqrt{\pi}} \Big)  \Big] \right\}.
\end{equation}
How does this operator act on square GKP codestates? As before, we can write $O=\mathrm{exp}\big(i2\pi P({q}/\sqrt{\pi})\big)$ where $P(x)=-x^{3}/6+x/6$ is a cubic polynomial. Moreover, note that $P(n)$ is an integer for every input integer $n$. This means that not only does $O$ perserve the GKP codespace, but it also acts as the logical identity on that codespace - it is a cubic phase \textit{stabilizer}.

Because $O$ acts trivially on the GKP codespace it can be multiplied by the $T_{\text{GKP}}$ gate to obtain a new cubic phase gate $T_{3}$ that still implements a logical $T$  gate:
\begin{subequations}
    \label{eq: T3 gate}
\begin{align}
    T_{3}&=T_{\text{GKP}}{O}\\   &=\exp\bigg\{i2\pi \Big[ \frac{1}{12} \!\left(\frac{{q} }{\sqrt{\pi}}\right)^3 \!
    + \frac{1}{8}\! \left(\frac{{q} }{\sqrt{\pi}}\right)^2 \! - \frac{1}{12} \!\left(\frac{{q} }{\sqrt{\pi}}\right)   \Big] \bigg\}. 
\end{align}
\end{subequations}
This same representation was previously found in Ref.~\cite{Konno2021} through numerical optimization and shown to achieve significantly higher gate fidelity than the original $T_{\rm GKP }$ gate. This improvement can be attributed to the reduced amplitude of the cubic (non-Gaussian) term, which reduces the magnitude of the induced shear errors discussed in Section~\ref{section: Error propagation of cubic phase gate}. Our approach here however suggests a more systematic general framework to obtain minimal representations of polynomial phase gates by characterizing the corresponding polynomial phase stabilizers, as we now explain.

Specifically, polynomial phase stabilizers are operators that act trivially on the GKP codespace and can be written as the exponential of a polynomial 
\begin{equation}
\label{eq: stabilizer canonical form}
    {S}_i  = \exp\Big[ i 2\pi P({x} )  \Big].
\end{equation}
Here, $P(x)$ is a real polynomial that generates the stabilizer ${S}_i$, and ${x}$ is \textit{any} operator such that the linear stabilizer $\exp(i 2 \pi {x})$ acts trivially on the ideal GKP states \cite{Royer2022}. Importantly, the linear stabilizer guarantees that the GKP codestates are superpositions of eigenstates of ${x}$ at integer eigenvalues. For the present discussion, we fix ${x}={q}/\sqrt{\pi}$, corresponding to a scaled position operator, such that $\exp(i 2\pi x)=\exp(i2\sqrt{\pi}q)=S_{q}$ is one of the usual GKP stabilizers. However, we will later show that other choices for ${x}$, such as ${a^\dag a}/2$, also yield valid stabilizer structures and reveal additional structure in the code space. 

Since the ideal GKP states have support only on eigenstates $\ket{m\sqrt{\pi}}_q$ where $m \in \mathbb{Z}$, it follows that the polynomial $P(x)$ must return integer values when evaluated at integer arguments $x$. In other words, we require 
\begin{equation}
    \label{eq: definition of integer value poly}
    P( \mathbb{Z}) \subseteq  \mathbb{Z}. 
\end{equation}
This class of polynomials is well known in mathematics as integer-valued polynomials. Importantly, P\'olya \cite{Polya1915} found that there is a canonical basis for the set of integer-valued polynomials, namely 
\begin{equation}
    \label{eq: trivial polynomial stabilizers}
    L_n(x) = \begin{cases}
        \frac{1}{n!}\prod_{i=1}^{n}(x+i-\frac{n}{2})  & \text{ for }n\text{ even} \\
        \frac{1}{n!}\prod_{i=1}^{n}(x+i-\frac{n+1}{2}) & \text{ for }n\text{ odd} 
    \end{cases} \quad .
\end{equation}
(see also Ref.~\cite{SINGMASTER1974} for a more pedagogical introduction). All integer-valued polynomials $P(x)$ can be written as an integer linear combination of the basis polynomials $\{ L_n(x) \}$; that is,
\begin{equation}
    P(x)=\sum_{n=1}^{m}c_{n}L_{n}(x),
\end{equation}
for some integer coefficients $c_{n}$.

Since the construction of polynomial phase stabilizers relies solely on the fact that the square GKP code has support on a discrete position-space grid, the same approach can be readily extended to other sets of grid states. In particular, we can generalize the operator ${x}$ by replacing it with any of the scaled operators $\{{q},{p},{q}-{p} \}/\sqrt{\pi}$ which correspond respectively to the logical ${Z}$, ${X}$, and ${Y}$ bases. Furthermore, we observe that the square GKP code has support only on even-number Fock states and is therefore stabilized by $\exp(i\pi a^\dag a)$, and so we can also set ${x}$ to ${{a^\dag a}}/2$. More generally, this observation implies that the concept of polynomial phase stabilizers can be naturally extended to the broader class of bosonic codes \cite{albert_2025} such as the multi-leg rotationally symmetric code \cite{Grimsmo2020} or the number-phase code \cite{hu2025generalized}. For instance, we can choose the operator ${x}$ to be ${a^\dag a}/\alpha$ in the multi-leg rotationally symmetric code where $\alpha$ is the number of legs.

\subsection{Minimal logical gate representations}

In this section, we consider the family of polynomial phase gates defined by
\begin{equation}\label{eq: logical canonical form}
    \Lambda_m[{x}] = \exp\Big(i 2\pi P_{\Lambda_{m}}({x}) \Big),
\end{equation}
in which its actions on the eigenstates $\ket{x}$ ($x \in \mathbb{Z}$) of the operator ${x}$ is given by 
\begin{equation}\label{eq: phase gate definition}
    \Lambda_m[{x}] \ket{x} = \begin{cases}
         \ket{x} \quad  & \text{for } x \equiv 0 \mod 2\\
         \exp\Big(i 2\pi/2^m \Big) \ket{x} & \text{for } x \equiv 1 \mod 2
    \end{cases}. 
\end{equation}
When the operator ${x}$ is chosen to be $ {q}/\sqrt{\pi}$, the gate family $\Lambda_m[{x}]$ encompasses the logical diagonal gates $\{ Z,S,T,\sqrt{T},\dots \}$ of the GKP code. More generally, by selecting alternative choices for ${x}$, the family $\Lambda_m[{x}]$ also captures other important gate families. For example, by setting ${x}= {a^\dag a}/2$, we yield the Hadamard gate family $\{ H,\sqrt{H},\sqrt[4]{H}, \dots \}$ \cite{Royer2022}.

In what follows, we will detail an algorithm that ``minimizes'' the representation of a particular polynomial phase gate. In particular, the minimization is with respect to a ``lexicographic'' ordering on the polynomials (technically this is a partial order). In the lexicographic ordering, given two real polynomials $P(x)=\sum_{n}c_{n}x^{n}$ and $Q(x)=\sum_{n}d_{n}x^{n}$, we compare the coefficients of the polynomials $P(x)$ and $Q(x)$ starting from the leading coefficient to lowest order coefficient. We define the polynomial $P(x)$ to be less than $Q(x)$ if
% if there exists an integer $i$ such that $|c_{i}|<|d_{i}|$ and $|c_{k}|=|d_{k}|$ for all $k>i$.
\begin{multline}
    P(x) \prec Q(x)
 \Leftrightarrow
\text{there exists } i\text{ such that }\\
|c_k| = |d_k|\text{ for all } k > i,\text{ and }
|c_i| < |d_i|.
\end{multline}

The advantage of the lexicographic ordering is that it minimizes the coefficients of the higher degree terms in the polynomial before moving to the lower degree terms. Intuitively, this is important because the higher degree terms in the polynomial contribute to higher degree errors that affect GKP codestates more severely. For this reason, as $\Delta\rightarrow0$, we expect a minimal polynomial phase gate to perform optimally among the polynomial phase gates that implement the same logical operator (although see footnote \footnote{It is worth noting that in principle the lexicographic ordering does not necessarily guarantee optimal performance of the minimal polynomial phase gate across \textit{all} values of $\Delta$. This could happen for example if a polynomial $P(x)$ has a smaller leading coefficient than $Q(x)$ and thus $P(x)\prec Q(x)$, but $P(x)$ has larger coefficients than $Q(x)$ in other positions. We did not however find any specific examples of this for two polynomials that implement \textit{the same} logical gate, but we do find it for some polynomials that implement different logical gates, see for example $T_{3}$ and $\sqrt{T}$ in \cref{sec: numerics}.}).

The algorithm that we describe will output a polynomial phase gate whose corresponding polynomial is not larger under the lexicographic ordering than any other polynomial that implements the same gate. The algorithm consists of two steps: first, we write down any ``starting'' representation of the ${\Lambda}_{m}[x]$ gate, and then we use the polynomial phase stabilizers to minimize the coefficients of the gate under the lexicographic order (we call this second step the ``coefficient reduction'' procedure). We show in Appendix~\ref{appendix: representation logical phase gate} that every $\Lambda_m[x]$ gate admits a “starting” representation with polynomial
\begin{equation}\label{eq: trivial polynomial representation}
  P_{\text{start}}(x) = \frac{x^{2^{m-1}}}{2^{m}}. 
\end{equation}
Alternatively, if a good representation $F(x)$ of the ${\Lambda}_{m-1}[x]$ gate is already known, one can choose a better starting point
\begin{equation}\label{eq: better polynomial starting point}
    P_{\text{start}}(x) = 2^{m-2} F(x)^2.
\end{equation}
In either case, the candidate polynomial for $P_{\text{start}}(x)$ can be expressed as
\begin{equation}\label{eq: polynomial start coefficients}
    P_{\text{start}}(x) = \sum_{i=1}^{N} a_i x^i.
\end{equation}

The goal of the coefficient reduction procedure is to systematically adjust the coefficients of $P_{\rm start}(x)$ by adding or subtracting suitable integer multiple of the basis polynomials $\{ L_n(x) \}$, starting from the highest-degree term and proceeding downward. Importantly, the leading coefficient of $L_{n}(x)$ is $1/n!$, which means that we can reduce each coefficient to be between $-1/(2n!)$ and $1/(2n!)$. Specifically, let $j$ denote a pointer, initialized at $j=N$ and decreased sequentially until $j=1$, and initialize $P(x)=P_{\text{start}}(x)$. At each step, we decompose the coefficient $a_j$ of the polynomial~\cref{eq: polynomial start coefficients} as 
\begin{equation}
a_j = \frac{n_j}{j!} + r_j,
\end{equation}
where $n_{j} \in \mathbb{Z}$ and the remainder satisfies $|r_j| \leq \frac{1}{2 j!}$.
We then update the polynomial $P(x)$ via  
\begin{equation}
P(x) \leftarrow P(x)-n_j L_j(x),
\end{equation}
where $L_j(x)$ is the degree-$j$ integer-valued polynomial defined in Eq.~\eqref{eq: trivial polynomial stabilizers}. This step removes the part of each coefficient that is an integer multiple of $1/j!$. Note that $n_{j}$ and $r_{j}$ are unique here \textit{unless} $r_{j}=\pm\frac{1}{2j!}$, in which case there are two choices of $n_{j}$. In this case, we proceed with both choices of $n_{j}$ until one of the choices is lexicographically smaller than the other, at which point we can discard the other choice of $n_{j}$. Iterating this procedure from  $j=N$ down to $j=1$ produces a reduced polynomial representation of the form 
\begin{equation}
P(x) = \sum_{k=0}^m a_k x^k,
\end{equation}
with coefficients satisfying  bound
\begin{equation}
|a_k| \leq \frac{1}{2 k!}, \qquad \forall k.
\end{equation}

In \cref{appendix: representation logical phase gate} we present multiple results about the coefficient reduction procedure. In Lemma \ref{lemma: guarantee of coefficient reduction procedure}, we formally prove that the above algorithm indeed returns a minimal representation under the lexicographic ordering. Moreover, we prove in \cref{theorem: minimal degree polynomial} that the degree of the polynomial that is returned is exactly $m$, and therefore that the lowest degree polynomial that can implement the $\Lambda_{m}[x]$ gate is degree $m$. While this result is well-known for diagonal GKP gates specifically, our theorem applies to all polynomial phase gates of the form described here (including the Hadamard gate hierarchy).

As an example, let us construct the minimal representation of the $\sqrt{T}$ gate. Our starting point is obtained by substituting the minimal $T$ gate into \cref{eq: better polynomial starting point}, giving
\begin{subequations}
\begin{align}
    P_{\Lambda_4}(x)&\leftarrow P_{\text{start}}(x)=4 P_{\Lambda_3}(x)^{2}\\
    &=\frac{x^6}{36}+\frac{x^5}{12}+\frac{x^4}{144}-\frac{x^3}{12}+\frac{x^2}{36}.
\end{align}
\end{subequations}
Now, the leading coefficient $1/36$ is equal to $20$ times $1/6!$, and therefore we update
\begin{subequations}
\begin{align}
    P_{\Lambda_{4}}(x)&\leftarrow P_{\Lambda_{4}}(x)-20 L_{6}(x)\\
    &=\frac{7x^4}{48}+\frac{x^3}{3}-\frac{x^2}{12}-\frac{x}{3}.
\end{align}
\end{subequations}
Notice that this update has coincidentally reduced the 5th order coefficient to zero. Moving on to the 4th order coefficient, there are actually two choices for the coefficient reduction of $7x^{4}/48$ since it can be written in two ways: $7/48=3\times 1/4! +1/48=4\times 1/4!-1/48$. After running the full algorithm it turns out that the second choice returns a smaller polynomial under the lexicographic ordering, so we update
\begin{subequations}\label{eq: lambda_4 polynomial}
\begin{align}
    P_{\Lambda_{4}}(x)&\leftarrow P_{\Lambda_{4}}(x)-4 L_{4}(x)\\
    &=-\frac{x^4}{48}+\frac{x^2}{12}.
\end{align}
\end{subequations}
This polynomial can no longer be reduced by any of $L_{3}(x)$, $L_{2}(x)$ or $L_{1}(x)$, and therefore we are done.

\cref{eq: lambda_4 polynomial} gives rise not only to a minimal representation of the $\sqrt{T}$ gate
\begin{equation}\label{eq: sqrt_T gate}
    \sqrt{T}= \exp\left\{i2\pi \Big[- \frac{1}{48} \left(\frac{{q} }{\sqrt{\pi}}\right)^4 + \frac{1}{12} \left(\frac{{q} }{\sqrt{\pi}}\right)^2   \Big] \right\},
\end{equation}
but also to a minimal representation of the $\sqrt[8]{H}$ gate
\begin{equation}\label{eq: 8_rt_H gate}
    \sqrt[8]{H} = \exp\left\{i2\pi \Big[- \frac{1}{48} \left(\frac{a^{\dag}a }{2}\right)^4 + \frac{1}{12} \left(\frac{a^{\dag}a }{2}\right)^2   \Big] \right\}. 
\end{equation}
It is important to note, however, that fault-tolerance is only guaranteed for polynomial phase gates of a quadrature operator such as \cref{eq: sqrt_T gate} and not for \cref{eq: 8_rt_H gate} since it is unclear how one would bias the noise in the latter case.

Moreover, taking the square of \cref{eq: sqrt_T gate} gives a quartic representation of the $T$ gate that we label $T_{4}$:
\begin{equation}\label{eq: T_4 gate}
    T_{4} = \exp\left\{i2\pi \Big[- \frac{1}{24} \left(\frac{{q} }{\sqrt{\pi}}\right)^4 + \frac{1}{6} \left(\frac{{q} }{\sqrt{\pi}}\right)^2   \Big] \right\}. 
\end{equation}
Although \cref{eq: T_4 gate} is not \textit{minimal}, it has the advantage over the minimal $T_{3}$ gate \cref{eq: T3 gate} that the polynomial that implements it is even, making it more convenient in some implementations. Indeed, in Ref.~\cite{kolesnikow2025protected} it was shown that logical phase gates of the GKP code can be implemented protectively in 0-$\pi$ qubits provided the underlying Hamiltonians involve only quadratic and quartic potentials.
This example demonstrates that we can use the polynomial phase stabilizers to find logical gates tailored to hardware constraints, while preserving logical equivalence. 

Although we have so far focused on single-qubit gates in this section, the same polynomial-stabilizer framework naturally extends to multi-qubit gates between square GKP qubits. In Appendix~\ref{appendix: optimal representation of multi-qubit gates}, we demonstrate how the coefficient reduction procedure can be generalized to construct minimal representations of such multi-qubit gates. For example, we find the following minimal representations of the $CS$ and $CCZ$ gates for the square GKP qubits 
\begin{subequations}
\begin{align}
CS
 &= \mathrm{exp}\bigg\{-\frac{i\pi}{2}\bigg[\Big(\frac{{q}_1 }{\sqrt{\pi}}\Big)^2 \Big(\frac{{q}_2 }{\sqrt{\pi}}\Big) + \Big(\frac{{q}_1 }{\sqrt{\pi}}\Big) \Big(\frac{{q}_2 }{\sqrt{\pi}}\Big)^2\nonumber\\
 &\qquad\qquad\qquad\qquad\qquad+ \Big(\frac{{q}_1 }{\sqrt{\pi}}\Big) \Big(\frac{{q}_2 }{\sqrt{\pi}}\Big)\bigg]\bigg\}\\
 CCZ&=\mathrm{exp}\Big(i{q}_{1}{q}_{2}{q}_{3}/\sqrt{\pi}\Big),\label{eq: CCZ gate}
\end{align}
\end{subequations}
acting on two and three modes respectively.

\section{Asymptotic Analysis}
\label{sec: error propagation and fault-tolerant theorem}

In the previous sections, we showed how to construct minimal representations of logical gates, and we showed how to perform on-demand noise biasing. We now turn to analyzing the error rates of these gates when acting on approximate biased square GKP states $\ket{\bar{\psi}}_{\Delta_{q},\Delta_{p}}$. Throughout this analysis, we will assume that ${x}={q}/\sqrt{\pi}$. In this section, we present our analytic results concerning the behaviour as $\Delta_{q},\Delta_{p}\rightarrow0$, with our numerical results explained later in \cref{sec: numerics}. In particular, we take two separate analytical approaches to analyze the fault-tolerance of polynomial phase gates. The first is to use the twirling approximation to obtain an asymptotic estimate of the error scaling in the limit of small $\Delta$. This estimate however uses a number of approximations as opposed to strict lower-bounds, and as such we also rigorously prove a fault-tolerance theorem for the minimal ${T}_{3}$ gate Eq.~\eqref{eq: T3 gate}. 

We begin, however, by explaining the error that we wish to characterize analytically. Given a polynomial phase gate $\Lambda$, we define the error operator as
\begin{equation}\label{eq: error operator}
    {E}=\Lambda \,\text{Env}_{\Delta_{q},\Delta_{p}}\,\Lambda^{\dag}.
\end{equation}
To understand ${E}$, consider applying each of its terms to an \textit{ideal} GKP codestate $\ket{\bar{\psi}}$. From right to left, the first term $\Lambda^{\dag}$ acts ideally because it is acting on an ideal GKP codestate. Then, the envelope operator prepares an approximate GKP codestate with logical information given by $\Lambda^{\dag}\ket{\psi}$. Finally, the last term $\Lambda$ acts non-ideally because it is applied to an approximate GKP codestate. Importantly, because we acted initially with $\Lambda^{\dag}$, the ideal \textit{logical} action of ${E}$ is the identity, which can be recovered by setting $\Delta_{q}=\Delta_{p}=0$. Therefore, ${E}$ captures the effect of spreading the envelope operator through the gate $\Lambda$ while avoiding the logical action of $\Lambda$.

\begin{figure*}
\centering
    \includegraphics{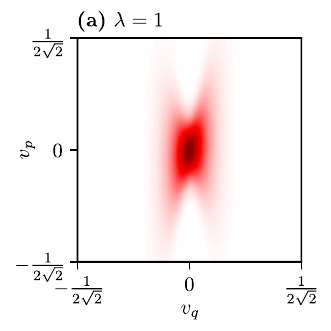}\includegraphics{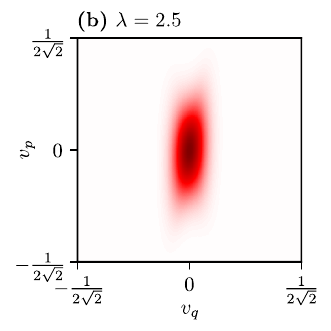}\includegraphics{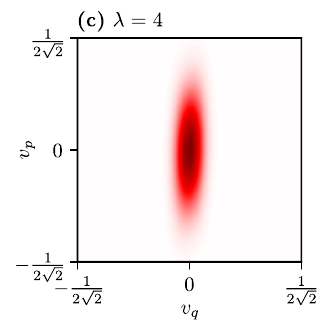}
\caption{A visual representation of twirled channel $\mathrm{Twirl}(\mathcal{E})$ for the optimized $T_{3}$ gate~\cref{eq: T3 gate} and for $\Delta=\sqrt{\Delta_{q}\Delta_{p}}=0.2$. Specifically, we plot the probability distribution $p_{\mathcal{E}}(\vect{v})$ of the random displacement function that is obtained by twirling the error operator $E$ as in~\cref{eq: twirling definition}. While $\vect{v}=(v_{q},v_{p})$ can take any value in $\mathbb{R}^{2}$, we show here only the correctable displacements. Moving from (a) to (c) we increase the noise bias $\lambda=\Delta_{p}/\Delta_{q}$, where (b) is the optimal noise biasing for this value of $\Delta$ as found by numerics in \cref{fig:T_gates}(a). When not enough noise biasing is applied as in (a), the noise from the $T$ gate dominates and increases the noise, while when too much noise biasing is applied as in (c) the noise from the biasing itself dominates. At intermediate values these two effects balance out so that the overall noise rate is minimal.}\label{fig: twirled_chi}
\end{figure*}

In our asymptotic analysis, we will consider the channel $\mathcal{E}(\cdot)={E}\cdot {E}^{\dag}$ (see footnote~\footnote{We use the term ``channel'' loosely here because $\mathcal{E}$ is not trace-preserving.}) and calculate the characteristic function $p_{\mathcal{E}}(\vect{v})$ of the twirled channel $\mathrm{Twirl}(\mathcal{E})$. Recall that this twirling approximation approximates $E$ as a random displacement channel, where the displacements are sampled from the (non-Gaussian) probability distribution $p_{\mathcal{E}}(\vect{v})$. This makes it easy to analyze in the context of GKP codes: if the displacement sampled is correctable as defined by \cref{eq: correctable displacement}, then no logical error occurs. Therefore to obtain an asymptotic estimate of the error rate, we calculate the mean and covariance matrix of the probability distribution $p_{\mathcal{E}}(\vect{v})$, and how these parameters scale with $\Delta_{q}$ and $\Delta_{p}$. In particular, if the mean and covariance matrix are small, there is a large chance that the sampled displacement will satisfy \cref{eq: correctable displacement} and be correctable.

Specifically, let the logical operator have the form
\begin{align}\label{eq: logical operator definition}
    \Lambda&=\mathrm{exp}\big[i2\pi P(q/\sqrt{\pi})\big],&P(x)&=\sum_{j=1}^{n}a_{j}x^{j},
\end{align}
where $n>1$ is the degree of the polynomial $P(x)$. Then, using the characteristic function expansion of the envelope operator $\chi_{\text{Env}}(\vect{v})$ from~\cref{eq: biased envelope characteristic function}, one can explicitly write the error operator as
\begin{subequations}\label{eq: main-text propagated error operator}
\begin{align}
        {E} &= \Lambda\,\mathrm{Env}_{\Delta_{q},\Delta_{p}}\,\Lambda^{\dag}\label{eq: E definition}\\
        &=\int d\vect{v}\,\chi_{\mathrm{Env}}(\vect{v})\,e^{i2\pi P(q/\sqrt{\pi})}W(\vect{v})\, e^{-i2\pi P(q/\sqrt{\pi})}\\
        &=\int d \bo v \,\chi_{\mathrm{Env}}(\vect{v})\,W(\bo v)\,e^{i 2 \pi \big[ P({q}/\sqrt{\pi}+\sqrt{2} v_q )-P({q}/\sqrt{\pi} )\big]}\label{eq: main-text propagated error operator exact}\\
        &\approx\int d \bo v\, \chi_{\mathrm{Env}}(\vect{v})\,W(\bo v)\,e^{i 2 \sqrt{2}\pi v_{q} P'({q}/\sqrt{\pi})},\label{eq: derivative approximation}
\end{align}
\end{subequations}
where the first-order derivative approximation in \cref{eq: derivative approximation} is valid for small $v_{q}$ because $\chi_{\text{Env}}(\vect{v})$ decays exponentially in $v_{q}$, and is made to simplify the calculations.

Next we apply the twirling approximation and calculate the characteristic function of the twirled channel $\mathrm{Twirl}(\mathcal{E})$, see \cref{fig: twirled_chi} for a visual representation of the twirled channel. In \cref{appendix: statistical moment of errors}, we exactly calculate the mean and covariance matrix of the resulting distribution $p_{\mathcal{E}}(\vect{v})$ as a function of $\Delta_{q}$ and $\Delta_{p}$. In particular, the mean $\vect{\mu}$ of the distribution is zero:
\begin{equation}\label{eq: mean scaling}
    \vect{\mu}=\begin{bmatrix}\mathbb{E}(v_{q})\\\mathbb{E}(v_{p})\end{bmatrix}=\vect{0},
\end{equation}
where $\mathbb{E}()$ denotes the expectation value of a variable. Moreover, the terms in the covariance matrix are given by
\begin{subequations}\label{eq: covariance matrix scaling}
\begin{align}
    \Sigma&=\begin{bmatrix}\mathbb{E}(v_{q}^{2})&\mathbb{E}(v_{q}v_{p})\\\mathbb{E}(v_{q}v_{p})&\mathbb{E}(v_{p}^{2})\end{bmatrix},\\
    \mathbb{E}(v_q^2) &= O(\Delta_q^2),\\
         \label{eq: variance vp leading contribution}
     \mathbb{E}(v_p^2) &= O(\Delta_p^2)+O(\Delta_q^2/\Delta_{p}^{2n-4}),\\
     \mathbb{E}(v_{q}v_{p}) &= O(\Delta_q^{2}/\Delta_{p}^{n-2}),
\end{align}
\end{subequations}
where the $O()$ notation is valid in the limit of $\Delta_{q},\Delta_{p}\rightarrow0$. For comparison, the bare envelope operator has $\mathbb{E}(v_q^2) = O(\Delta_q^2)$, $\mathbb{E}(v_p^2) = O(\Delta_p^2)$ and $\mathbb{E}(v_{q}v_{p})=0$. Therefore, the logical gate $\Lambda$ has the effect of adding the additional $n$-dependent terms in \cref{eq: covariance matrix scaling}. In the case of no biasing $\Delta_{q}=\Delta_{p}=\Delta$, the variance $\mathbb{E}(v_{p}^{2})$ does not go to 0 in the limit of small $\Delta$ for $n\geq 3$, necessitating the use of noise biasing to achieve arbitrarily low error rates.

At this point, we now wish to choose a noise biasing $\lambda=\Delta_{p}/\Delta_{q}$ that minimizes the components of the covariance matrix $\Sigma$. For our purposes, we will simply choose $\lambda$ to minimize the most problematic component of the covariance matrix, $\mathbb{E}(v_{p}^{2})$. Indeed, choosing the correct $\lambda$ is important because \cref{eq: variance vp leading contribution} contains both a $O(\Delta_{p}^{2})$ term from the envelope that will increase with increased biasing, and a $O(\Delta_{q}^{2}/\Delta_{p}^{2n-4})$ term from the $\Lambda$ gate that will decrease with increased biasing. We illustrate this trade-off visually in \cref{fig: twirled_chi}, which shows how the probability distribution $p_{\mathcal{E}}(\text{v})$ varies with $\lambda$. Letting $\Delta=\sqrt{\Delta_{q}\Delta_{p}}$, it is straightforward to show that the optimal $\lambda$ scales as
\begin{equation}
    \lambda_{\text{opt}}=O(\Delta^{4/n-2}),
\end{equation}
which scales \textit{inversely} with $\Delta$ for $n\geq 3$. With this choice, the remaining terms in the covariance matrix become
\begin{subequations}\label{eq: optimal lambda covariance matrix scaling}
\begin{align}
    \mathbb{E}(v_{q}^{2})&=O(\Delta^{4-4/n}),\\
    \mathbb{E}(v_{p}^{2})&=O(\Delta^{4/n}),\\
    \mathbb{E}(v_{q}v_{p})&=O(\Delta^{2}).
\end{align}
\end{subequations}
Now, with this optimized choice $\lambda_{\rm opt}$, we see that all the components of the covariance matrix go to zero in the limit of small $\Delta$, as desired (see also footnote \footnote{Technically, what we need to ensure is that the \textit{eigenvalues} of the covariance matrix go to zero as $\Delta$ goes to zero, because it is the maximum eigenvalue that determines the maximum ``spread'' of the probability distribution. This is guaranteed, however, because both the trace and determinant of $\Sigma$ are positive and go to zero as $\Delta\rightarrow0$.}). Note that the \textit{quantitative} estimate of the optimal biasing $\lambda_{\text{opt}}$ is not the same as the numerical estimates we obtain later in \cref{sec: numerics}, but should reflect the correct scaling with $\Delta$. 

\cref{eq: optimal lambda covariance matrix scaling} also justifies our choice of the lexicographic ordering in \cref{section: optimized logical gate representations}. In particular, the performance of the polynomial phase gate $\Lambda$ in the limit of small $\Delta$ is governed by the degree $n$ of the polynomial. Moreover, in \cref{appendix: statistical moment of errors} we show that the elements of the covariance matrix also depend on the leading coefficient $|a_{n}|$ of the polynomial. To give a concrete example, the variance $\mathbb{E}(v_{p}^{2})$ is nine times larger for the $T_{\text{GKP}}$ logical gate~\cref{eq: T_GKP polynomial} than for the optimized $T_{3}$ gate~\cref{eq: T3 gate}. Therefore, to improve the asymptotic performance, the most important thing to reduce is the degree of the polynomial $n$, and the magnitude of the leading coefficient $|a_{n}|$. We therefore expect that for two polynomials $P(x)$ and $Q(x)$ that implement the same logical gate, if one polynomial $P(x)$ is less than $Q(x)$ under the lexicographic ordering, then for sufficiently small $\Delta$, the polynomial phase gate of $P(x)$ will outperform that of $Q(x)$. Reducing the magnitude of the non-leading coefficients $\{a_{j}\}_{j=1}^{n-1}$ of $P(x)$ can have an effect for fixed larger values of $\Delta$, but this effect will become small for sufficiently small $\Delta$.

While \cref{eq: optimal lambda covariance matrix scaling} provides strong evidence of the fault-tolerance of polynomial phase gates with an appropriate bias, it falls short of a rigorous proof. Not only did we use the twirling approximation to perform the calculations, but we also only calculated the covariance matrix of the corresponding random displacement channel rather than explicitly calculating the error rate of the gate which also depends on higher statistical moments.

Therefore, in \cref{appendix: proof of fault-tolerance theorem} we also rigorously prove that for the optimized $T_{3}$ gate~\cref{eq: T3 gate}, the ``logical'' average gate fidelity of the error operator (precisely defined in \cref{appendix: proof of fault-tolerance theorem}) goes to one as $\Delta\rightarrow0$. This is mathematically equivalent to saying that, for any target fidelity $F_{\text{targ}}$, there exists a $\Delta_{\text{targ}}$ such that for any $\Delta<\Delta_{\text{targ}}$, there exists a choice of noise biasing $\lambda_{\text{opt}}(\Delta)$ such that the fidelity $F$ of the error operator $E$ is greater than $F_{\text{targ}}$. We expect that similar results could also be proven for higher order polynomial phase gates, but due to the somewhat tedious nature of the proof we leave this to future work.

\section{Numerical Analysis}\label{sec: numerics}

To complement our asymptotic analysis of the performance of the logical gates in \cref{sec: error propagation and fault-tolerant theorem}, here we numerically assess the performance of the logical gates in a more realistic parameter regime. The aim of this section is to provide a proof-of-principle that implementing non-Clifford gates unitarily in the GKP code can have good performance even for realistic values of $\Delta$. To provide such a proof-of-principle, we will use the assumption that the only source of noise is the approximate GKP codestates themselves; we do not attempt to consider all experimentally-relevant sources of noise in our simulations and leave such investigations to future work.

Throughout this section we express the quality of the approximate codestates $\Delta=\sqrt{\Delta_{q}\Delta_{p}}$ in two different ways: first in terms of the average photon number $\bar{n}\approx 1/(2\Delta^{2})-1/2$~\cite{Albert18}, and second in terms of the GKP squeezing $\Delta_{\mathrm{dB}}=-10\log_{10}(\Delta^{2})$; note that both $\bar{n}$ and $\Delta_{\mathrm{dB}}$ increase as $\Delta$ decreases. Also note that the average photon number $\bar{n}$ that we use corresponds to a \textit{non-biased} square (or rectangular) GKP state, a biased square codestate would have a higher average photon number than listed here. For comparison, recent experiments in superconducting qubits used non-biased square codestates with $\Delta\approx0.32$ ($\bar{n}\approx4.4$, $\Delta_{\mathrm{dB}}\approx9.9$)~\cite{brock2025quantum}, while experiments in optics~\cite{Larsen2025} only achieved $\Delta \approx 0.93$ ($\bar{n}\approx0.074$, $\Delta_{\mathrm{dB}}\approx0.6$). The achieved squeezing in superconducting devices of $\Delta=0.32$ is exactly the same as the theoretically predicted threshold of the surface-GKP code when only considering noise from approximate GKP codestates~\cite{Noh2022}; however, the threshold is expected to be lower in the presence of additional experimentally-relevant sources of noise.

In \cref{sec: mid-circuit numerics} we benchmark a variety of polynomial phase gates including those that implement logical gates up to the sixth level of the Clifford hierarchy. Taking $\Delta=0.25$ ($\bar{n}=7.5$, $\Delta_{\mathrm{dB}}=12$) as an example data point, all of the optimized polynomial phase gates achieve infidelities of less than 1\% with moderate asymmetries of $\lambda\leq2$. Moreover, in \cref{sec: state numerics} we find that polynomial phase gates can outperform the vacuum-state method~\cite{Baragiola2019} for preparing $\ket{T}$ magic states, requiring a low postselection rate of approximately 20 \% for the vacuum-state method to match the performance of the minimal cubic phase gate at the same $\Delta=0.25$ data point. This comparison also underestimates the noise in the vacuum-state method: it ignores the noise that would arise from Clifford gates in the $\ket{T}$ state injection circuit itself, which could be avoided in the polynomial phase gate method by directly applying a ${T}$ gate on the desired GKP qubit.

\subsection{On-Demand Polynomial Phase Gates}\label{sec: mid-circuit numerics}

We are interested in benchmarking the performance of polynomial phase gates in two contexts, the first of which is ``on-demand'' at an arbitrary location in an algorithm. Specifically, we want to model polynomial phase gates (for simplicity here represented by a $T$ gate) in the square GKP code that are preceded by a round of on-demand biasing using the $q$-Steane QEC circuit
\begin{equation}\tag{\ref{eq: T gate with on-demand biasing}}
    \begin{quantikz}
        \lstick{$\cdots$}&\gate{\substack{q\text{-Steane}\\\text{QEC}}}\gategroup[1,steps=1,style={dashed,rounded corners,fill=blue!20,inner sep =0 pt},background,label style={label position=above,anchor=south,yshift=-0.2cm}]{$\substack{\text{On-demand}\\\text{biasing}}$}&\gate{T}&\gate{\substack{\text{Passive}\\\text{Knill}\\\text{QEC}}}&\rstick{$\cdots$}
    \end{quantikz},
\end{equation}
as explained in \cref{section: mid-circuit noise biasing protocol}. We assume that noise arises in this circuit only due to the envelope operator, which appears in \cref{eq: T gate with on-demand biasing} in the $q$-Steane and passive Knill QEC circuits~\cref{eq: q-Steane circuit,eq: passive Knill} but not in the polynomial phase gate itself. This noise model therefore captures how the noise from the envelope operator is propagated through logical gates, but does not explicitly include some other experimentally-relevant noise sources such as loss, dephasing, and imperfect gate implementation.

Intuitively, when performing a round of QEC with approximate GKP ancilla codestates, the approximate GKP codestates introduce noise in two ways. First, it causes noise on the stabilizer measurement outcomes (``measurement'' noise), and second, it adds envelope noise to the output state (``incoming'' noise). In fact, this intuition can be justified theoretically: if one applies the twirling approximation to the envelope operator to obtain a Gaussian random displacement channel, then the noise from approximate codestates in the passive-Knill circuit \cref{eq: passive Knill} exactly corresponds to an ideal QEC cycle preceded \textit{and} followed by two independent Gaussian random displacement channels representing the measurement error and envelope noise discussed above~\cite{Rozpedek23,shaw2024logical}. A similar result holds for the $q$-Steane QEC circuit but with biased noise following the gate. Our noise model is therefore analogous to the ``phenomenological'' noise model commonly used in the DV QEC literature. With this, \cref{eq: T gate with on-demand biasing} becomes
\begin{equation}\label{eq: circuit C in context}
    \begin{quantikz}
        \lstick{${\cdot}{\cdot}{\cdot}\!$}&[-0.45cm]\gate{\substack{\text{non-}\\\text{biased}\\\text{noise}}}&[-0.35cm]\gate{\substack{\text{Ideal}\\\text{QEC}}}\gategroup[1,steps=5,style={dashed,rounded corners,fill=red!20,inner sep=-1.5 pt},background,label style={label position=above,anchor=south,yshift=-0.2cm}]{$\mathcal{C}$}&[-0.35cm]\gate{\substack{\text{biased}\\\text{noise}}}&[-0.2 cm]\gate{T}&[-0.2 cm]\gate{\substack{\text{non-}\\\text{biased}\\\text{noise}}}&[-0.35cm]\gate{\substack{\text{Ideal}\\\text{QEC}}}&[-0.35cm]\gate{\substack{\text{non-}\\\text{biased}\\\text{noise}}}&[-0.45cm]\rstick{$\!{\cdot}{\cdot}{\cdot}$}
    \end{quantikz}
\end{equation}
see \cref{appendix: mid-circuit morphing} for a discussion of how this could be modified if one wished to simulate the $T$ gate using on-demand morphing instead of on-demand biasing.

Now, to capture the performance of the polynomial phase gate we will isolate only the part of the circuit that is ``in between'' the two ideal QEC circuits, labeled $\mathcal{C}$ in \cref{eq: circuit C in context}. This is because the noise outside of this part of the circuit can be considered as the noise due to the other logical gates that come before and after the $T$ gate. Next, we define the \textit{effective logical channel} $\mathcal{E}$ that maps a \textit{qubit} input state to a \textit{qubit} output state and captures the performance of the circuit $\mathcal{C}$:
\begin{equation}
\label{eq: logical channel equation}
\raisebox{-.43\height}{\includegraphics{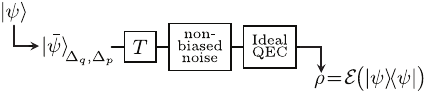}}
\end{equation}
In words, the effective logical channel $\mathcal{E}$ has the following steps.
\begin{enumerate}
    \item First we take as input an arbitrary qubit state $\ket{\psi}$ and encode it into the ideal square GKP code $\ket{\bar{\psi}}$ (this can be viewed as the output of the initial ideal QEC round in \cref{eq: circuit C in context});
    \item we apply the biased envelope operator that prepares the state $\ket{\bar{\psi}}_{\Delta_{q},\Delta_{p}}$ (the ``incoming'' noise);
    \item we apply the logical gate ${T}$;
    \item we apply a non-biased Gaussian random displacement channel to model the measurement noise in the passive Knill QEC circuit; and finally,
    \item we apply a round of ideal QEC, which returns us to the ideal GKP codespace, and read-out the resulting logical qubit density matrix $\rho$.
\end{enumerate}

We can then use the average gate fidelity to quantify the quality of the logical operation, given by~\cite{NIELSEN2002}
\begin{equation}\label{eq: average gate fidelity}
F(\mathcal{E},U)=\int d\psi\bra{\psi}U^{\dag}\mathcal{E}(\ket{\psi}\!\bra{\psi})U\ket{\psi};
\end{equation}
we will also refer to $1-F(\mathcal{E},U)$ as the average gate \textit{in}fidelity.
Some subtleties (such as the orthonormalization of the GKP codewords) have been omitted in the above description for the sake of brevity, see \cref{sec: numerical methods} for more details as well as a description of our numerical methods, which are a truncated Fock space simulation of the channel $\mathcal{E}$.

\begin{table}
    \centering
    \begin{tabular}{c|c}
         CV Gate & Polynomial $P(x)$\\
         \hline
         ${I}$& 0\\
         ${T}_{3}$& $x^{3}/12+x^{2}/8-x/12$\\
         ${T}_\text{GKP}$& $x^{3}/4+x^{2}/8-x/4$\\
         ${T}_4$& $-x^{4}/24+x^{2}/6$\\
         ${\sqrt{T}}$& $-x^{4}/48+x^{2}/12$\\
         $T^{1/4}$& $x^{5}/240-x^{4}/96-x^{3}/48+x^{2}/24+x/60$\\
         $T^{1/8}$& $x^{6}/1440-5x^{4}/576+17x^{2}/720$\\
    \end{tabular}
    \caption{Polynomials used to construct the polynomial phase gates. In particular, each CV gate is implemented using the unitary $\mathrm{exp}\big(2\pi iP({q}/\sqrt{\pi})\big)$. Note that all of the gates labelled ${T}$ implement a logical $T$ gate using a different polynomial, each of these are equivalent up to polynomial stabilisers.}
    \label{tab: simulation polynomials}
\end{table}

The gates and polynomials that we simulated are listed in \cref{tab: simulation polynomials}. For each CV gate, we calculated the average gate fidelity of the corresponding effective logical channel $\mathcal{E}$ as a function of the asymmetry $\lambda$ and the strength $\Delta$ of the envelope operator. We then took the optimal asymmetry $\lambda_{\mathrm{opt}}$ for each value of $\Delta$, and plot the average gate fidelity with the optimized asymmetry.

\begin{figure}
    \includegraphics[]{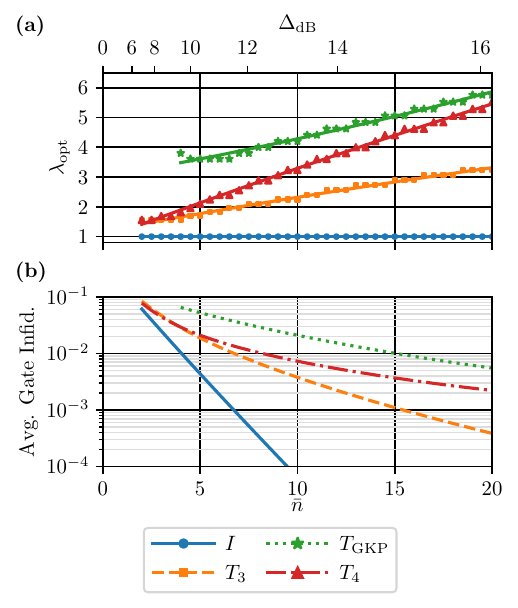}
    \caption{Performance of ${T}$ gates generated by different polynomials. (a) The optimal asymmetry for each gate. We sampled asymmetry values in finite increments, resulting in the optimal asymmetry data points which are fitted with a polynomial fit. Note that for the ${T}_{\text{GKP}}$ gate, for $\bar{n}<4$ no optimal value was found within the sampled range of $\lambda$ (up to 6.5). (b) The average gate infidelity of each gate when using states prepared with the optimal asymmetry.}\label{fig:T_gates}
\end{figure}

\begin{figure}
    \includegraphics[]{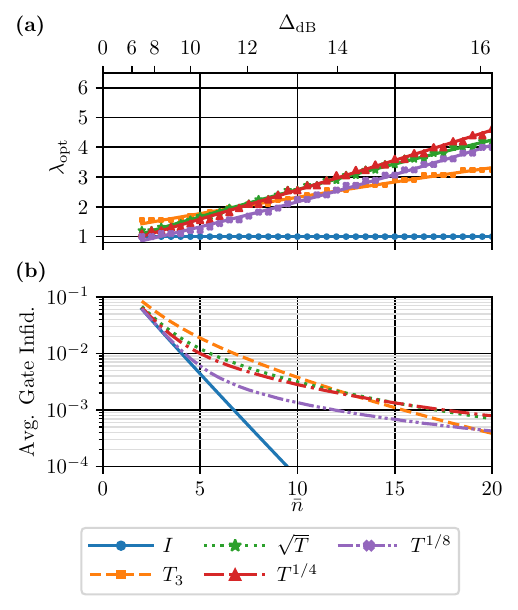}
    \caption{Performance of ${\Lambda}_{m}$ gates. (a) and (b) represent the same quantities as in \cref{fig:T_gates}.}\label{fig:Lambda_gates}
\end{figure}

The results are shown in \cref{fig:T_gates,fig:Lambda_gates}. The plots labeled with ${I}$ represent idling, where the only noise comes from the incoming envelope operator and the noisy syndrome extraction, and serves as a comparison for the other non-trivial logical gates. In \cref{fig:T_gates} we plot three different implementations of the $T$ gate: $T_{\text{GKP}}$ from the original GKP proposal~\cite{GKP2001}, $T_3$ which optimizes the cubic polynomial required to implement a $T$ gate, and $T_{4}$ whose polynomial is even but quartic. Meanwhile, in \cref{fig:Lambda_gates} we plot optimized implementations of the ${\Lambda}_m$ gates for $m=3,4,5,6$, corresponding to logical Bloch sphere rotations around the $Z$-axis by angles of $\pi/4,\pi/8,\pi/16$ and $\pi/32$ respectively. For all of the curves, achieving the precise optimal asymmetry is not necessary to achieve nearly optimal gate fidelities, as we show in \cref{fig:all_gates,fig:all_states} in \cref{sec: numerical methods}. 

We can see that $T_{3}$ outperforms both the original $T_{\text{GKP}}$ and the $T_{4}$ gate across all values of $\Delta$, as expected. Surprisingly however, when compared with phase gates in higher levels of the Clifford hierarchy, the $T_{3}$ does not have the best performance for lower values of $\Delta$, with the ${\sqrt{T}}$ and $T^{1/4}$ gates outperforming $T_3$ for $\Delta>0.2$ ($\bar{n}<12$, $\Delta_{\mathrm{dB}}<14$) and $T^{1/8}$ outperforming $T_3$ for $\Delta>0.16$ ($\bar{n}<19$, $\Delta_{\mathrm{dB}}<16$).

This performance can be explained by two factors: first, because the coefficients in the polynomials are smaller for these higher degree polynomials, at lower $\bar{n}$ (higher $\Delta$) these higher degree polynomial phase gates may spread errors less than the $T_{3}$ gate. Despite this, it is worth emphasizing that the $T_{3}$ gate is guaranteed to perform best asymptotically because it only contains up to cubic terms while the remaining gates have higher degree polynomials, as we showed in \cref{sec: error propagation and fault-tolerant theorem}.

Second, however, it is important to note that the gates representing smaller angle rotations in the Bloch sphere also tend to inherently perform better under the average gate fidelity metric. As a back-of-the-envelope calculation, consider the average gate infidelity between the identity gate and a $\theta$ $Z$-axis rotation in the Bloch sphere, which is given explicitly by
\begin{equation}
    1-F(I,e^{i\theta Z/2})=(1-\cos\theta)/3.
\end{equation}
For the $T$-gate, $\theta=\pi/4$ and $1-F(I,e^{i\theta Z/2})\approx 10\%$, indicating that if the infidelity of the \textit{logical} ${T}$ gate is greater than 10\%, it would be better to simply not perform the $T$ gate at all. In \cref{fig:T_gates,fig:Lambda_gates} we see that the ${T}$ gates have lower than 10 \% infidelity for all the sampled values of $\Delta$ and therefore are achieving non-trivial performance. However, for the $T^{1/8}$ gate we have $\theta=\pi/32$ and $1-F(I,e^{i\theta Z/2})\approx 1.6\times 10^{-3}$, which means that only at $\Delta < 0.24$ ($\bar{n}>8$, $\Delta_{\mathrm{dB}}>12.4$) is the $T^{1/8}$ gate achieving non-trivial performance (see \cref{sec: numerical methods} for more analysis of this). Nevertheless, we expect that directly implementing an $T^{1/8}$ would outperform a compilation of the $T^{1/8}$ gate into logical $T$ and $H$ gates in this regime.

\subsection{Polynomial Phase Gates For Magic State Preparation}\label{sec: state numerics}

\begin{figure}
    \includegraphics[]{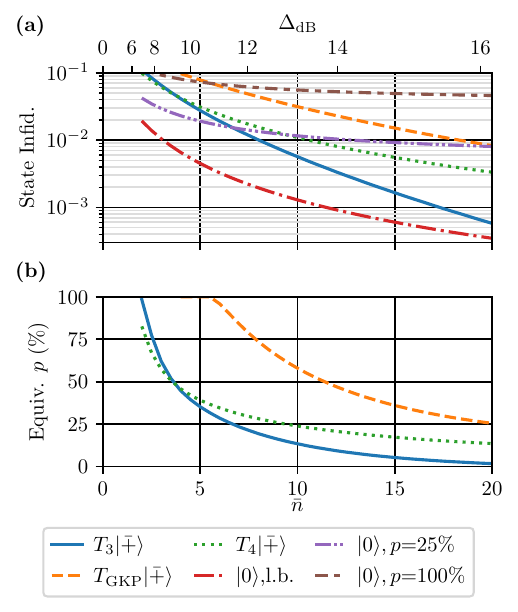}
    \caption{Performance of $T$ gates in preparing magic states. (a) The state infidelity of $\ket{{T}}$ states prepared by applying a logical phase gate to an encoded $\ket{\bar{+}}$ state, compared to those obtained via vacuum state injection using various postselection probabilities $p$~\cite{Baragiola2019}. ``l.b.'' here stands for ``lower bound'' and corresponds to projecting the vacuum state $\ket{\mathrm{vac}}$ onto the ideal GKP codespace, which occurs with probability zero, and therefore lower-bounds the infidelity of any postselection scheme obtained by vacuum state injection. (b) The postselection probability $p$ that is required for the vacuum state injection scheme to match each unitary $T$ gate scheme.}\label{fig:T_states}
\end{figure}

The second context that we are interested in is using the polynomial phase gates as a method for preparing magic states. This is useful in part because it allows us to directly compare to existing work to prepare encoded $\ket{{T}}=T\ket{\bar{+}}$ states~\cite{Baragiola2019}. This latter method---which we refer to as the vacuum state method---consists of preparing a vacuum state $\ket{\mathrm{vac}}$ and performing a single round of GKP QEC. Remarkably, this procedure has a high chance of projecting into a logical GKP codestate that is close to a Hadamard eigenstate, or a magic state that is Clifford equivalent to a Hadamard eigentstate. Therefore, after a Clifford correction dependent on the stabilizer outcome, one can prepare a noisy encoded logical $\ket{{T}}$ state (which is itself Clifford equivalent to a Hadamard eigenstate).

We use the state fidelity of the produced state to quantify the quality of the states prepared by each method. For the polynomial phase gate method, we consider the three polynomial phase gates from \cref{tab: simulation polynomials} that implement a $T$ gate, using the effective logical channel $\mathcal{E}$ from \cref{sec: mid-circuit numerics} applied to the $\ket{+}$ state and taking the state fidelity between the output state and $\ket{T}$. Meanwhile for the vacuum state method we follow the numerical analysis of Ref.~\cite{Baragiola2019} modified to add measurement noise to the round of QEC. We assume that the Clifford correction is noiseless since it could be considered a part of the $T$-injection gadget. We also allow for postselection based on the stabilizer measurement outcomes in the vacuum state method.

In \cref{fig:T_states}(a) we show the state infidelity achieved using each of the three polynomial phase gates and using the vacuum state method. For the vacuum state method we plot three different postselection fractions: no postselection ($p=100\%$), postselecting on the best 25\% of stabilizer outcomes, and postselecting on obtaining $+1$ stabilizer measurement outcomes for both measurements. The latter curve serves as a tight lower bound on the state infidelity that can be achieved by the vacuum state method, since it corresponds to the limit as the postselection probability goes to 0. Meanwhile, in \cref{fig:T_states}(b), we take the state infidelity achieved by each of the polynomial phase gates and calculate the postselection fraction required in the vacuum state method to achieve the same infidelity. From both these plots we see that the ${T}_{3}$ method performs well compared to the vacuum state method, requiring a post-selection fraction of less than 20\% for all values of $\Delta<0.25$ ($\bar{n}>7.5$, $\Delta_{\text{dB}}>12$). It is also possible that at higher $\bar{n}$ the $T_{3}$ method achieves a lower infidelity than even the vacuum state method lower bound, although this is not clear from the collected data. 

It is worth noting that the comparison shown in \cref{fig:T_states} if anything underestimates the performance of the polynomial phase gate method when compared to the vacuum state method. This is because the polynomial phase gate can be implemented directly whenever a $T$ gate appears in a circuit, rather than needing to be used in a magic state injection step. With this in mind, the vacuum state method would also suffer from noise occurring in the injection step, for example from the CX gate and the conditional Clifford correction. On the other hand, if one really wants to use a magic state injection step even with access to a polynomial phase gate, we would gain the ability to postselect on stabiliser measurement outcomes, which would improve the fidelity further. It is therefore reasonable to expect that with these factors taken into account, the polynomial phase gate method outperforms the vacuum state method by a larger degree than shown in \cref{fig:T_states}.

\section{Discussion and Outlook}\label{sec: discussion}

In this manuscript, we have presented a scheme for implementing fault-tolerant non-Clifford gates in the GKP code using polynomial phase gates. The key insight is that on-demand noise biasing can be achieved within the standard Steane QEC circuit by measuring only one of the two stabilizers (\cref{section: biased syndrome measurement circuit}). The cost of performing on-demand biasing is, therefore, the same as the cost of performing a round of QEC. The scheme relies on noise-biased ancilla states, which we show can be prepared directly from non-biased approximate square GKP states. Furthermore, we prove that the cubic phase gate $T_3$, when implemented with the on-demand biasing circuit, is fault-tolerant (\cref{sec: error propagation and fault-tolerant theorem}).

Beyond the cubic phase gate, we develop an analytical framework for studying and optimizing general polynomial phase gates, including those at higher levels of the Clifford hierarchy (\cref{section: optimized logical gate representations}). Central to this framework is the concept of polynomial phase stabilizers, which enable the identification of minimal polynomial phase gate representations through a coefficient reduction procedure. Using this approach, we identify ${T}_3$ as the minimal representation of the logical $ T$ gate, and we also find minimal representations for more exotic gates such as the ${T}^{1/2}$ or ${H}^{1/2}$ gate.

We numerically test our results using Fock-state simulations. We find the the ${T}_3$ gate significantly outperforms the originally proposed ${T}_{\rm GKP}$ gate \cite{GKP2001} and compares favorably to the vacuum-state method \cite{Baragiola2019} even with low post-selection rates. We also investigate minimal polynomial phase representations for logical gates such as the ${T}^{1/2}$, ${T}^{1/4}$, and ${T}^{1/8}$ gate. Surprisingly, these optimized polynomial phase gates can achieve logical gate fidelity comparable to or even exceeding that of the cubic-phase ${T}_3$ gate in some regimes. 

There still remain multiple challenges to exactly implement our fault-tolerant unitary $T$ gate in experimental platforms. In photonics, implementing the polynomial phase gates remains a challenge, although we note that our scheme has multiple attractive features to photonic platforms because our unitary $T$ gate does not require any DV qubits and because the on-demand morphing (\cref{appendix: mid-circuit morphing}) can be implemented using passive Gaussian operations and GKP input states. In superconducting platforms, the polynomial phase gate is more realistic, but these platforms are still limited by DV qubit errors used in error-correction. As a final remark, it is of course possible to use autonomous error-correction techniques that have already been demonstrated in superconducting platforms~\cite{Sivak2023} for most of the error correction, and only use active error correction to implement the on-demand biasing for the non-Clifford gates.

There are several interesting questions that we have not addressed in this manuscript. A natural next step is to incorporate more realistic noise models into our numerical analysis, including experimentally relevant sources such as photon loss, dephasing, and gate imperfections. We note, however, that photon loss can be converted into a Gaussian random displacement channel (and, more generally, any noise channel can be converted into a not-necessarily-Gaussian random displacement channel) via the CV twirling approximation, and that performing this approximation only has a small effect on the logical gate fidelity~\cite{shaw2024logical}. We therefore expect that many of our results that we obtained for Gaussian random displacement channels would still hold in a more general setting. It would also be interesting to compare polynomial phase gates and vacuum-state methods for implementing ${T}$ gates in the context of a full ${T}$ gate injection. Another interesting direction is to numerically investigate the performance of multi-qubit gates such as the ${CS}$ and ${CCZ}$ gates. While we present their minimal polynomial representations in Appendix~\ref{appendix: optimal representation of multi-qubit gates}, and we expect on-demand biasing would improve multi-qubit gate fidelities like it does in the single-qubit ${T}$ gate, we leave it to future work to confirm this. Part of the challenge here is that our numerical simulations used Fock space simulations at high truncation dimensions (roughly 1000), which is infeasible for multiple modes.

Finally, an open question is whether a multi-mode GKP code exists in which implementing the logical ${T}$ gate is easier and induces less errors. A trivial solution to this is to concatenate the GKP code with a DV code that has a transversal ${T}$ gate (which is equivalent to defining a multi-mode GKP code with the same stabilizers), however moving beyond such a concatenated framework would be interesting. If such a code exists, one could envision performing code switching to this code for the ${T}$ gate and switching back to the square GKP code for standard Clifford operations.

One potentially interesting avenue of future work is to find other applications of our on-demand noise-biasing technique beyond enabling fault-tolerant non-Clifford logical gates. As a possible example, perhaps on-demand morphing that we discussed in \cref{appendix: mid-circuit morphing} could be used to improve the error rate of operations involving DV qubits by shortening the conditional displacement gates between CV and DV qubits. However we leave such considerations to future work.

\section{Acknowledgments}
This work is supported by QuTech NWO funding 2020-2024 – Part I “Fundamental Research”, project number 601.QT.001-1, financed by the Dutch Research Council (NWO). 

We acknowledge the use of the DelftBlue supercomputer for running the numerics.
We thank Barbara M.~Terhal for her support throughout the project and for providing feedback on a draft manuscript. We thank Stefano Bosco for discussions on efficient numerical generation of approximate GKP codewords and for providing feedback on the draft manuscript. We thank Xanda Kolesnikow for thoughtful discussions during the project.

\newpage
\appendix

\section{On-demand morphing circuits}\label{appendix: mid-circuit morphing}

In the main text, we showed how to perform on-demand biasing by using a $q$-Steane QEC circuit in two ways in \cref{eq: biased noise q-Steane,eq: biased noise q-Steane with rect ancilla}. In this appendix, we show how to perform on-demand \textit{code morphing} from a non-biased square GKP codestate to a non-biased rectangular GKP code. This code morphing method has the advantage of using only passive Gaussian unitaries, while still enabling improved performance under polynomial phase gates due to the equivalence between non-biased rectangular and biased square GKP codestates in~\cref{eq: relation between biased square GKP to non-biased rectangular GKP}.
We follow a similar structure here as in the main text, beginning by how to construct an on-demand morphing circuit that introduces a \textit{fixed} amount of asymmetry, before generalizing to a circuit that can generate arbitrary amounts of asymmetry.

We start with the $q$-Steane circuit rewritten in terms of a beam-splitter and squeezing operators:
\begin{equation}\tag{\ref{eq: q-Steane rewritten}}
\begin{quantikz}
     \lstick{$\ket{\bar \psi}_{\Delta}$} &  \gate[2]{ \rm BS(\frac{\pi}{4}) } & \gate{\text{Sq}^{\dg}(\sqrt{2}) }   &   \\
     \lstick{$\ket{\bar +}_{\Delta}$} &                                & \gate{\text{Sq}^{\dg}(\frac{1}{\sqrt{2}})} & \meter{{q}}
\end{quantikz},
\end{equation}
where the output (after the correction) is a biased square codestate with a fixed amount of bias $\ket{\bar{\psi}}_{\Delta/\sqrt{2},\sqrt{2}\Delta}$.
Now, if we remove the squeezing operators, the output state will be
\begin{equation}
    {\text{Sq}}(\sqrt{2})\ket{\bar{\psi}}_{\Delta/\sqrt{2},\sqrt{2}\Delta}=\ket{\bar{\psi}}_{\Delta}^{\lambda=2}
\end{equation}
by \cref{eq: relation between biased square GKP to non-biased rectangular GKP}. Therefore, the circuit
\begin{equation}\label{eq: mid-circuit morphing circuit}
\begin{quantikz}
     \lstick{$\ket{\bar \psi}_{\Delta}$} &  \gate[2]{ \rm BS(\frac{\pi}{4}) }&\\
     \lstick{$\ket{\bar +}_{\Delta}$}&& \meter{{q}}
\end{quantikz}
\end{equation}
is a on-demand morphing circuit from the square to the rectangular GKP code with $\lambda=2$. Note that, compared to \cref{eq: q-Steane rewritten}, we will need to rescale the measurement outcome to account for the squeezing operator that we removed. Moreover, from \cref{eq: BS envelope identity}, the envelope operator is preserved through \cref{eq: mid-circuit morphing circuit}, and therefore both the input and output states will be non-biased.

So far, we have only been able to morph from the square to the $\lambda=2$ rectangular GKP code; however, we wish to be able to morph to rectangular GKP codes with arbitrary $\lambda$. To derive such a circuit, we begin with the (arbitrary-$\lambda$) on-demand biasing circuit from the main text
\begin{equation}\tag{\ref{eq: biased noise q-Steane with rect ancilla}}
    \begin{quantikz}
        \ket{\bar \psi}_{\Delta} & \gate[2]{ e^{-i\sqrt{\lambda}{q}_{1}{p}_{2}}}  &  \\
       \ket{ \bar +}_{\Delta}^{\lambda} &   & \meter{{q}}  
    \end{quantikz},
\end{equation}
which is written in the form such that it uses a non-biased rectangular ancilla codestate.

Now, we wish to perform the same rearrangement as before by rewriting this circuit in terms of a beam-splitter and squeezing operations. In particular, we wish to find $\xi,\theta,\alpha_{1}$ and $\alpha_{2}$ such that:
\begin{multline}\label{eq: desired rearrangement}
    \begin{quantikz}
        \lstick{$\ket{\bar \psi}_{\Delta }$} & \gate[2]{ e^{-i\sqrt{\lambda}{q}_{1}{p}_{2}}} &  \gate[2]{ e^{i\xi{p}_{1}{q}_{2}}} &  \\
       \lstick{$\ket{ \bar +}_{\Delta}^{\lambda}$}  & & & \meter{{q}}  
    \end{quantikz}\\
    =\begin{quantikz}
     \lstick{$\ket{\bar \psi}_{\Delta}$} &  \gate[2]{ \rm BS (\theta)}  & \gate{\text{Sq}(\alpha_{1}) }   & \\
     \lstick{$\ket{\bar +}_{\Delta}^{\lambda}$} &                                & \gate{\text{Sq}(\alpha_{2})} & \meter{{q}}.
\end{quantikz}
\end{multline}
As before, we are allowed to freely add the $e^{i\xi p_{1}q_{2}}$ term to \cref{eq: biased noise q-Steane with rect ancilla} because it can be compensated for by a displacement correction $e^{-i\xi p_{1}q_{\text{m}}}$ based on the measurement outcome $q_{\text{m}}$.

One can solve this equation using the formalism of Gaussian operators. Given a two-mode Gaussian unitary operator ${U}$, we can always define a $4\times 4$ matrix $S$ such that
\begin{equation}
    S\begin{bmatrix}{q}_{1}\\{q}_{2}\\{p}_{1}\\{p}_{2}\end{bmatrix}=\begin{bmatrix}{U}^{\dag}{q}_{1}{U}\\{U}^{\dag}{q}_{2}{U}\\{U}^{\dag}{p}_{1}{U}\\{U}^{\dag}{p}_{2}{U}\end{bmatrix}.
\end{equation}
Moreover, given an arbitrary $4\times4$ matrix $S$ that satisfies the relation
\begin{align}
    S^{T}\Omega S&=\Omega,&\Omega=\begin{bmatrix}0&\mathrm{Id}\\-\mathrm{Id}&0\end{bmatrix},
\end{align}
(that is, $S$ is a \textit{symplectic} matrix) then one can uniquely define a corresponding Gaussian unitary operator~\cite{Weedbrook12}. With this, the corresponding $S$ matrices for the left- and right-hand sides of \cref{eq: desired rearrangement} are equal if and only if the unitaries themselves are equal. The corresponding matrices for the left- and right-hand sides respectively are given by
\begin{subequations}
\begin{align}
    S_{\text{LHS}}&=\begin{bmatrix}1-\xi\sqrt{\lambda}&-\xi&0&0\\\sqrt{\lambda}&1&0&0\\0&0&1&-\sqrt{\lambda}\\0&0&\xi&1-\xi\sqrt{\lambda}\end{bmatrix}\\
    S_{\text{RHS}}&=\begin{bmatrix}\cos(\theta)/\alpha_{1}&-\sin(\theta)/\alpha_{1}&0&0\\\sin(\theta)/\alpha_{2}&\cos(\theta)/\alpha_{2}&0&0\\0&0&\alpha_{1}\cos(\theta)&-\alpha_{1}\sin(\theta)\\0&0&\alpha_{2}\sin(\theta)&\alpha_{2}\cos(\theta)\end{bmatrix}
\end{align}
\end{subequations}
From here, it is straightforward to see that the solution is given by
\begin{subequations}
\begin{align}
    \xi&=\frac{\sqrt{\lambda}}{\lambda+1},&\theta&=\mathrm{arctan}(\sqrt{\lambda}),\\
    \alpha_{1}&=\sqrt{\lambda+1},&\alpha_{2}&=\frac{1}{\sqrt{\lambda+1}}.
\end{align}
\end{subequations}
Therefore, the circuit
\begin{equation}\label{eq: mid-circuit arbitrary morphing circuit}
\begin{quantikz}
     \lstick{$\ket{\bar \psi}_{\Delta}$} &  \gate[2]{ \rm BS\big(\mathrm{arctan}(\sqrt{\lambda})\big) } & \\
     \lstick{$\ket{\bar +}_{\Delta}^{\lambda}$}&& \meter{{q}}
\end{quantikz}
\end{equation}
is a on-demand morphing circuit from the square to the $\lambda+1$ rectangular GKP code, with output state (after the correction) given explicitly by $\ket{\bar{\psi}}_{\Delta}^{\lambda+1}$.

Having explained how the on-demand morphing circuit works, we now briefly summarize the difference between the on-demand biasing and on-demand morphing methods for performing a polynomial phase gate associated with the polynomial $P(x)$. In the on-demand biasing method we follow the steps outlined in~\cref{eq: T gate with on-demand biasing}, first performing the on-demand biasing circuit~\cref{eq: biased noise q-Steane} or \eqref{eq: biased noise q-Steane with rect ancilla}, followed by the square GKP polynomial phase gate $\mathrm{exp}\big(i2\pi P({q}/\sqrt{\pi})\big)$, followed by the passive Knill error correction circuit~\cref{eq: passive Knill} to return to the non-biased codespace. In the on-demand morphing method, we first perform the on-demand morphing circuit~\cref{eq: mid-circuit arbitrary morphing circuit} and then perform the \textit{rectangular} GKP polynomial phase gate $\mathrm{exp}\big(i2\pi P({q}/\sqrt{\lambda\pi})\big)$. At this point, however, we are in a rectangular GKP code and therefore cannot simply perform the passive Knill error correction circuit because this would keep us in the rectangular codespace instead of morphing us back to the square one.

Therefore, we need to perform a morphing circuit based on the $p$-Steane QEC circuit that returns us to the square GKP code. The derivation of this circuit is much the same as for \cref{eq: mid-circuit arbitrary morphing circuit}, but one does need to take care that the $p$-Steane QEC circuit that one begins with is for the rectangular GKP code instead of the square GKP code. With this, the morphing circuit is given explicitly by
\begin{equation}\label{eq: mid-circuit reverse morphing circuit}
\begin{quantikz}
     \lstick{$\ket{\bar \psi}_{\Delta}$} &  \gate[2]{ {\rm BS}\big(\mathrm{arctan}(\sqrt{\lambda-1})\big) } & \\
     \lstick{$\ket{\bar 0}_{\Delta}^{\lambda'}$}&& \meter{{q}}
\end{quantikz},
\end{equation}
where $\lambda'=(1-1/\lambda)^{-1}$. In the full circuit, this looks like
\begin{equation}
    \begin{quantikz}
        \lstick{$\cdots$}&\gate{\substack{\text{Passive}\\q\text{-Steane}\\\text{QEC}}} \gategroup[1,steps=1,style={dashed,rounded corners,fill=green!20,inner sep =0 pt},background,label style={label position=above,anchor=south,yshift=-0.2cm}]{$\substack{\text{On-demand}\\\text{morphing}}$}&\gate{T^{(\lambda)}}&\gate{\substack{\text{Passive}\\p\text{-Steane}\\\text{QEC}}}\gategroup[1,steps=1,style={dashed,rounded corners,fill=green!20,inner sep =0 pt},background,label style={label position=above,anchor=south,yshift=-0.2cm}]{$\substack{\text{On-demand}\\\text{morphing}}$}&\rstick{$\cdots$}\label{eq: T gate with on-demand morphing}
    \end{quantikz}.
\end{equation}

On-demand morphing has two potential advantages over on-demand biasing. First, \cref{eq: mid-circuit arbitrary morphing circuit,eq: mid-circuit reverse morphing circuit} use only passive Gaussian unitary operators (beam-splitters) instead of logical $\text{CX}$ gates that are not energy-preserving, which may simplify experimental implementations. And second, because the resulting rectangular states are non-biased, their average photon number is lower than the corresponding biased square codestates. This means that loss and dephasing may affect the non-biased rectangular states less severely than the biased square codestates.

In \cref{sec: numerics}, our simulations were based on the on-demand biasing method~\cref{eq: T gate with on-demand biasing} instead of the on-demand morphing method~\cref{eq: T gate with on-demand morphing}. In particular, we simulated the circuit $\mathcal{C}$ given by
\begin{equation}\tag{\ref{eq: circuit C in context}}
    \begin{quantikz}
        \lstick{${\cdot}{\cdot}{\cdot}\!$}&[-0.45cm]\gate{\substack{\text{non-}\\\text{biased}\\\text{noise}}}&[-0.35cm]\gate{\substack{\text{Ideal}\\\text{QEC}}}\gategroup[1,steps=5,style={dashed,rounded corners,fill=red!20,inner sep=-1.5 pt},background,label style={label position=above,anchor=south,yshift=-0.2cm}]{$\mathcal{C}$}&[-0.35cm]\gate{\substack{\text{biased}\\\text{noise}}}&[-0.2 cm]\gate{T}&[-0.2 cm]\gate{\substack{\text{non-}\\\text{biased}\\\text{noise}}}&[-0.35cm]\gate{\substack{\text{Ideal}\\\text{QEC}}}&[-0.35cm]\gate{\substack{\text{non-}\\\text{biased}\\\text{noise}}}&[-0.45cm]\rstick{$\!{\cdot}{\cdot}{\cdot}$}
    \end{quantikz},
\end{equation}
where the specific choices of biased and non-biased noise are motivated by the QEC circuits used in the on-demand biasing method~\cref{eq: T gate with on-demand biasing}. To adapt this circuit for on-demand morphing~\cref{eq: T gate with on-demand morphing}, we would do two things. First, replace the biased and non-biased boxes to match the noise that arises from the QEC circuits used in \cref{eq: T gate with on-demand morphing}, and second, use \cref{eq: relation between biased square GKP to non-biased rectangular GKP} (change the noise bias as appropriate) to find an equivalent circuit that acts in the square GKP code, giving
\begin{equation}
    \begin{quantikz}
        \lstick{${\cdot}{\cdot}{\cdot}\!$}&[-0.45cm]\gate{\substack{\text{non-}\\\text{biased}\\\text{noise}}}&[-0.35cm]\gate{\substack{\text{Ideal}\\\text{QEC}}}\gategroup[1,steps=5,style={dashed,rounded corners,fill=red!20,inner sep=-1.5 pt},background,label style={label position=above,anchor=south,yshift=-0.2cm}]{$\mathcal{C}'$}&[-0.35cm]\gate{\substack{\text{biased}\\\text{noise}}}&[-0.2 cm]\gate{T}&[-0.2 cm]\gate{\substack{\text{biased}\\\text{noise}}}&[-0.35cm]\gate{\substack{\text{Ideal}\\\text{QEC}}}&[-0.35cm]\gate{\substack{\text{non-}\\\text{biased}\\\text{noise}}}&[-0.45cm]\rstick{$\!{\cdot}{\cdot}{\cdot}$}
    \end{quantikz}.
\end{equation}
Simulating $\mathcal{C}'$ would therefore give quantitatively different results to $\mathcal{C}$, but should give \textit{qualitatively} similar results because the noise coming from the final QEC round does not get spread by the $T$ gate.

\section{Passive breeding circuit for non-biased noise rectangular GKP qubit states}
\label{appendix: passive biasing circuit}

In all of the on-demand biasing and morphing circuits with $\lambda>2$ in \cref{eq: biased noise q-Steane,eq: biased noise q-Steane with rect ancilla,eq: mid-circuit arbitrary morphing circuit}, we have made use of either a biased square or non-biased rectangular GKP ancilla state. However we have not explained how one can produce such a state in the first place. While it may be possible to directly produce such an ancilla state by modifying existing state generation schemes \cite{Le2019,Rymarz2021,hastrup2021measurement,kolesnikow2025protected,Weigand2016}, here we instead show how to obtain a non-biased rectangular codestate iteratively from only non-biased square GKP codestates. If desired, the non-biased rectangular codestate could then be converted to a biased square codestate using a squeezing unitary~\cref{eq: relation between biased square GKP to non-biased rectangular GKP}.

The circuit to breed such a state is conceptually quite similar to the on-demand morphing circuit~\cref{eq: mid-circuit arbitrary morphing circuit}, but here we have slightly different requirements that need to be satisfied. In particular, we require that each input ancilla state is a non-biased square GKP codestate $\ket{\bar{\psi}}_{\Delta}$, but each input \textit{data} state is a non-biased rectangular GKP codestate $\ket{\bar{\psi}}_{\Delta}^{\lambda}$. We then want to output a non-biased rectangular GKP codestate with an increased $\lambda$, in particular this will end up being $\ket{\bar{\psi}}_{\Delta}^{\lambda+1}$.

To construct this circuit we begin with the $q$-Steane QEC circuit for the rectangular GKP code
\begin{equation}
    \label{eq: rectangular q-Steane circuit with rectangular ancilla}
    \begin{quantikz}
      \lstick{$\ket{\bar{\psi}}_{\Delta}^{\lambda}$}&  \gate[2]{ e^{-i {q}_1  {p}_2}} &   \\
      \lstick{$\ket{\bar+}_{\Delta}^{\lambda}$}  &    & \meter{q}
\end{quantikz}.
\end{equation}
We can turn the input rectangular ancilla codestate into a square codestate using \cref{eq: relation between biased square GKP to non-biased rectangular GKP}
\begin{equation}
    \label{eq: rectangular q-Steane circuit with square ancilla}
    \begin{quantikz}
      \lstick{$\ket{\bar{\psi}}_{\Delta}^{\lambda}$}&  \gate[2]{ e^{-i \frac{{q}_1  {p}_2}{\sqrt{\lambda}} } } &   \\
      \lstick{$\ket{\bar+}_{\Delta}$}  &    & \meter{q}
\end{quantikz}.
\end{equation}

It is straightforward following the same method as in \cref{appendix: mid-circuit morphing} that the $q$-Steane QEC circuit in Eq.~\eqref{eq: rectangular q-Steane circuit with square ancilla} can be rewritten in terms of beam-splitters and squeezers
\begin{subequations}
\begin{align}
     \label{eq: BS identity}
    & \begin{quantikz}
     &  \gate[2]{ e^{-i \frac{{q}_1 {p}_2 }{\sqrt{\lambda}}} } & \gate[2]{ e^{i \frac{ \sqrt{\lambda}}{\lambda+1} {p}_1 {q}_2}} &  \\
     &  &   &
\end{quantikz} = \\
   & \begin{quantikz}
    & \gate[2]{ {\rm BS}\Big(\mathrm{arctan}\big(\frac{1}{\sqrt{\lambda}}\big)\Big) } &  \gate{ {\text{Sq}}\Big(\sqrt{\frac{\lambda+1}{\lambda}}\Big) } &    \\
    & & \gate{ {\text{Sq}}\Big(\sqrt{\frac{\lambda}{\lambda+1}}\Big)} & 
\end{quantikz}.
\end{align}
\end{subequations}
By absorbing the single-mode squeezer on the ancilla into the homodyne measurement and omitting the squeezer on the data qubit, we can map a rectangular GKP state with non-biased noise and asymmetry $\lambda$ to one with asymmetry $\lambda+1$ 
\begin{equation}\label{eq: BS round}
\begin{quantikz}
    \lstick{$\ket{ \bar \psi }_{\Delta}^{\lambda}$} &  \gate[2]{ {\rm BS}(\theta_{\lambda}) } & \\
    \lstick{$\ket{\bar +}_{\Delta} $}    &  &  \meter{q}
\end{quantikz}
\end{equation}
Therefore, by repeating Eq.~\eqref{eq: BS round} a total of $(\lambda-1)$ times, we can generate a non-biased rectangular GKP state $\ket{\bar{+}}_{\Delta}^{\lambda}$ from $\lambda$ non-biased square GKP states $\ket{\bar +}_{\Delta}$. We refer to this procedure as the passive rectangular GKP state breeding protocol. Importantly, the output of this protocol will always have $\lambda\in\mathbb{Z}$; however this is not a problem because $\lambda$ does not need to be fine-tuned in order achieve close-to-optimal error rates, as we show in \cref{fig:all_gates,fig:all_states}.

\section{Proof of no-go theorem}
\label{appendix: proof nogo theorem}

In the main text, we showed how to perform on-demand biasing by using a $q$-Steane QEC circuit in two ways in \cref{eq: biased noise q-Steane,eq: biased noise q-Steane with rect ancilla}. We found that the noise profile after the on-demand biasing circuit satisfies the trade-off relation
\begin{equation}
    \tag{\ref{eq: trade-off relation}}
    \Delta_q \Delta_p \geq \Delta^2.
\end{equation}
This naturally raises the question of whether a syndrome measurement circuit could surpass the trade-off relation in Eq.~\eqref{eq: trade-off relation}. In other words, is it possible to suppress noise in one quadrature without inducing a corresponding increase in the other?

In this appendix, we show in Theorem~\ref{theorem: formal no-go theorem} that no such circuit can be constructed using deterministic Gaussian operations and approximate GKP ancillas. As a consequence, the $q$-Steane QEC circuits that we present achieve the optimal output $\Delta_{q}$ and $\Delta_{p}$ parameters for deterministic Gaussian circuits.

Before presenting the formal statement of Theorem~\ref{theorem: formal no-go theorem}, we briefly review the generalization of the Gaussian random displacement channel to multi-mode system and the definition of symplectic matrices. In an $N$-mode system, the $N$-mode displacement operator ${W}(\bo v)$ is defined as 
\begin{equation}
    {W}(\bo v) = \exp(i\sqrt{2\pi}~{\boldsymbol \zeta} \Omega \bo v )
\end{equation}
where $\boldsymbol{\zeta}$ is the quadrature column vector
    \begin{equation}
        \boldsymbol{ \zeta} = \begin{bmatrix} {q}_1 & \dots & {q}_N & {p}_1 & \dots & {p}_N \end{bmatrix}^{T},
    \end{equation}
and the matrix $\Omega$ is the $2N\times 2N$ symplectic form 
\begin{equation}\label{eq: Omega matrix}
    \Omega = \begin{bmatrix}
        \bo 0_{N\times N} & \mathrm{Id}_{N\times N}\\
        \mathrm{Id}_{N\times N} & \bo{0}_{N\times N}
    \end{bmatrix}.
\end{equation}
The group of $2N\times 2N$ symplectic matrices $S$ is the group of real matrices that preserve the symplectic form 
\begin{equation}
    S^{T} \Omega S = \Omega. 
\end{equation}
Gaussian unitary operators are, by definition, unitary operators that map quadrature operators to other quadrature operators; or, equivalently, they are operators that can be written as a product of exponents that are quadratic in the quadrature operators. Gaussian unitary operators that contain \textit{only} quadratic terms are in one-to-one correspondence with symplectic matrices, given by the relation
\begin{equation}
    U_{S}^{\dag} \vect{\zeta} U_{S}=S\vect{\zeta},
\end{equation}
where $U_{S}$ acts component-wise on the elements of $\vect{\zeta}$.

The multi-mode Gaussian random displacement channel $\mathcal{G}_{\boldsymbol{\mu},\Sigma}$ is defined similar to the $2$-mode channel in Eq.~\eqref{eq: GRD channel definition}
\begin{align}
        \mathcal{G}_{\boldsymbol{\mu},\Sigma}(\rho) & = \int d\bo v \frac{ \exp\Big[-\pi(\bo v-\boldsymbol{\mu})^{T}~ \Sigma^{-1}(\bo v-\boldsymbol{\mu})  \Big]  }{\sqrt{|\det \Sigma|}} \\
        & \qquad \qquad  \times {W}(\bo v ) {\rho} {W}(\bo v)^{\dg}
\end{align}
but here we have a $2N$-dimension normal distribution; note here that unlike \cref{eq: GRD channel definition} we also consider channels with non-zero mean $\vect{\mu}$.

We consider the following model. We assume that the data qubit and the ancillary qubits can be encoded using different GKP codes. All physical modes are subjected to independent Gaussian displacement noise channels $\mathcal{G}_{\bo 0, \Delta^2 \times \rm Id}$, which is the noise model one obtains when applying the twirling approximation to $\text{Env}_{\Delta}$. The syndrome extraction circuit consists of Gaussian unitary operators, ideal homodyne measurements on the ancillas, and arbitrary classical decoders. Under these assumptions, the circuit model under consideration is given by
    \begin{equation}
        \label{eq: no-go circuit model}
        \begin{quantikz}
        \ket{\bar{\psi}}_{D} &  \gate{\mathcal{G}_{\bo 0,\Delta^2 \times \rm Id}}    & \gate[2]{{U}_{S}}   &  \\
       \ket{\bar \varphi}_{A} & \gate{\mathcal{G}_{\bo 0, \Delta^2 \times \rm Id}^{\otimes N}} &  & \meter{{q}}
        \end{quantikz} 
    \end{equation}
where $\ket{\bar{\psi}}_{D}$ is a single-mode ideal GKP ``data'' codestate, $\ket{\bar{\phi}}_{A}$ is an $N$-mode ideal GKP ``ancilla'' codestate, ${U}_{S}$ is a Gaussian unitary with the associated symplectic matrix $S$ and $\Delta^2$ is the variance of the noise applied to each GKP codestate. Note that homodyne measurements in other bases can be incorporated into this model by writing them as a Gaussian unitary preceding a position measurement, and then incorporating the Gaussian unitary into $U_{S}$. 

\begin{theorem}\label{theorem: formal no-go theorem}
     The resulting state of the data qubit after the circuit \cref{eq: no-go circuit model} and the displacement correction step is
           \begin{align}
                  \int d\bo v~ \mathcal{X}(\bo v)    W(\bo v) \ket{\bar\psi}_D \bra{\bar\psi}_D  W(\bo v)^{\dg}, 
           \end{align}
     where $\mathcal{X}(\bo v)$ is a mixture of Gaussian distributions
     \begin{equation}
         \mathcal{X}(\bo v) = \sum_{i} p_i ~\mathcal{N}(\boldsymbol{\mu}_i, \Delta^2 \Sigma^{'}, \bo v )
     \end{equation}
     for some positive weights $p_i$ satisfying $\sum_{i} p_i  =1$, for some mean vectors $\boldsymbol{\mu}_{i} $ and for some covariance matrix $\Sigma'$ satisfying $\det(\Sigma^{'}) = 1$. The random displacement vector $\bf{v}$ drawn from the distribution $\mathcal{X}(\bo v)$ satisfies the inequality 
     \begin{equation}\label{eq: formal no-go theorem}
         \det(E[ (\bo{v}-E[\bo v]) (\bo{v}-E[\bo v])^{T}  ]) \geq \Delta^4.
     \end{equation}
\end{theorem}

The circuit in Eq.~\eqref{eq: no-go circuit model} can be interpreted as a noise-biasing circuit with noisy ancillas, while Eq.~\eqref{eq: formal no-go theorem} provides a lower bound on the product of the variances of position and momentum errors. The determinant appearing in Eq.~\eqref{eq: formal no-go theorem} is the generalization of the informal trade-off relation in Eq.~\eqref{eq: trade-off relation}, extending it to include cross-correlations between position and momentum displacement errors.

As a final comment, we emphasize that although Theorem~\ref{theorem: formal no-go theorem} holds regardless of the choice of decoder, it does not imply that all decoders perform equally well. The role of a good decoder is to approach the lower bound as close as possible. Theorem~\ref{theorem: formal no-go theorem} also shows that the biased $q$-Steane QEC circuit presented in Eq.~\eqref{eq: biased noise q-Steane} achieves this optimal bound. 

\begin{proof}
Since ${U}_{S}$ is a Gaussian unitary operator, we can commute the random Gaussian displacement channels $\mathcal{G}_{\bo 0, \Delta^2 \times \rm Id}$ through ${U}_{S}$, resulting in a new channel $\mathcal{G}_{\bo 0, \Delta^2 \Sigma}$.
\begin{equation} 
    \label{eq: commute G through Us}
     \begin{quantikz}
         &  \gate{\mathcal{G}_{\bo 0,\Delta^2\times\mathrm{Id}}}    & \gate[2]{{U}_{S}}  &  \\
         & \gate{\mathcal{G}_{\bo 0,\Delta^2\times\mathrm{Id}}^{\otimes N}} &  & 
        \end{quantikz} =    \begin{quantikz}
            & \gate[2]{{U}_{S}}  & \gate[2]{ \mathcal{G}_{\bo 0, \Delta^2 \Sigma} } &  \\
            &                  &                                  &
        \end{quantikz}.
\end{equation}
The normalized covariance matrix $\Sigma = S S^{T}$ is both symplectic and symmetric, and because it is symplectic it has unit determinant $\det(\Sigma) = 1$.

We decompose $\Sigma$ into a $4 \times 4$ block matrix 
\begin{equation}
     \Sigma = \begin{pmatrix}
    [\Sigma_{qq}]_{DD} & [\Sigma_{qq}]_{DA} & [\Sigma_{qp}]_{DD} & [\Sigma_{qp}]_{DA} \\
    [\Sigma_{qq}]_{AD} & [\Sigma_{qq}]_{AA} & [\Sigma_{qp}]_{AD} & [\Sigma_{qp}]_{AA} \\
    [\Sigma_{pq}]_{DD} & [\Sigma_{pq}]_{DA} & [\Sigma_{pp}]_{DD} & [\Sigma_{pp}]_{DA} \\
    [\Sigma_{pq}]_{AD} & [\Sigma_{pq}]_{AA} & [\Sigma_{pp}]_{AD} & [\Sigma_{pp}]_{AA}
    \end{pmatrix}
\end{equation}
Here, the subscripts $q$ and $p$ refer to the position and momentum quadratures, respectively. The labels $D$ and $A$ indicate whether the corresponding block involves the data qubit ($D$) or ancillary modes ($A$). For example, $[\Sigma_{qp}]_{DA}$ denotes the cross-covariance between the position quadrature of the data qubit and the momentum quadratures of the ancilla qubits. The size of each block is
\begin{equation}
    \begin{aligned}
        [\Sigma_{qq}]_{DD}, [\Sigma_{pp}]_{DD}, [\Sigma_{pq}]_{DD}, [\Sigma_{qp}]_{DD} \in \mathbb{R}^{1\times 1} \\
        [\Sigma_{qq}]_{AD}, [\Sigma_{qp}]_{AD}, [\Sigma_{pq}]_{AD}, [\Sigma_{pp}]_{DD} \in \mathbb{R}^{N \times 1} \\
        [\Sigma_{qq}]_{DA}, [\Sigma_{qp}]_{DA}, [\Sigma_{pq}]_{DA}, [\Sigma_{pp}]_{DA} \in \mathbb{R}^{1 \times N} \\
        [\Sigma_{qq}]_{AA}, [\Sigma_{pp}]_{AA}, [\Sigma_{pq}]_{AA}, [\Sigma_{qp}]_{AA} \in \mathbb{R}^{N \times N}
    \end{aligned}
\end{equation}

Since $\Sigma$ is symplectic, its inverse $\Sigma^{-1}$ is also symplectic and can be expressed in terms of the entries of $\Sigma$, i.e $\Sigma^{-1} = \Omega^{T} \Sigma^{T} \Omega$. The explicit $4\times4$ block matrix  of $\Sigma^{-1}$ is given by 
\begin{equation}
\begin{aligned}
    \Sigma^{-1} & = \begin{pmatrix}
    [\Sigma_{pp}]_{DD} & [\Sigma_{pp}]_{DA} & -[\Sigma_{qp}]_{DD} & -[\Sigma_{qp}]_{DA} \\
    [\Sigma_{pp}]_{AD} & [\Sigma_{pp}]_{AA} & -[\Sigma_{qp}]_{AD} & -[\Sigma_{qp}]_{AA} \\
    -[\Sigma_{pq}]_{DD} & -[\Sigma_{pq}]_{DA} & [\Sigma_{qq}]_{DD} & [\Sigma_{qq}]_{DA} \\
    -[\Sigma_{pq}]_{AD} & -[\Sigma_{pq}]_{AA} & [\Sigma_{qq}]_{AD} & [\Sigma_{qq}]_{AA}
    \end{pmatrix} .
\end{aligned}
\end{equation}

Because we measure only the position quadrature of the ancillary modes, we gain no information about displacement errors occurring in the momentum quadrature. Therefore, the joint distribution of displacement errors on the position and momentum quadratures of the data qubit, as well as the position quadratures of the ancillas, is described by the marginal of the full error distribution. Specifically, this corresponds to marginalizing the original Gaussian distribution with covariance matrix $\Delta^2 \Sigma$. The marginal of a multivariate Gaussian distribution is itself a Gaussian, with covariance matrix $\Delta^2 \tilde{\Sigma}$ where $\tilde{\Sigma}$ is obtained by removing the rows and columns of $\Sigma$ corresponding to the unmeasured variables—in this case, the momentum quadratures of the ancilla qubits. For our block ordering, this corresponds to removing the fourth row and fourth column of $\Sigma$. For convenience, we permute the second and third block-row and block-column of $\Delta^2 \tilde{\Sigma}$ to obtain
\begin{subequations}
    \begin{align}
    \label{eq: marginal cov block form}
    % \sigma^2 \Sigma \xrightarrow[\text{permuted}]  {\text{marginalized}}
    \Delta^2 \tilde{\Sigma} & = \Delta^2 \begin{pmatrix}
    [\Sigma_{qq}]_{DD} & [\Sigma_{qp}]_{DD} & [\Sigma_{qq}]_{DA}  \\
    [\Sigma_{pq}]_{DD} & [\Sigma_{pp}]_{DD} & [\Sigma_{pq}]_{DA}  \\
    [\Sigma_{qq}]_{AD} & [\Sigma_{qp}]_{AD} & [\Sigma_{qq}]_{AA}  
    \end{pmatrix} \\
    & = \Delta^2 \begin{pmatrix}
        [\tilde{\Sigma}]_{DD} & [\tilde{\Sigma}]_{DA} \\
        [\tilde{\Sigma}]_{AD} & [\tilde{\Sigma}]_{AA}
    \end{pmatrix}.
\end{align}
\end{subequations}

After performing position homodyne measurements on the ancillary qubits, we obtain partial information about the displacement error introduced by the Gaussian random displacement channel: since the ancillary qubits are encoded in GKP code states, the measurement only reveals the displacement modulo the GKP grid. This means we can only determine the error up to a logical operation. For example, in the standard $q$-Steane circuit where both the data and ancillary qubits are square-lattice GKP qubits, the measurement outcome allows us to infer the displacement error on the ancillary qubit only up to a logical $X$ operator. 

Consequently, the measurement outcome is consistent with a discrete set of possible position displacements $\{ \bo{u}_{q,i} \} $, each occurring with probability $p_i$ such that $\sum_{i} p_i = 1$ and $p_i \geq 0$. These tuples  $\{ p_i, \bo{u}_{q,i}\}$ depend on the GKP encodings of the data and ancillary modes and are typically intractable to compute analytically. Fortunately, explicit knowledge of these values is not required for the analysis that follows.

Now, suppose that a specific displacement error $\bo{u}_{q,i}$ occurred on the ancillary modes. We are then interested in the resulting displacement error on the data qubit, conditioned on this event. In other words, we wish to determine the conditional distribution of the random displacement on the data qubit given that $\bo{u}_{q,i}$ occurred on the ancillary qubit. 

As discussed earlier, the joint distribution of displacement errors is a multivariate Gaussian with zero mean and covariance matrix $\Delta^2 \tilde{\Sigma}$. Given this, the conditional distribution of the data qubit’s displacement—conditioned on a displacement error $\{ \bo{u}_{q,i} \}$ occurring on the ancillary qubit—is also Gaussian. Its mean and covariance are given by
\begin{equation}
    \label{eq: bayesian update rule}
    \begin{aligned}
        \boldsymbol \mu_i & = [\tilde{\Sigma}]_{DA} \frac{1}{ [\tilde{\Sigma}]_{AA}} \bo{u}_{q,i}, \\
        \Sigma^{'} & = [\tilde{\Sigma}]_{DD}  - [\tilde{\Sigma}]_{DA}  \frac{1}{[\tilde{\Sigma}]_{AA}} [\tilde{\Sigma}]_{AD} = \frac{1}{ [\tilde{\Sigma}^{-1}]_{DD} }.
    \end{aligned}
\end{equation}
Therefore, the state of the data qubit after the homodyne measurement step is given by
\begin{equation}
       \sum_{i} p_i  \int d\bo v~\mathcal{N}(\boldsymbol{\mu}_i, \Delta^2 \Sigma^{'}, \bo v ) W(\bo v) \ket{\bar\psi}_D \bra{\bar\psi}_D W(\bo v)^{\dg}. 
\end{equation}

We note that the conditional covariance matrices are identical and independent of the measured syndrome. This fact is crucial since it shows that the role of a decoder is simply determine the correction displacement that we will apply. Therefore, the decoder only shifts the error profile, and cannot change the covariance matrix $\Sigma^{'}$. 

We now prove that $\det(\Delta^2 \Sigma^{'}) = \Delta^4$ (or equivalently $\det(\Sigma')=1$) by utilizing the following fact, which holds for any general invertible $2\times 2$ block matrix
    \begin{align}
   \det(A) =  \frac{\det( [A]_{11})}{\det( [A^{-1}]_{22} )}.
\end{align}
If we identify $A = \Delta^2 \Sigma$, $[A]_{11}= \Delta^2 \tilde{\Sigma} $, and $[A^{-1}]_{22} = \frac{[\Sigma^{-1}_{pp}]_{AA}}{\Delta^2} = \frac{[\Sigma_{qq}]_{AA}}{\Delta^2}$ , the determinant of the marginalized covariance matrix $\tilde{\Sigma}$ is given by 
\begin{subequations}
    \label{eq: det marg cov 1}
\begin{align}
    \det(\Delta^2 \tilde{\Sigma}) &= \det( \frac{[\Sigma^{-1}_{pp}]_{AA}}{\Delta^2} ) \det(\Delta^2 \Sigma)\\
    &= \frac{\det([\Sigma_{pp}]_{AA})}{(\Delta^{2})^{N}}(\sigma^{2})^{2N+2}\\
    &= (\Delta^2)^{N+2} \det( [\Sigma_{qq}]_{AA} )
\end{align}
\end{subequations}
Applying the same formula for $A = \Delta^2 \tilde{\Sigma}$ in Eq. \eqref{eq: marginal cov block form} and using the Bayesian rules from Eq. \eqref{eq: bayesian update rule}, we also have
\begin{subequations}
    \label{eq: det marg cov 2}
    \begin{align}
    \det(\Delta^2 \tilde{\Sigma}) &= \frac{ \det(\Delta^2 [\tilde{\Sigma}]_{AA}) }{\det(\frac{[\tilde{\Sigma}^{-1}]_{DD}}{\Delta^2})}\\
    &= (\Delta^2)^N  \det( \Delta^2 \Sigma^{'}) \det([\Sigma_{qq}]_{AA} ).
    \end{align}
\end{subequations}
Comparing Eq. \eqref{eq: det marg cov 1} and \eqref{eq: det marg cov 2}, we conclude that 
\begin{equation}
    \det(\Delta^2 \Sigma^{'}) = \Delta^4.
\end{equation}

The decoder then decides which correction $\boldsymbol{\bar \mu}$ to apply based on some strategy. Therefore, after the correction steps, the error profiles on the data qubit state is 
\begin{equation}
    \label{eq: mixture of Gaussian channel}
    \sum_{i} p_i \mathcal{G}_{\boldsymbol{\tilde \mu}_i, \Delta^2 \Sigma^{'}}(\ket{\psi_D}\bra{\psi_D}), \quad \boldsymbol{\tilde \mu}_i = \boldsymbol \mu_i - \boldsymbol{\bar \mu}.
\end{equation} 
Since Eq.~\eqref{eq: mixture of Gaussian channel} is a mixture of Gaussian displacement channels with identical covariance matrices $\Delta^2 \Sigma'$, we intuitively expect that the covariance matrix of the random variable $\bo v$ with distribution function $\sum_{i} p_{i}~\mathcal{N}(\boldsymbol{\tilde \mu}_i,\Delta^2 \Sigma^{'},\bo v)$ to have larger generalized variance than $\Delta^2 \Sigma^{'}$, namely 
\begin{equation}
    \det(\mathbb{E} [ (\bo{v}- \mathbb{E}[\bo v]) (\bo{v}-\mathbb{E} [\bo v])^{T}  ]) \geq \det( \Delta^2 \Sigma^{'} ) = \Delta^4.
\end{equation}

This fact can be proven rigorously using the generalized Jensen inequality and the Minkowski determinant inequality. By direct computation, we can show that  
\begin{subequations}
    \begin{align}
        & \mathbb{E} [ (\bo{v}-\mathbb{E}[\bo v]) (\bo{v}-\mathbb{E}[\bo v])^{T}  ] \\
        & = \Delta^2 \Sigma'  + \sum_{i} p_i ( \boldsymbol{\tilde \mu}_i \boldsymbol{\tilde \mu}_i^{T} ) -   (\sum_{i}  p_i\boldsymbol{\tilde \mu}_i  )  (\sum_{j}  p_j\boldsymbol{\tilde \mu}_i^{T}) \\
        & = \Delta^2 \Sigma^{'} + \Sigma_{r}.
    \end{align}
\end{subequations}
The matrix $\Sigma_{r}$ is symmetric and also positive semi-definite. To establish the positive semi-definiteness of $\Sigma_{r}$, we will demonstrate that
\begin{equation}
    \bo{x}^{T} \Sigma_{r} \bo{x} \geq 0, \quad \forall \bo{x} \in \mathbb{R}^{2}.
\end{equation}
To this end, we define the function $\varphi_{\bo x}: \mathbb{R}^{2} \to \mathbb{R}$ for all $\bo x \in \mathbb{R}^{2}$:
\begin{equation}
    \varphi_{\bo x}(\boldsymbol{\tilde \mu}_i) = (\bo x^{T} \boldsymbol{\tilde \mu}_i)^2,
\end{equation}
which is a convex function. Therefore, using the generalized Jensen inequality, we can prove that
\begin{equation}
    \sum_{i} p_i  \varphi_{\bo x}(\boldsymbol{\tilde \mu}_i) -  \varphi_{\bo x}( \sum_i p_i \boldsymbol{\tilde \mu}_i) \geq 0, \quad \forall p_i \geq 0 \text{ and } \sum_i p_i = 1. 
\end{equation}
Since the left-hand side of the above inequality is simply $ \bo{x}^{T} \Sigma_{r} \bo{x}$ and the inequality holds for any $\bo x$, the covariance matrix $\Sigma_{r}$ is positive semi-definite.

As the matrices $\Sigma^{'}$ and $\Sigma_{r}$ are positive semi-definite matrices, the Minkowski determinant inequality implies
\begin{subequations}
    \begin{align}
            & \det(E[ (\bo{v}-E[\bo v]) (\bo{v}-E[\bo v])^{T}]) \\
            & \geq \det(\Delta^2 \Sigma^{'}) + \det( \Sigma_{r} ) \\
            & \geq \det( \Delta^2 \Sigma^{'} ) = \Delta^4 
    \end{align}
\end{subequations}
ending the proof. 
\end{proof}
We remark that the main result presented in Eq.~\eqref{eq: formal no-go theorem} continues to hold even when displacement operators are allowed in combination with Gaussian unitaries. Two modifications arise in the proof:
\begin{itemize}
    \item In Eq.~\eqref{eq: commute G through Us}, after commuting through the gates, the mean displacement of the resulting Gaussian displacement channel may no longer remain zero
    \item In Eq.~\eqref{eq: bayesian update rule}, the mean displacement of the Gaussian displacement channel acting on the data qubit after measuring the ancilla will be modified.
\end{itemize}
Aside from these shifts in the mean, the structure of the proof is unchanged. The key reasoning relies solely on the covariance matrix of the noise profile after measurement, as given in Eq.~\eqref{eq: bayesian update rule}, and this covariance matrix cannot be altered by displacement operators.

\section{Minimal Representations of Polynomial Phase Gates}
\label{appendix: representation logical phase gate}

In \cref{section: optimized logical gate representations}, we introduced the formalism of polynomial phase stabilizers and the coefficient reduction procedure to generate minimal polynomial phase gate representations of logical gates. In this appendix, we prove a number of results relating to the coefficient reduction procedure. These results can be summarized in the following theorem 
\begin{theorem}
    \label{theorem: minimal degree polynomial}
   The minimal polynomial $P_{\mathrm{min},m}({x})$, under the lexicographical order, that generates the logical gate $\Lambda_m$ has degree
    \begin{equation}
        \label{eq: first part of theorem minimal degree polynomial}
        \deg  P_{\mathrm{min},m}({x})  = m, 
    \end{equation}
    and $k$th-order coefficients $a_k$ that satisfy the bound
    \begin{equation}
        \label{eq: second part of theorem minimal degree polynomial}
        |a_k| \leq \frac{1}{2(k!)}.
    \end{equation}
\end{theorem}
To prove this, we begin by proving that the coefficient reduction procedure always returns a minimal polynomial under the lexicographic partial ordering introduced in \cref{section: optimized logical gate representations}, thereby establishing the second part of Theorem~\ref{theorem: minimal degree polynomial}. Appendix~\ref{section: hierachy of minimal polynomial for Lambdam} demonstrates that the minimal polynomials $P_{\mathrm{min},m}(x)
$ form a hierarchy satisfying $\deg(P_{\mathrm{min},m+1}) > \deg(P_{\mathrm{min},m})$, which serves as a stepping stone for the following section. Finally, in Appendix~\ref{section: induction proof for minimal polynomial of LambdaM}, we use mathematical induction to show that if $\deg(P_{\mathrm{min},m}) = m$, then $\deg(P_{\mathrm{min},m+1}) = m+1$. Together with the base case $\deg(P_{\mathrm{min},1}) = 1$, this completes the proof of the first part of Theorem~\ref{theorem: minimal degree polynomial}.

A recent result \cite{Hahn2025} establishes that non-Clifford gates in the GKP code cannot be implemented \textit{deterministically} using Gaussian operations alone. Our theorem complements this by identifying the minimal non-Gaussian resources required to realize a given logical non-Clifford gate with a single polynomial phase gate.

We remark that if ${x}$ denotes either ${q}$ or ${p}$, corresponding to ${Z}$-type or ${X}$-type phase gates in the GKP code, \cref{theorem: minimal degree polynomial} can be readily proven using the Clifford hierarchy and induction. However, the proof we present below is more general, as it also applies to other families of gates, such as the Hadamard-type gates that lie outside the Clifford hierarchy, and is applicable to other bosonic codes.

Throughout this appendix, we write
\begin{equation}
W(x) \equiv Q(x)
\end{equation}
to denote that the polynomials $W(x)$ and $Q(x)$ are equivalent up to the addition of a integer-valued polynomial $\sum_{i}c_{i}L_{i}(x)$ for $c_{i}\in\mathbb{Z}$. One simple fact that we will repeatedly use is as follows.
\begin{lemma}\label{lem: minimum leading coefficient}
    For any integer-valued polynomial $P(x)\equiv 0$ of degree $m$, its leading coefficient is an integer multiple of $1/m!$.
\end{lemma}
The proof is fairly simple but is important to show that linear combinations of higher-degree generating polynomials cannot be combined to create a lower-degree polynomial with a smaller leading coefficient.
\begin{proof}
    This is a direct consequence of the generating set of integer-valued polynomials $\{L_{i}(x) \}$ defined in ~\cref{eq: trivial polynomial stabilizers}. Because $\{L_{i}(x) \}$ is a generating set, we can always write
    \begin{equation}
        P(x)=\sum_{j=1}^{J}c_{j}L_{j}(x),
    \end{equation}
    for some $J\in\mathbb{Z}$ and $c_{j}\in\mathbb{Z}$ where we assume without loss of generality that $c_{J}\neq 0$. Now, note that the degree of $P(x)$ is $J$, and therefore we have $J=m$. Because the leading coefficient of $L_{m}(x)$ is $1/m!$, the lemma follows.
\end{proof}

\subsection{Optimality of the coefficient reduction procedure}
\label{sec: proof of coefficient-reduction procedure}

We first prove that the coefficient reduction procedure is guaranteed to return a minimal polynomial under the lexicographical order. Consequently, starting from any polynomial $P_{\rm start}(x)$ of the $\Lambda_{m}$ gate, we can always obtain a minimal polynomial $P_{\text{min},m}(x)$ by applying the coefficient reduction procedure. Since the resulting polynomial from this procedure always has coefficients $a_k$ that satisfy Eq.~\eqref{eq: second part of theorem minimal degree polynomial}, this establishes the second part of Theorem~\ref{theorem: minimal degree polynomial}.

We begin by present how to generate a starting polynomial representation $P_{\text{start}}(x)$ for any $\Lambda_m$ gate. We note that for the ${\Lambda}_1$ gate, there is a minimal representation 
\begin{equation}
    \label{eq: definition F1}
    P_{\text{min},1}(x) = \frac{x}{2},
\end{equation}
which, in the usual case where $x={q}/\sqrt{\pi}$, corresponds to the logical ${Z}=e^{i\sqrt{\pi}{q}}$ gate via \cref{eq: logical canonical form}. In the rest of this appendix, we assume that the polynomial only takes integer inputs; of course, when it is turned into a polynomial phase gate the polynomial will take real values again. Let us assume that we have found a representation $P_{m}$ for the ${\Lambda}_m$ gate. Then, a representative $P_{\text{start}}(x)$ for the $\Lambda_{m+1}$ gate can be constructed as
\begin{equation}
    \label{eq: recursive formula Fm}
    P_{\text{start}}(x) = 2^{m-1} P_{m}^2(x),
\end{equation}
which is a generalization of a result from Ref.~\cite{Royer2022}. 

The recursive formula presented in Eq.~\eqref{eq: recursive formula Fm} can be explained as follow. Since $P_{m}$ is a polynomial whose polynomial phase gate implements ${\Lambda}_m$, we can write $P_{m}(x)$ as 
\begin{equation}
    \label{eq: decomposing Fm}
    P_{m}(x) = n(x) + \frac{a(x)}{2^{m}},
\end{equation}
where $n(x)$ is some integer depending on $x$ and $a(x)=0$ for even $x$ and $a(x) = 1 $ for odd $x$. Substituting Eq.~\eqref{eq: decomposing Fm} into Eq.~\eqref{eq: recursive formula Fm}, we obtain
\begin{subequations}
    \begin{gather}
        P_{\text{start}}(x) = 2^{m-1} n^2(x) +n(x) a(x) + \frac{a(x)^2}{2^{m+1}}   \\
        \Rightarrow P_{\text{start}}(x)\;\;\mathrm{mod}\, 1  = \frac{a(x)^2}{2^{m+1}}
    \end{gather}
\end{subequations}
proving that $P_{\text{start}}(x)$ constructed from Eq.~\eqref{eq: recursive formula Fm} is indeed a polynomial representation of $\Lambda_{m+1}$. 

From the definition of $P_{\text{min},1}(x)$ in Eq.~\eqref{eq: definition F1}, the explicit formula for $P_{m}(x)$ is
\begin{equation}
        \label{eq: trivial representation of Lambda_m}
         P_{m}(x) = \frac{ x^{2^{m-1}}}{2^m},
\end{equation}
Consequently, there exists a starting polynomial representation of the logical gate ${\Lambda}_m$ for all $m\geq 1$. These trivial polynomial representations can be used as the starting polynomial $P_{\rm start}(x)$ for the coefficient reduction procedure. 

Now, we prove the following lemma relating to the optimality of the coefficient reduction procedure itself.
\begin{lemma}
    \label{lemma: guarantee of coefficient reduction procedure}
    Let $P_{\rm start}(x)$ be any polynomial phase gate representation of the logical gate ${O}$ 
    \begin{equation}
    P_{\rm start}(x) = \sum_{i=0}^{N} a_i x^{i}.
    \end{equation}
    Applying the coefficient reduction procedure to $P_{\rm start}(x)$ always yields a minimal polynomial phase gate representation of the logical gate ${O}$ under the lexicographic order, independent of the starting polynomial $P_{\rm start}(x)$.
\end{lemma}

\begin{proof}
    From the definition of the lexicographic order, if we can prove that at a given order $j$ the coefficient reduction procedure produces the minimal coefficient in $j$, then we know that the overall procedure will output a minimal polynomial.
    
    Due to the way we define the coefficient reduction procedure, at each order $j$ the coefficient reduction procedure outputs a coefficient $r_{j}$ satisfying $|r_{j}|\leq 1/(2j!)$. From \cref{lem: minimum leading coefficient}, the output coefficient is unique and cannot be reduced by any other integer-valued polynomial unless the remainder $r_{j} = \pm \frac{1}{2j!}.$ In the case that the remainder $r_{j} = \pm \frac{1}{2j!}$, there are two possible branches in the procedure. There is no guarantee which branch will return the minimal polynomial under the lexicographic order. Therefore, we need to keep track of both branches and continue the procedure. Subsequent iterations might generate additional branches that we also need to keep track of. At the end of the procedure, choose the branch with the minimal polynomial. This guarantees that the coefficient reduction procedure always return the minimal polynomial. 
\end{proof}

\subsection{Hierarchy of minimal polynomials $P_{\mathrm{min},m}$} 
\label{section: hierachy of minimal polynomial for Lambdam}
From \cref{lemma: guarantee of coefficient reduction procedure}, we now have a constructive method to generate a minimal polynomial representation $P_{\mathrm{min},m}(x)$ corresponding to the polynomial phase gate $\Lambda_{m}$. We now prove some properties of the resulting minimal polynomial representations, beginning in this section by establishing that the minimal polynomials $P_{\mathrm{min},m}(x)$ form a hierarchy where the degree of $P_{\mathrm{min},m+1}(x)$ is strictly larger than the degree of $P_{\mathrm{min},m}(x)$. It will serve as a first step to the proof of Eq.~\eqref{eq: first part of theorem minimal degree polynomial} that we will finish in the following section. In the rest of this appendix, we denote the coefficients of the polynomial $P_{\mathrm{min},m}(x)$ as 
\begin{equation}
    P_{\mathrm{min},m}(x)=\sum_{j=1} p_{m,j}~ x^{j}.
\end{equation}

We begin with a simply corollary to \cref{lem: minimum leading coefficient}.
\begin{corollary}
\label{corollary: leading coefficient of stabilizer polynomial}
Let $P(x) $ and $Q(x)$ be polynomials with degree $n$ and $m$ respectively, where $n > m$. If $P(x) \equiv Q(x)$, then leading coefficient of $P(x)$ must be an integer multiple of $1/n!$.    
\end{corollary}
\begin{proof}
This result follows directly from Lemma~\ref{lem: minimum leading coefficient}. Since $P(x)- Q(x)\equiv 0$, Lemma~\ref{lem: minimum leading coefficient} implies that the leading coefficient of $P(x)- Q(x)$ must be an integer multiple of $1/n!$. Moreover, the leading coefficient of $P(x)-Q(x)$ coincides with the leading coefficient of $P(x)$ since $\deg P(x) > \deg Q(x)$, proving the result.
\end{proof}
Given Corollary~\ref{corollary: leading coefficient of stabilizer polynomial}, we are now ready to present and prove the main result of this section.
\begin{lemma}
    \label{lemma: hierachy of phase gate representation}
    The minimal polynomial $P_{\mathrm{min},m}(x)$ satisfies a hierarchy
    \begin{equation}
    \label{eq: hierachy of phase gate representation}
    \deg\big(P_{\mathrm{min},m+1}\big) > \deg\big(P_{\mathrm{min},m}\big).
    \end{equation}
\end{lemma}

\begin{proof}
Note that \cref{eq: hierachy of phase gate representation} is already true for $m=2$ since we know the minimal polynomials $P_{\mathrm{min},1}(x)=x/2$ and $P_{\mathrm{min},2}(x)=x^{2}/4$ by using the coefficient reduction procedure.

Therefore, working towards a proof by induction, we assume that Eq.~\eqref{eq: hierachy of phase gate representation} holds for $P_{\mathrm{min},m}(x)$, i.e.
\begin{equation}\label{eq: inductive assumption}
n = \deg\big(P_{\mathrm{min},m}\big) > \deg\big(P_{\mathrm{min},k}\big), \quad \forall k < m.
\end{equation}

In the following paragraph, we determine the leading coefficient of $P_{\mathrm{min},m}(x)$. Since the logical gates satisfy the recurrence relation $\Lambda_{m-1}=\Lambda_{m}^2 $, their polynomial representations obey the corresponding recurrence equation
\begin{equation}
\label{eq: recurrent relation for Pm}
2 P_{\mathrm{min},m}(x) \equiv P_{\mathrm{min},m-1}(x).
\end{equation}
Given our assumption that $P_{\mathrm{min},m}(x)$ is a minimal polynomial representation, its leading coefficient $p_{m,n}$ must satisfy the upper bound 
\begin{equation}
    \label{eq: upper bound to coefficient of minimal polynomial}
    |p_{m,n}| \leq \frac{1}{2 (n!)}.
\end{equation}
However, from the recurrence relation in Eq.~\eqref{eq: recurrent relation for Pm} we know that $2P_{\mathrm{min},m}(x)\equiv P_{\mathrm{min},m-1}(x)$. The inductive assumption \cref{eq: inductive assumption} therefore allows us to use Corollary~\ref{corollary: leading coefficient of stabilizer polynomial}, showing that the leading coefficient of $2P_{\mathrm{min},m}(x)$ must be
\begin{equation}
     \label{eq: leading coefficient of Pm}
    2|p_{m,n}| = \frac{1}{ n!}.
\end{equation}

Given the leading coefficient of $P_{\mathrm{min},m}(x)$ satisfies $|p_{m,n}|=1/(2n!)$, we now show that it leads to  
\begin{equation}
    \deg P_{\mathrm{min},m+1} > \deg P_{\mathrm{min},m}.
\end{equation}
Assume, for the sake of contradiction, that
\begin{equation}\label{eq: contradictive assumption 1}
    \text{deg}P_{\mathrm{min},m+1}  < \text{deg}P_{\mathrm{min},m}.
\end{equation}
From the recurrence relation in Eq.~\eqref{eq: recurrent relation for Pm}, we have
    \begin{equation}
         2 P_{\mathrm{min},m+1}(x) \equiv P_{\mathrm{min},m}(x). 
    \end{equation}
By the assumption~\cref{eq: contradictive assumption 1}, the left-hand side is a polynomial of degree strictly less than $n$, while the right-hand side has degree $n$. Corollary~\ref{corollary: leading coefficient of stabilizer polynomial} implies that the leading coefficient $|p_{m,n}|$ of $P_{\mathrm{min},m}(x)$ is $1/n!$, which is a contradiction to Eq.~\eqref{eq: leading coefficient of Pm}. Therefore, we conclude that
\begin{equation}
    \deg P_{\mathrm{min},m+1} \geq \deg P_{\mathrm{min},m}.
\end{equation}

We now eliminate the scenario in which the degree of $P_{\mathrm{min},m+1}(x)$ equals the degree of $P_{\mathrm{min},m}(x) $. Again, for the sake of contradiction, we assume that
\begin{equation}\label{eq: contradictive assumption 2}
    \deg P_{\mathrm{min},m+1} = \deg P_{\mathrm{min},m}.
\end{equation}
Under this assumption, the recurrence relation Eq.~\eqref{eq: recurrent relation for Pm} implies that
\begin{equation}
    2 p_{m+1,n} + j \frac{1}{n!} = p_{m,n} = \pm \frac{1}{2(n!)},
\end{equation}
where $j$ is some integer and $1/n!$ is simply the leading coefficient of the stabilizer $L_{n}(x)$. Together with the requirement that $|p_{m+1,n}| \leq 1/(2n!)$ from Eq.~\eqref{eq: upper bound to coefficient of minimal polynomial}, we can deduce that the leading coefficient of $P_{\mathrm{min},m+1}(x) $ must necessarily be 
 \begin{equation}\label{eq: contradictory leading coefficient}
        |p_{m+1,n}| = \frac{1}{4(n!)}.
 \end{equation}

In the following argument, we will show that Eq.~\eqref{eq: contradictory leading coefficient} cannot be true. We will make use of the identity
    \begin{equation}
        \label{eq: parity identity for Pm}
        \begin{aligned}
            P_{\mathrm{min},m+1}(x+2) - P_{\mathrm{min},m+1}(x) \equiv 0,
        \end{aligned}
    \end{equation}
which holds for all polynomial representations of the phase gate $\Lambda_{m+1}$.  This identity follows from the observation that $x$ and $x+2$ share the same parity, both are either even or odd. Consequently, the polynomials $P_{\mathrm{min},m+1}(x+2)$ and $P_{\mathrm{min},m+1}(x) $ represent the same logical gate.
%Since the left-hand side is equivalent to zero, it must lie in the integer linear span of $L_{j}(x)$ for $j<n$. 

However, the left-hand side (LHS) of Eq.~\eqref{eq: parity identity for Pm} cannot be an integer-valued polynomial.
%lie in the integer linear span of the integer-valued polynomials $\{ L_{j}(x) \}_{j < n}$.
To see this, we explicitly expand the LHS:
\begin{subequations}
\begin{align}
& P_{\mathrm{min},m+1}(x+2) - P_{\mathrm{min},m+1}(x) \\
&= \sum_{j=1}^{n} p_{m+1,j} \big[(x + 2)^j - x^j\big] \\
&= 2 n p_{m+1,n}  x^{n-1} + O(x^{n-2}) \\
&= \pm \frac{1}{2 (n - 1)!} x^{n-1} + O(x^{n-2}),
\end{align}
\end{subequations}
where we have used that $(x+2)^j - x^j$ is a polynomial of degree $j-1$, and in the last line we used~\cref{eq: contradictory leading coefficient}. Since the coefficient of the leading term $x^{n-1}$ is $\pm \frac{1}{2(n-1)!}$, Corollary~\ref{corollary: leading coefficient of stabilizer polynomial} tells us that the LHS of Eq.~\eqref{eq: parity identity for Pm} cannot be an integer-valued polynomial. This contradicts the assumption Eq.~\eqref{eq: contradictive assumption 2} . Therefore, we conclude that
\begin{equation}
\deg P_{\mathrm{min},m+1} > \deg P_{\mathrm{min},m}(x).
\end{equation}
\end{proof}

\subsection{Degree of the minimal polynomial $P_{\mathrm{min},m}(x)$}
\label{section: induction proof for minimal polynomial of LambdaM}

In the previous section, we have proven that the degree of the minimal polynomials $P_{\mathrm{min},m}(x)$ follows a hierarchy Lemma~\ref{lemma: hierachy of phase gate representation}. In this section, we will prove that the difference in the degrees of the polynomials $P_{\mathrm{min},m}(x)$ and $P_{\mathrm{min},m+1}(x)$ is one, which leads to the theorem.

\begin{lemma}\label{lem: minimal polynomial degree}
    The degree of any minimal polynomial phase gate representation of the $\Lambda_{m}$ gate satisfies
    \begin{equation}
        \mathrm{deg} P_{\mathrm{min},m}=m.
    \end{equation}
\end{lemma}

\begin{proof}
We will use induction to prove this lemma; the base case is $\text{deg}(P_{\mathrm{min},1} )=1$. Before continuing, we will present two facts about the basis of integer-valued polynomials $\{ L_{n }(x) \}$ that we will use. First, if $n$ is odd, then $L_{n}(x)$ contains only odd powers of $x$. This can be shown by noting that for odd $n$ the integer-valued polynomial $L_{n}(x)$ from~\cref{eq: definition of integer value poly} can be written as 
\begin{equation}
    L_{n}(x) = \frac{x}{n!} \prod_{i=1}^{(n-1)/2} (x^2-i^2)
\end{equation}
which contains only odd powers of $x$. Second, when $n$ is even, the second-leading coefficient of $L_{n}(x)$ is nonzero and given by
\begin{equation}
    \label{eq: second leading coeff of even Kn}
    \ell_{n,n-1}= \frac{1}{2(n-1)!},
\end{equation}
where we have written $L_{n}(x)=\sum_{i}\ell_{n,i} x^{i}$.
This can be seen by noting that for even $n$, the stabilizer polynomial $L_{n}(x) $ can be expressed as 
\begin{subequations}
    \begin{align}
        L_{n}(x) &= \frac{ (x+n/2)}{n} \frac{x \prod_{i=1}^{n/2-1} (x^2-i^2)}{ (n-1)!}\\
        &= \Big( \frac{x}{n} + \frac{1}{2} \Big) L_{n-1}(x).
    \end{align}
\end{subequations}
Expanding this product shows that the two leading coefficients of $L_{n}(x)$ are
\begin{equation}
    \ell_{n,n} = \frac{1}{n!},~\quad \ell_{n,n-1}= \frac{1}{2(n-1)!}. 
\end{equation}
    
As our inductive assumption, we assume that the degree of $P_{\mathrm{min},m}(x)$ is $m$. For convenience, we denote the degree of $P_{\mathrm{min},m+1}(x)$ as 
\begin{equation}
    \text{deg}(P_{\mathrm{min},m+1} ) = m+c,
\end{equation}
where $c\geq 1$ due to Lemma~\ref{lemma: hierachy of phase gate representation}. Our goal is to prove that $c=1$. To that end, we will repeatedly use the following two identities
\begin{subequations}
\begin{align}
\label{eq: identity 1 for Pm}
P_{\mathrm{min},m+1}(x) + P_{\mathrm{min},m+1}(-x) &\equiv P_{\mathrm{min},m}(x), \\
\label{eq: identity 2 for Pm}
P_{\mathrm{min},m+1}(x+2) - P_{\mathrm{min},m+1}(x) &\equiv 0.
\end{align}
\end{subequations}
These identities follow directly from the recurrence relation for the phase gates in Eq.~\eqref{eq: recurrent relation for Pm}, along with the fact that the action of $\Lambda_{m}$ on a $y$-eigenstate of the $x$ operator $\ket{y}_{x}$ is determined only by whether $y$ is even or odd and not on its specific value.

Our proof strategy is to show that if  $c > 1$, then either Eq.~\eqref{eq: identity 1 for Pm} or/and Eq.~\eqref{eq: identity 2 for Pm} must be violated. This contradiction will imply that $c=1$ is the only consistent possibility. To carry out the argument, we consider two separate cases depending on the parity of $m$: when $m$ is even and when $m$ is odd.
    
\textbf{Case 1: $m$ is an even number} 

\textbf{Subcase A: $c \notin \{2,4,\dots,\} $} 

As a first step, we show that $c$ cannot belong to the set of even integers greater than or equal to $2$, i.e $c \notin \{2,4,\dots \}$. Suppose, for the sake of contradiction, that $c$ belongs to this set. Then, Eq.~\eqref{eq: identity 1 for Pm} implies
\begin{subequations}
    \label{eq: violation for c even}
    \begin{align}
             P_{\mathrm{min},m}(x)  &\equiv   P_{\mathrm{min},m+1}(x) +  P_{\mathrm{min},m+1}(-x) \\
             &= \pm \frac{x^{m+c}}{(m+c)!} +  2 \sum_{i=2,4,\dots}^{m+c-2} p_{m+1,i} x^{i}\\
             &  \equiv 2 \sum_{i=2,4,\dots}^{m+c-2} p_{m+1,i} x^{i} \mp \sum_{i=1}^{m+c-1} \ell_{m+c,i} x^{i}.\label{eq: violation for c even last line}
    \end{align}
\end{subequations}
where we have used Eq.~\eqref{eq: leading coefficient of Pm} in the second line, and we have added or subtracted the polynomial $L_{m+c}(x)$ in the third line to remove the leading term on the second line.

The leading coefficient on Eq.~\eqref{eq: violation for c even last line} is $\ell_{m+c,m+c-1}$ and its value is given by Eq.~\eqref{eq: second leading coeff of even Kn}, which is \textit{not }an integer multiple of $1/(m+c-1)!$. On the other hand, the degree of Eq.~\eqref{eq: violation for c even last line} is $m+c-1$ which is greater than the degree of $P_{\text{min},m}(x)$. Therefore, Corollary~\ref{corollary: leading coefficient of stabilizer polynomial} implies that the leading coefficient on the third line of Eq.~\eqref{eq: violation for c even} must be integer multiple of $1/(m+c-1)!$. This leads to a contradiction and we conclude that $c$ cannot belong to the set of positive even integers.

\textbf{Subcase B: $c \notin \{3,5,\dots \} $} 

We now prove that $c$ cannot be in the set $ \{3,5,\dots \}$. Let us assume by contradiction that $c \in \{3,5,\dots \}$. Then Eq.~\eqref{eq: identity 1 for Pm} implies that 
\begin{subequations}
    \begin{align}
                \label{eq: m even c odd recurrence}
             P_{\mathrm{min},m}(x) & \equiv P_{\mathrm{min},m+1}(x) + P_{\mathrm{min},m+1}(-x) \\
             & = 2 \sum_{i=2,4,\dots}^{m+c-1} p_{m+1,i} x^{i} 
    \end{align}
\end{subequations}
Due to $P_{\mathrm{min},m+1}(x)$ being a minimal polynomial, by definition, there are three possibilities for $p_{m+1,m+c-1}$. They are 
\begin{subequations}
    \begin{align}
            \label{eq: first possibility}
             0  < 2 |p_{m+1,m+c-1}| & < \frac{1}{(m+c-1)!}, \\
            \label{eq: second possibility}
            2 |p_{m+1,m+c-1}|  & = \frac{1}{(m+c-1)!}, \\
            \label{eq: third possibility}
            p_{m+1,m+c-1}      & = 0. 
    \end{align}
\end{subequations}

For the first possibility Eq.~\eqref{eq: first possibility}, the leading coefficient on the right-hand side of Eq.~\eqref{eq: m even c odd recurrence} is not an integer multiple of $1/(m+c-1)!$ and thus we can again use Corollary~\ref{corollary: leading coefficient of stabilizer polynomial} to find a contradiction. 

For the second possibility Eq.~\eqref{eq: second possibility}, the right-hand side of Eq.~\eqref{eq: m even c odd recurrence} can be reduced using the integer-valued polynomial $L_{m+c-1}(x)$ which yields
\begin{equation}
    \label{eq: reduction by Kn trick}
    P_{\mathrm{min},m}(x)  \equiv  2 \sum_{i=2,4,\dots}^{m+c-3} p_{m+1,i} x^{i} - \sum_{i=1}^{m+c-2} \ell_{m+c-1,i} x^{i}. 
\end{equation}
However, the leading coefficient $\ell_{m+c-1,m+c-2}$ on the right-hand side is not an integer multiple of $1/(m+c-2)!$ due to Eq.~\eqref{eq: second leading coeff of even Kn}. Therefore, we can again use Corollary~\ref{corollary: leading coefficient of stabilizer polynomial} to find a contradiction. 
   
For the third possibility Eq.~\eqref{eq: third possibility}, we apply Eq.~\eqref{eq: leading coefficient of Pm} together with Eq.~\eqref{eq: identity 2 for Pm}. This yields  
\begin{subequations}
    \begin{align}
            0 & \equiv P_{\mathrm{min},m+1}(x+2)- P_{\mathrm{min},m+1}(x) \\
            & = \pm \frac{x^{m+c-1}}{(m+c-1)!} \pm \frac{x^{m+c-2}}{(m+c-2)!} + O(x^{m+c-3}) \\
            % & \equiv \Big( \pm \frac{1}{(m+c-2)!} \mp k_{m+c-1,m+c-2} \Big) x^{m+c-2}  + O(x^{m+c-3}) \\
            & \equiv \mp \ell_{m+c-1,m+c-2} x^{m+c-2}  + O(x^{m+c-3}),
    \end{align}
\end{subequations}
where on the second line we have used that $|p_{m+1,m+c}|=1/(2(m+c)!)$ due to Eq.~\eqref{eq: leading coefficient of Pm}, and to get to the third line we have subtracted the polynomials $L_{m+c-1}(x)$ and $L_{m+c-2}(x)$ respectively to eliminate the leading terms. The resulting leading coefficient on the final line is not an integer multiple of $1/(m+c-2)!$ due to Eq.~\eqref{eq: second leading coeff of even Kn}. Hence, we can use Lemma~\ref{lem: minimum leading coefficient} to find a contradiction.

In summary, $c$ cannot be in the set $\{3,5,\dots\}$ the only remaining possibility is $c=1$. In this case, the right-hand side of Eq.~\eqref{eq: m even c odd recurrence} matches the degree of the left-hand side, and thus we cannot use Corollary~\ref{corollary: leading coefficient of stabilizer polynomial}. Since we know from Section~\ref{sec: proof of coefficient-reduction procedure} that $P_{\mathrm{min},m+1}(x)$ exists,  we conclude that if $m$ is even, then $\text{deg}\Big( P_{\mathrm{min},m+1}(x) \Big) = m+1$.

\textbf{Case 2: $m$ is an odd number} 

The proof for the case where $m$ is odd follows a similar structure to that of the even case. However, the order of the arguments is reversed: we first show that $c \notin  \{3,5,\dots \}$, and then show that $c \notin  \{2,4,\dots \}$. The logic of each step remains analogous, with appropriate adjustments to account for the parity of $m$.
\end{proof}

Finally, Lemmas~\ref{lem: minimal polynomial degree} and \ref{lemma: guarantee of coefficient reduction procedure} trivially combine to give the main theorem, Theorem~\ref{theorem: minimal degree polynomial}, as promised.

\section{Representation of multi-qubit polynomial phase gates}
\label{appendix: optimal representation of multi-qubit gates}
In this appendix, we present a method for constructing minimal representations of multi-qubit gates in the square GKP code using multi-mode polynomial phase stabilizers. These results extend the single-qubit polynomial stabilizer formalism to the multi-mode setting, highlighting its generality. We will begin by outlining how to construct \textit{a} multi-mode polynomial phase gate that can implement a given logical diagonal multi-qubit gate. In particular, we focus on the multi-controlled phase gate  $C^{N-1}{\Lambda_m}$ acting on $N$ qubits, which applies a phase $\mathrm{exp}(i2\pi/2^{m})$ to the logical state $\ket{\overline{1\dots 1}}$ while leaving all other logical computational basis states unchanged. Any diagonal logical multi-mode gate can then be generated by a product of $C^{N-1}\Lambda_{m}$ gates. After this, we will provide a characterization of the multi-mode polynomial phase stabilizers (which correspond to integer-valued polynomials on multiple variables) and describe the multi-mode coefficient reduction procedure that returns a minimal representation of a given multi-mode polynomial phase gate.

We begin by constructing \textit{a} multi-mode polynomial phase gate that implements the logical $C^{N-1}\Lambda_{m}$ gate. As in Appendix~\ref{appendix: representation logical phase gate}, we start with the base case case $C^{N-1}{\Lambda}_1$ and then recursively build higher-order gates $C^{N-1}{\Lambda}_m$ from it.

We first claim that the $C^{N-1}\Lambda_1$ gate can be implemented by the polynomial 
\begin{equation}
    P[C^{N-1}{\Lambda}_1 ](\bo{x}) = \frac{x_1 x_2 \dots x_N }{2},
\end{equation}
where $\bo{x}=(x_1,\dots,x_n)$. To see this, we decompose each variable $x_i$ as
\begin{equation}
    x_i = 2 n_i + r_i,
\end{equation}
with $n_i \in \mathbb{Z}$ and $r_i \in \{0,1\}$, indicating whether $x_i$ is even ($r_i=0$) or odd ($r_i=1$). Substituting into $P[C^{N-1}{\Lambda}_1 ](\bo{x})$ gives
\begin{equation}
P[C^{N-1}{\Lambda}_1 ](\bo x) \bmod 1 = \frac{r_1 r_2 \cdots r_N}{2},
\end{equation}
which implements the ${\Lambda}_1$ gate precisely when all $x_i$ are odd, i.e.~when all encoded qubits are in the logical state $\ket{\bar{1}}$.

Next, we construct a (non-optimized) representation of the logical $C^{N-1}\Lambda_{m+1}$ gate from $P[C^{N-1}\Lambda_{m}]$. The key is to reuse the recursion formula from Eq.~\eqref{eq: recursive formula Fm}, suitably generalized to multi-qubit gates. In this setting, we require that each representation $P[C^{N-1}\Lambda_{m}](\bo x)$ gate must be able to be decomposed as
\begin{equation}
    P[C^{N-1}\Lambda_{m}](\bo x) = n(\bo x) + \frac{a(\bo x)}{2^{m}},
\end{equation}
where $n(\bo x)$ is again some integer depending on $\bo x$ but now the parity term satisfies
\begin{equation}
    a(\bo x) = \begin{cases}
        1 \text{ if } x_1 =\dots=x_N=1 \\
        0 \text{ otherwise }
    \end{cases} \quad .
\end{equation}
Consequently, we obtain the general recursive formula for the $C^{N-1}\Lambda_m$ gate
\begin{equation}
    \label{eq: multi-qubit control gate recursive}
    P[C^{N-1}\Lambda_{m+1}](\mathbf{x}) = \frac{(x_1 \dots x_N)^{2^{m}}}{2^{m+1}};
\end{equation}
the proof of \cref{eq: multi-qubit control gate recursive} follows the same steps as in Appendix~\ref{sec: proof of coefficient-reduction procedure}.  

Next, we need to identify a complete basis for the integer-valued polynomials on $N$ variables. Such a basis is given by the set 
\begin{equation}\label{eq: multi-variate integer-valued poly basis}
    \big\{L_{i_1}(x_{1})L_{i_2}(x_{2})\dots L_{i_N}(x_{N})\big\}.
\end{equation}
To see why this is indeed a complete basis, note that we already know this is true for $N=1$, so preparing for an inductive argument let us assume it is true for $N=n$. Now suppose we have an integer-valued polynomial $P(\vect{x})$ on $N=n+1$ variables; that is, it satisfies $P(\mathbb{Z}^{n+1})\subseteq\mathbb{Z}$. Viewing $P(\vect{x})$ as a univariate integer-valued polynomial of the variable $x_{n+1}$, we can always write it as
\begin{equation}
    \label{eq: decomposing multivariate polynomial into univariate}
    P(\vect{x})=\sum_{i_{n+1}}c_{i_{n+1}}(x_{1},\dots,x_{n})L_{i_{n+1}}(x^{n+1}),
\end{equation}
with $c_{i_{n+1}}(x_{1},\dots,x_{n})$ being an integer-valued polynomial using the completeness of the univariate integer-valued polynomial basis~\cref{eq: definition of integer value poly}.

Note, by our inductive assumption, the integer-valued polynomial $c_{i_{n+1}}(x_{1},\dots,x_{n})$ can be decomposed into
\begin{equation}
    \label{eq: decomposing multinomial coefficient}
    c_{i_{n+1}}(x_{1},\dots,x_{n}) = \sum_{i_1,\dots,i_n} c_{i_1,\dots,i_n}  L_{i_1}(x_1) \dots L_{i_n}(x_n),
\end{equation}
with the coefficient $c_{i_1,\dots,i_n}$ being integers. Therefore, by substituting Eq.~\eqref{eq: decomposing multinomial coefficient} back to Eq.~\eqref{eq: decomposing multivariate polynomial into univariate}, we find that $P(\vect{x})$ can always be written as
\begin{equation}
P(\vect{x})=\sum_{i_{1},\dots,i_{n+1}}c_{\vect{i}}\big(L_{i_{1}}(x_{1})\dots L_{i_{n+1}}(x_{n+1})\big),
\end{equation}
for some integer coefficients $c_{\vect{i}}\in\mathbb{Z}$. This shows that \cref{eq: multi-variate integer-valued poly basis} is indeed a complete basis of the multi-variate integer-valued polynomials.

Finally, given a polynomial phase gate $P[C^{N-1}\Lambda_{m}](\vect{x})$, we minimize the representation using a multi-variate coefficient reduction procedure as follows. In particular, the coefficient of each monomial $x_{1}^{d_{1}}\dots x_{N}^{d_{N}}$ in $P[C^{N-1}\Lambda_{m}](\vect{x})$ will be reduced by adding or subtracting a suitable integer multiple of the integer-valued polynomial $L_{d_{1}}(x_{1})\dots L_{d_{N}}(x_{N})$. However, defining the order in which these coefficients are reduced is crucial for the coefficient reduction procedure to work correctly.

In particular, for any monomial $x_{1}^{d_1}\dots x_{N}^{d_N}$ that appears in the polynomial $ P[C^{N-1}\Lambda_m]$, we define its degree as 
\begin{equation}
    \text{deg}(x_{1}^{d_1}\dots x_{N}^{d_N}) = d_1 + d_2 + \dots + d_N. 
\end{equation}
The degree of the polynomial $P[C^{N-1}\Lambda_m](\mathbf{x})$ is simply the maximal degree of its terms. The key difference in this multi-variate coefficient reduction procedure is that, for any given degree $D$ of the polynomial $F_m[C^{N-1}\Lambda_m](\mathbf{x})$, there can be multiple monomials of the form 
\begin{equation}
    x_{1}^{d_1}\dots x_{N}^{d_N}
\end{equation}
such that
\begin{equation}
    d_1 + \cdots + d_N = D.
\end{equation}

The crucial observation is that the coefficients of any two monomials with the same degree can be independently reduced, without the reduction of one of the monomials affecting the reduction of the other. This follows from the fact that the polynomial $L_{d_{1}}(x_{1})\dots L_{d_{N}}(x_{N})$ only contains monomials $x^{m_{1}}_{1}\dots x^{m_{N}}_{N}$ where \textit{every} $m_{n}$ satisfies $m_{n}\leq d_{n}$. This greatly simplifies the coefficient reduction procedure: we simply begin by reducing the monomials in $P[C^{N-1}\Lambda_{m}]$ that have the highest degree (in any order), before reducing the monomials with the next highest degree (in any order), and so on until we have minimized all the coefficients in the polynomial. The output $P_{\text{min}}[C^{N-1}\Lambda_{m}](\vect{x})$ will once again be minimal under a suitably generalized version of the lexicographic order defined in \cref{section: optimized logical gate representations}.

As an example, consider the controlled-$S$ gate, $CS \equiv C^{1}\Lambda_2$. From Eq.~\eqref{eq: multi-qubit control gate recursive}, the recursive formula gives
\begin{equation}
P[CS](x_1,x_2)
 = \frac{x_1^2 x_2^2}{4}.
\end{equation}
This representation, however, is reducible by the stabilizer $L_2(x_1)L_2(x_2)$. Subtracting this term yields the optimized form
\begin{equation}
P_{\text{min}}[CS](x_1,x_2)
 = -\frac{x_1^2 x_2 + x_1 x_2^2 + x_1 x_2}{4}.
\end{equation}
This irreducible representation achieves the minimal possible coefficients and interaction order; hence, constitutes a minimal representation for the $CS$ gate. The representation for the the $CCZ\equiv C^{2}\Lambda_{1}$ gate given in \cref{eq: CCZ gate} can be derived in the same way.

\section{Statistical moment of propagated errors}
\label{appendix: statistical moment of errors}
In this appendix, we present a quantitative analysis of the propagated error operator $\mathcal{E}(\cdot)=E\cdot E^{\dag}$ induced by a general polynomial phase gates $\exp[i2\pi P({q}/\sqrt{\pi})]$, filling in the gaps in our derivation in the main text in \cref{sec: error propagation and fault-tolerant theorem}. As a reminder, we define 
\begin{gather}
    \Lambda=\mathrm{exp}\big[i2\pi P(q/\sqrt{\pi})\big],\quad P(x)=\sum_{j=1}^{n}a_{j}x^{j},\tag{\ref{eq: logical operator definition}}\\
    {E} = \Lambda\,\mathrm{Env}_{\Delta_{q},\Delta_{p}}\,\Lambda^{\dag},\quad \mathcal{E}(\cdot)=E\cdot E^{\dag}\tag{\ref{eq: E definition}}
\end{gather}
and we will derive \cref{eq: mean scaling,eq: covariance matrix scaling}, which are asymptotic expressions for the covariance and mean of the probability density function $p_{\mathcal{E}(\vect{v})}$ corresponding to the characteristic function of the twirled channel $\mathrm{Twirl}[\mathcal{E}]$~\cref{eq: twirling definition}.

In Appendix~\ref{section: Twirling approximation of the propagated error}, we perform the twirling approximation and write down an explicit expression for $p_{\mathcal{E}}(\bo v)$. To characterize the noise induced by the $\text{Twirl}[\mathcal{E}]$, we compute the first and second statistical moments—the mean displacement and covariance matrix—of $p_{\mathcal{E}}(\bo v)$, as detailed in Appendix~\ref{appendix: Conditional statistical moments of the twirled error operato}. These quantities serve as proxies for the gate infidelity, since the exact error rate is generally analytically intractable. In Appendix~\ref{app: optimal bias estimate}, we derive our main result: the asymptotic scaling $\Delta \to 0$ of these statistical moments in terms of both the polynomial degree defining the logical gate and its coefficients. This analysis confirms our intuition that the highest-order interaction term sets the fundamental limit on gate fidelity, and provides a quantitative justification for the lexicographic ordering introduced in the main text. Finally, in Appendix~\ref{appendix: scaling of covariance matrices at finite squeezing}, we analyze the scaling of the first and second moments of $p_{E}(\bo v)$ at finite $\Delta$ and show that minimal polynomial phase gates $\Lambda_{n}[q/\sqrt{\pi}]$ with high degree $n$ \textit{can} have lower gate infidelity compared to the minimal cubic phase gate, providing a theoretical explanation for the results in \cref{fig:Lambda_gates}.

\subsection{Twirling approximation of $E$}
\label{section: Twirling approximation of the propagated error}

Throughout this appendix, we employ three approximations to simplify the analytical calculation. First, we approximate the characteristic function of the biased envelope operator using the Taylor approximation $\mathrm{tanh}(x)\approx x$ given in Eq.~\eqref{eq: envelope characteristic function approx}. Second, we assume that the GKP state is strongly squeezed in the position quadrature so that $v_{q } \propto O(\Delta_q)$ and make the approximation 
\begin{equation}\label{eq: first-order approximation}
    P_n({x} + \sqrt{2} v_q) - P_n({x}) \approx \sqrt{2} v_q \frac{d P_n({x})}{d {x}} \equiv \sqrt{2} v_q  P_{n}'(x).
\end{equation}
The third approximation is used in Appendix~\ref{app: optimal bias estimate}, where we estimate the optimal asymmetry $\lambda_{\rm opt}$ by minimizing the propagated shear error, following the same approach as in Section~\ref{sec: error propagation and fault-tolerant theorem}.

While we could proceed without resorting to this approximation, it significantly simplifies the algebra and still provides us the dominant contribution in the error operator. Under these two approximations, the characteristic function of the propagated error ${E}$ is given by 
\begin{multline}\label{eq: characteristic of general Pn propagated error}
    \chi_{{E}}(\bo v)= \frac{ \mathcal{N}(0,\Delta_{q}^2,v_q)} {\sqrt{2\pi}}\\
    \times\int dx e^{-i \sqrt{2\pi} v_p  x}  e^{- \frac{\Delta_p^2 x^2}{2}  }  e^{i 2\sqrt{2} \pi v_q P^{'}_{n}(\frac{x}{\sqrt{\pi}})  },
\end{multline}
which we derive more explicitly in \cref{appendix: proof of fault-tolerance theorem} in \cref{eq: characteristic of P3 propagated error}.

Performing the twirling approximation, the error channel $\mathcal{E}$ becomes a random displacement channel 
    \begin{equation} 
        \label{eq: unnormalized error channel}
         \text{Twirl}[\mathcal{E}](\rho) =  \int d \bo v \frac{(\mathcal{N}[0,\Delta_q^2,v_q])^2 f(v_p|v_q)}{2 \pi}  {W}(\bo v) \rho {W}  (\bo v)^{\dg}.
    \end{equation}
where we have denoted
\begin{equation}
    f(v_p|v_q) = \Big|  \int dx e^{-i \sqrt{2\pi} v_p  x}  e^{- \frac{\Delta_p^2 x^2}{2}}  e^{i 2\sqrt{2} \pi v_q P^{'}_{n}(\frac{x}{\sqrt{\pi}})  } \Big|^2.
\end{equation}
As explained in the main text, we can interpret the integrand
\begin{equation}
     \frac{(\mathcal{N}[0,\Delta_q^2,v_q])^2 f(v_p|v_q)}{2 \pi}
\end{equation}
as an \textit{unnormalized} probability density function $p_{{E}}(\bo v)$ of the displacement error $\bf{v}$ on the data qubit. In the following, we will find the normalization constant of this twirled channel $\text{Twirl}[\mathcal{E}]$, which is defined as 
\begin{equation}
    \label{eq: definition of renormalization constant after twirling}
   \int d \bo v \frac{(\mathcal{N}[0,\Delta_q^2,v_q])^2 f(v_p|v_q)}{2 \pi}   .
\end{equation}

Our first step is to rewrite $ f(v_p|v_q)$ in a more illuminating form
\begin{subequations}
    \begin{align}
        & f(v_p|v_q) = \Big|  \int~ dx e^{-i \sqrt{2\pi} v_p  x}  e^{- \Delta_p^2 x^2  }  e^{i 2\sqrt{2} \pi v_q P^{'}_{n}(\frac{x}{\sqrt{\pi}})  } \Big|^2 \\\
        & = \int~ dx_1~ e^{-i \sqrt{2\pi} v_p  x_1}  e^{- \Delta_p^2 x_1^2  }  e^{i 2\sqrt{2} \pi v_q P^{'}_{n}(\frac{x_1}{\sqrt{\pi}})  } \nonumber\\
        & \quad\times \int~ dx_2~ e^{i \sqrt{2\pi} v_p  x_2}  e^{- \Delta_p^2 x_2^2  }  e^{-i 2\sqrt{2} \pi v_q P^{'}_{n}(\frac{x_2}{\sqrt{\pi}})  }\\
        & = \iint dx_+ dx_-~ e^{i 2 \sqrt{2} \pi v_q\Big( P_{n}^{'}\Big(\frac{x_+ + x_-}{\sqrt{2}}\Big)-P_{n}^{'}\Big(\frac{x_+ - x_-}{\sqrt{2}}\Big)  \Big)  }  \nonumber\\
        & \qquad\times e^{- \frac{\Delta_p^2}{2}( x_+^2 + x_-^2 )} e^{-i 2 \sqrt{\pi} v_p x_-}
    \end{align}
\end{subequations}
where on the third equality, we have performed a transformation of variable with unit Jacobian  
\begin{equation}
    x_{+} = \frac{x_1 +x_2}{\sqrt{2}},~ x_{-} = \frac{x_1 - x_2}{\sqrt{2}}.
\end{equation}
For clarity, we denote
\begin{equation}
    [\delta P'_{n}](x,y) =  P'_{n}\Big(\frac{x + y}{\sqrt{2 }} \Big) - P^{'}_{n}\Big(\frac{x-y}{\sqrt{2}}\Big).
\end{equation}

The integral over $x_{+}$ and $x_{-}$ cannot be evaluated analytically. However, the integral in Eq.~\eqref{eq: definition of renormalization constant after twirling} can be computed exactly by changing the order of integration. Specifically, we first integrate over $v_p$, which yields a Dirac delta function $\delta(2\sqrt{\pi} x_-)$. We then successively integrate over $x_-$, $x_+$, and finally $v_q$. After some straightforward but tedious algebra, the normalization constant is found to be
\begin{equation}
    \int d \bo v \frac{(\mathcal{N}[0,\Delta_q^2,v_q])^2 f(v_p|v_q)}{2 \pi}  = \frac{1}{2 \Delta_q \Delta_p}.
\end{equation}
Dividing Eq.~\eqref{eq: unnormalized error channel} by this normalization factor yields the normalized error channel $\mathrm{Twirl}[\tilde{\mathcal{E}}]$
\begin{subequations}
    \label{eq: normalized effective poly err channel}
    \begin{align}
        \text{Twirl}[\tilde{\mathcal{E}}](\rho) & = (2\Delta_q \Delta_p) \int d \bo v \frac{(\mathcal{N}[0,\Delta_q^2,v_q])^2   f(v_p|v_q)}{2 \pi}  \nonumber\\
        & \qquad\qquad\qquad\qquad\times {W}(\bo v) \rho {W}  (\bo v)^{\dg} \\
        & =  \int d \bo v \Big[ \sqrt{2} \Delta_q  (\mathcal{N}[0,\Delta_q^2,v_q])^2 \Big]   \Big(\Delta_p \frac{ f(v_p|v_q)}{\sqrt{2} \pi} \Big)   \nonumber\\
        & \qquad\qquad\qquad\qquad\times  {W}(\bo v) \rho {W} (\bo v)^{\dg}  \\
        & = \int d \bo v~\mathcal{N}\Big[0,\frac{\Delta_q^2}{2},v_q \Big] \Big[\frac{1}{\pi} \frac{\Delta_p}{\sqrt{2}}   f(v_p|v_q) \Big] \nonumber\\
        & \qquad\qquad\qquad\qquad\times {W}(\bo v) \rho {W} (\bo v)^{\dg} .
    \end{align}
\end{subequations}
In the transformation from the second to the third equality, we have absorbed the factor $\sqrt{2} \Delta_q$ into the squared normal distribution in $v_q$, resulting in a single Gaussian with variance $\Delta_q^2/2$.

\subsection{Conditional statistical moments of the twirled error operator}
\label{appendix: Conditional statistical moments of the twirled error operato}
In the previous section, we renormalized the twirled channel $\text{Twirl}[\mathcal{E}]$ such that $p(\vect{v})$ is a normalized probability distribution function. In this section, we describe how to compute the statistical moments of this probability distribution function. Of particular interest are the first and second statistical moments, which characterize the average and the spread of the displacement error in the twirled channel $\mathrm{Twirl}[\tilde{\mathcal{E}}]$, respectively. Since $f(v_q|v_p)$ cannot be computed analytically, we cannot determine the exact error rate of the random displacement channel. Instead, we use these two moments as proxies for the error rate. If the mean is far away from the origin or the covariance matrix has large eigenvalues, we expect that the average gate fidelity will be low because the probability of sampling an uncorrectable displacement error~\cref{eq: correctable displacement} is high.

As we have explained previously, we can view $\frac{\Delta_p}{\sqrt{2}}f(v_p|v_q)$ as a conditional distribution of $v_p$ conditioned on the displacement error in position $v_q$. Therefore, we will first compute the $m$-th conditional statistical moment $\mathbb{E}(v_p^m|v_q)$ defined as
\begin{subequations}
    \label{eq: statistical moment mid calculation}
    \begin{align}
       & \mathbb{E}(v_p^{m}|v_q)  = \int dv_p \frac{v_p^m}{\pi} \frac{\Delta_p}{\sqrt{2}}   f(v_p|v_q) \\
        & =\int dx_+~\mathcal{N}(0, \frac{2}{\Delta_p^2},x_+) \int dx_- \Big( \frac{i}{2 \pi} \Big)^{m} \delta^{(m)}(x_-) \nonumber\\
        & \qquad\times e^{i 2 \sqrt{2} \pi v_q~ \delta P'_{n}(x_+,x_-) }  e^{- \pi \frac{\Delta_p^2}{2} x_-^2}\label{eq: statistical moment mid calculation exponent}.
    \end{align}
\end{subequations}
From the first line to the second line of Eq.~\eqref{eq: statistical moment mid calculation}, we have simply performed the integral over $v_p$ and trasformed the dummy variables everything in term of the variables 
\begin{equation}
    x_{+,\text{new}} = \frac{x_{+,\text{old}}}{\sqrt{\pi}},\quad ~x_{-,\text{new}} = \frac{x_{-,\text{old}}}{\sqrt{\pi}}.
\end{equation}
Here, $\delta^{(m)}(x)$ is the $m$-th order derivative of the Dirac-delta function, which is defined to have the property 
\begin{equation}
    \int dx~\delta^{(m)}(x) g(x) = (-1)^{m} g^{(m)}(0).
\end{equation}

For brevity, we denote $\Phi(x_+,x_-,v_q)$ to be the exponent that appears inside the integral over $x_-$ in Eq. \eqref{eq: statistical moment mid calculation exponent}: 
\begin{equation}
    \label{eq: definition of phi(x1,x2)}
    \Phi(x_+,x_-,v_q) = i 2 \sqrt{2} \pi v_q~ [\delta P'_{n}](x_+,x_-) - \frac{\pi \Delta_p^2}{2} x_-^2,
\end{equation}
such that $\mathbb{E}(v_p^m|v_q)$ is written as
\begin{subequations}
        \label{eq: compact form of E(vpm,vq)}
        \begin{align}
            \mathbb{E}(v_p^m|v_q)& = \int d x_{+} \mathcal{N}\Big( 0,\frac{2}{\Delta_p^2},x_+ \Big)\\
            & \times \int d x_- \Big( \frac{i}{2 \pi} \Big)^{m} \delta^{(m)}(x_-) e^{i \Phi(x_+,x_-,v_q)}.
    \end{align}
\end{subequations}
Due to the definition of the polynomial $\delta P'_{n}(x)$, we find that 
\begin{equation}
    \Phi(x_+,0,v_q) = 0, \quad \forall~  x_+,v_q \in \mathbb{R}.
\end{equation}
Using this property, the integral over $x_-$ in Eq.~\eqref{eq: compact form of E(vpm,vq)} can be evaluated using Faà di Bruno's formula \cite{Roman1980}
\begin{multline}
    \int dx_-~\delta^{(m)}(x_-) e^{ \Phi(x_+,x_-,v_q)}  \\
    %  = (-1)^{m} \frac{\partial^{m} e^{i \Phi(x_+,x_-,v_q)} }{\partial x_-^m} \Big|_{x_- = 0}
    % & = (-1)^{m} e^{\Phi(x_+,0,v_q)} B_{m}[\partial^{(1)}_{x_-} \Phi(x_+,0,v_q), \dots,\partial^{(m)}_{x_-} \Phi(x_+,0,v_q) ] \\
    = (-1)^m B_{m}[\partial^{(1)}_{x_-} \Phi(x_+,0,v_q), \dots,\partial^{(m)}_{x_-} \Phi(x_+,0,v_q) ],
\end{multline}
where $B_{m}(x_1,\dots,x_n)$ is the $m$-th complete exponential Bell polynomial
\begin{equation}
    B_{m}(x_1,\dots,x_n) = (m!) \sum_{1 j_1 + \dots m j_m = m} \prod_{k=1}^{m} \frac{ x_{k}^{j_k} }{ (k!)^{j_k} (j_k!)  }.
\end{equation}
Consequently, the analytical formula for the conditional statistical moments $\mathbb{E}(v_p^{m}|v_q)$ is given by 
\begin{subequations}
    \label{eq: formula conditional moment vp}
   \begin{align}
        & \mathbb{E}(v_p^{m}|v_q) = \int d x_+~\mathcal{N}\Big(0, \frac{2}{\Delta_p^2},x_+ \Big) \Big( \frac{i}{2 \pi} \Big)^{m} (-1)^{m} \nonumber\\
        & \qquad\times B_{m}[\partial^{(1)}_{x_-} \Phi(x_+,0,v_q), \dots,\partial^{(m)}_{x_-} \Phi(x_+,0,v_q) ].
   \end{align}
\end{subequations}

From the general formula Eq.~\eqref{eq: formula conditional moment vp} for $\mathbb{E}(v_p^m|v_q)$, it is then a straight-forward calculation to show that the first conditional moment $\mathbb{E}(v_p|v_q)$ for a general polynomial $P_n(x)$ with coefficients $\{ a_i \}$ is given by
\begin{align}
    \mathbb{E}(v_p|v_q) &= \sqrt{2} v_q  \sum_{k=2}^{n}
        a_k \beta_{k} \int dx_+~\mathcal{N}\Big(0, \frac{2}{\Delta_p^2},x_+ \Big)  x_+^{k-2} \nonumber\\
    & = \sqrt{2} v_q  \sum_{k=2,4,\dots}^{n} a_k \beta_{k}~ \mathbb{E}(x_+^{k-2}),
\label{eq: first conditional moment vp}
\end{align}
where the coefficients $\beta_{k}$ are given by
\begin{equation}
    \begin{aligned}
        \beta_{k} & = \frac{2 k (k-1)}{ \sqrt{2}^{k-1}} ,
    \end{aligned}
\end{equation}
and we define the expectation value $\mathbb{E}[g( x_+)]$ of a function $g(x_+)$ as
\begin{equation}
    \mathbb{E}\Big(g(x_+) \Big) = \int d x_+~ \mathcal{N}\Big(0, \frac{2}{\Delta_p^2},x_+ \Big)  g(x_+).
\end{equation}
Note that only even values of $k$ contribute to Eq.~\eqref{eq: first conditional moment vp} since $\mathbb{E}(x_+^{k})$ vanishes for odd $k$ due to the symmetry of the Gaussian distribution. We have seen this previously in Eq.~\eqref{eq: absolute square of chi gives a probability function} where the third order coefficient $a_3$ gives correction to $\mathbb{E}(v_p|v_q)$ to the second order $O(v_q^2)$. Similarly, we find the second conditional moment $\mathbb{E}(v_p^2|v_q)$ to be 
\begin{equation}
    \begin{aligned}
        \mathbb{E}(v_p^2 |v_q) 
        & = \frac{\Delta_{p}^2}{4 \pi} + 2 v_q^2~\mathbb{E} \Big[ \Big(\sum_{k=2}^{n}
            a_k \beta_{k} x_+^{k-2}  \Big)^2 \Big].
    \end{aligned}
\end{equation}
Given the formula for the conditional moments, we can obtain the formulas for the unconditional first and second moments, which are presented in the next section.

\subsection{Asymptotic scaling of the first and second unconditional moments}\label{app: optimal bias estimate}

In the previous section, we derived an expression for the conditional moment $\mathbb{E}(v_p^m|v_q)$. In this section, we use this result to determine the asymptotic scaling behavior of the mean vector 
\begin{equation}
    \begin{pmatrix}
        \mathbb{E}(v_q) \\ \mathbb{E}(v_p)
    \end{pmatrix}
\end{equation}
and covariance matrix
\begin{equation}
    \begin{pmatrix}
        \mathbb{E}(v_q^2) & \mathbb{E}(v_q v_p) \\ \mathbb{E}(v_q v_p) & \mathbb{E}(v_p^2)
    \end{pmatrix}.
\end{equation}
As we already discussed in the main text, the covariance matrix can be used as a proxy for the average gate fidelity. Using this, we can obtain an asymptotic estimate of the optimal asymmetry $\lambda_{\rm opt}$ by minimizing the elements of covariance matrix.

By definition, the unconditional moment $\mathbb{E}(v_p^m)$ is obtained from the conditional moment $\mathbb{E}(v_p^m|v_q)$ by averaging over the distribution of position displacement errors $v_q$, i.e. 
\begin{subequations}
    \begin{align}
        \mathbb{E}(v_p)& = \int dv_q~\mathcal{N}\Big[0,\frac{\Delta_q^2}{2},v_q \Big]~\mathbb{E}(v_p|v_q) =0,
    \label{eq: momentum mean polyphase}\\
        \mathbb{E}(v_p^2)  & =\int dv_q~\mathcal{N}\Big[0,\frac{\Delta_q^2}{2},v_q \Big]~\mathbb{E}(v_p^2|v_q) \\
         & = \frac{\Delta_{p}^2}{4 \pi} +  \frac{\Delta_q^2}{2\pi}~\mathbb{E} \Big[ \Big(\sum_{k=2}^{n}
            a_k \beta_{k} x_+^{k-2}  \Big)^2 \Big] \label{eq: momentum variance poly phase}
    \end{align}
\end{subequations}
To find the first and second moment $\mathbb{E}(v_q)$ and $\mathbb{E}(v_q^2)$, we note that the probability that a displacement error $v_q$ happens is independent of $v_p$. Therefore, from Eq.~\eqref{eq: normalized effective poly err channel}, we find that 
\begin{equation}
            \label{eq: moment of vq}
            \mathbb{E}(v_q)  = 0, \quad \mathbb{E}(v_q)^2 = \frac{\Delta_q^2}{4 \pi}.
\end{equation}
Finally, the cross-correlation term can be computed by combining Eq.~\eqref{eq: first conditional moment vp} and Eq.~\eqref{eq: moment of vq}
\begin{subequations}
    \begin{align}
        \mathbb{E}(v_q v_p) & = \int dv_q ~ \mathcal{N}\Big(0,\frac{\Delta_q^2}{2},v_q  \Big)~v_q ~\mathbb{E}(v_p|v_q) \\
    & = \sqrt{2}~ \mathbb{E}(v_q^2) \sum_{k=2,4,\dots}^{n}a_k \beta_k \mathbb{E}(x_{+}^{k-2}). \\
    & = \frac{\Delta_q^2}{2\sqrt{2} \pi} \sum_{k=2,4,\dots}^{n}a_k \beta_k \mathbb{E}(x_{+}^{k-2}).\label{eq: Evqvp with sum}
    \end{align} 
\end{subequations}

We can obtain an explicit formula for \cref{eq: momentum variance poly phase,eq: Evqvp with sum} by using the known formula for the statistical moments of a Gaussian distribution 
\begin{subequations}
\label{eq: statistical moment of x1}
    \begin{align}
          &  \mathbb{E}(x_+^{k})  = \int dx_+~\mathcal{N}\Big(0, \frac{2}{\Delta_p^2},x_+ \Big)  x_+^{k} \\
          & = 2^{-1+\frac{k}{2}}[1+(-1)^k] \pi^{-\frac{(1+k)}{2}} \Gamma\Big[\frac{1+k}{2} \Big]\Delta_{p}^{-k}.
    \end{align}
\end{subequations}

In the limit of infinitely squeezed GKP states, the leading contribution to $\mathbb{E}(v_p^2)$ comes from the leading coefficient $a_n$ of the polynomial $P_{n}(x)$. As a result, the second moment $\mathbb{E}(v_p)^2$ acquires a simple form in this limit 
\begin{subequations}
    \label{eq: Evp2 appendix}
    \begin{align}
            \mathbb{E}(v_p^2) & \approx \frac{\Delta_{p}^2}{4 \pi} + a_n^2 \beta_{n}^2  \frac{  2^{n-3} \pi^{1/2-n} \Gamma\Big[n-\frac{3}{2}\Big] }{\Delta_p^{2n-6}}\frac{\Delta_q^2}{\Delta_p^2}. \\
            & = \frac{\Delta_p^2}{4\pi} + O \Big(\frac{\Delta_q^2}{\Delta_p^{2n-4}}  \Big)
    \end{align}
\end{subequations}

In a similar fashion, we retain the leading contribution in $\mathbb{E}(v_q v_p)$. However, depending on whether the degree $n$ of $P_n$ is even or odd, the leading contribution can either come from the leading coefficient or the second leading coefficient. If we assume that the degree of $P_n(x)$ is even, the cross-correlation is given by 
\begin{align}
        \label{eq: Evqvp appendix}
       \mathbb{E}(v_q v_p) & = \frac{\Delta_q^2}{2\sqrt{2} \pi} a_n \beta_n 2^{-2+\frac{n}{2}} \pi^{-\frac{(n-1)}{2}} \Gamma\Big[ \frac{n-1}{2} \Big] \Delta_p^{-(n-2)} \\
       & = O\Big( \frac{\Delta_q^2 }{\Delta_p^{n-2}}  \Big).
\end{align}
Again, assuming that the degree $P_{n}(x)$ is odd, we find that 
\begin{equation}
    \mathbb{E}(v_q v_p) = O\Big( \frac{\Delta_q^2}{\Delta_p^{n-3}} a_{n-1}  \Big)=O\Big( \frac{\Delta_q^2}{\Delta_p^{n-2}}\Big).
\end{equation}

We assume that the GKP has asymmetry $\lambda$ such that squeezing parameters are given by 
\begin{equation}
    \Delta_q = \frac{\Delta}{\sqrt{\lambda}},~ \Delta_p = \sqrt{\lambda} \Delta. 
\end{equation}
Since $\mathbb{E}(v_p^2)$ is the dominant error in the limit of infinitely squeezed state, our goal is to minimize it. The global minimum of $\mathbb{E}(v_p^2)$ is obtained at the optimal asymmetry value 
\begin{subequations}
    \label{eq: estimated lambdac}
    \begin{align}
            \lambda_{\rm opt} & = \sqrt[n]{4 \pi (n-1) \frac{a_n^2 \beta_{n}^2 2^{n-3} \pi^{1/2-n} \Gamma\Big[n-\frac{3}{2}\Big] }{ \Delta^{2n-4} }  } \\
            & = O\big(\Delta^{-(2-4/n)}\big). 
    \end{align}
\end{subequations}
Consequently, we can easily show that the second moment $\mathbb{E}(v_p^2)$ and $\mathbb{E}(v_q^2)$ evaluated at $\lambda_{\rm opt}$ scales as
\begin{subequations}
    \begin{equation}
            \mathbb{E}(v_q^2) \propto \Delta^{4-4/n} , \quad \mathbb{E}(v_p^2) \propto \Delta^{4/n} , 
    \end{equation}
    \begin{equation}
        \mathbb{E}(v_q v_p) \propto  \begin{cases}
        \Delta^2 & \text{for even }n \\
        \Delta^{2+\frac{2}{n}} & \text{for odd }n
    \end{cases} . 
    \end{equation}
\end{subequations}
which vanish in the limit of infinitely squeezed GKP states. This result implies that, even for high-order polynomial phase gates, the propagated error induced by the non-Gaussian interaction can be mitigated by on-demand noise biasing.

\subsection{Scaling of covariance matrices at finite squeezing}
\label{appendix: scaling of covariance matrices at finite squeezing}

We have analyzed the asymptotic scaling of polynomial phase gates and shown that, in the limit of infinite squeezing, gates with lower-degree polynomials asymptotically achieve higher fidelities than gates with higher-degree polynomials. However, this asymptotic behavior does not imply that higher-degree polynomial phase gates perform worse for all finite values of $\Delta$. To investigate this, we now analyze how the leading coefficient $a_n$ of a polynomial phase gate affects its error scaling when applied to GKP states with finite $\Delta$. In particular, we are interested in how the error in lexicographically minimal polynomial phase gates $\Lambda_{n}[q/\sqrt{\pi}]$ scales with increasing $n$.

We first note that the prefactor $\beta_n$ scales as $\beta_n \propto O(2^{-n/2})$, which exactly cancels the exponential dependence on $2^{n/2}$ appearing in Eqs.~\eqref{eq: Evp2 appendix} and~\eqref{eq: Evqvp appendix}. In other words,
\begin{equation}
\beta_n 2^{n/2} \propto O(1).
\end{equation}
With this simplification, the second moments $\mathbb{E}(v_p^2)$ and $\mathbb{E}(v_q v_p)$ for finite $\Delta$ scale with $n$ as
\begin{subequations}
    \begin{gather}
        \mathbb{E}(v_p^2)  \propto \frac{\Delta_p^2}{4\pi} + O \Big( a_n^2\frac{ \Gamma\Big[n-\frac{3}{2}\Big] }{\pi^{n}} \frac{\Delta_q^2}{\Delta_p^{2n-4}}  \Big) \\
        \mathbb{E}(v_q v_p) \propto O \Big( a_n\frac{ \Gamma\Big[\frac{n-1}{2}\Big] }{\pi^{n/2}} \frac{\Delta_q^2}{\Delta_p^{n-2}}  \Big)
    \end{gather}
\end{subequations}
The appearance of the Gamma functions in the numerators implies a factorial growth with $n$, suggesting that, for fixed $\Delta$, non-optimized higher-order polynomial phase gates could lead to larger propagated errors. However, for minimal polynomial phase gates, the leading coefficients satisfy $a_n \propto O(1/n!)$ due to Theorem~\ref{theorem: minimal degree polynomial}. Substituting this scaling, we obtain
\begin{equation}
    a_n^2 \Gamma[n-\frac{3}{2}] \propto O(\frac{1}{n!}), \quad a_n \Gamma[\frac{n-1}{2}]\propto O(\frac{1}{n!}).
\end{equation}
As a result, for lexicographically minimal polynomial phase gates, the second moments  $\mathbb{E}(v_p^2)$ and $\mathbb{E}(v_q v_p)$ would scale with $n$ as
\begin{subequations}
    \begin{gather}
            \mathbb{E}(v_p^2)  \propto \frac{\Delta_p^2}{4\pi} + O \Big( \frac{ 1 }{\pi^{n}~n!} \frac{\Delta_q^2}{\Delta_p^{2n-4}}  \Big) \\
            \mathbb{E}(v_q v_p) \propto O \Big(\frac{ 1}{\pi^{n/2}~n!} \frac{\Delta_q^2}{\Delta_p^{n-2}}  \Big).
    \end{gather}
\end{subequations}
The factorial suppression in $n$ indicates that, for finite but fixed squeezing $\Delta$, higher-order minimal polynomial phase gates can in principle outperform lower-order ones such as the cubic phase gate. 

This analytical scaling is consistent with our numerical simulations shown in Fig.~\ref{fig:Lambda_gates}, where the $T^{1/8}$ gate have higher fidelity than the $T$ gates in the simulated parameter range. Nevertheless, we emphasize that this analysis neglects experimental challenges associated with realizing high-order polynomial phase gates, such as their non-linearity requirements and potential control errors. Thus, our results should be viewed as a proof of principle that establishes the theoretical potential for improved performance, rather than a practical prescription for near-term implementations.

\section{Proof of fault-tolerance theorem}
\label{appendix: proof of fault-tolerance theorem}

In this appendix, we prove that the minimized $T_{3}$ gate with on-demand noise biasing is fault-tolerant. Our strategy is to show that, for any given squeezing parameter $\Delta$, there exists an asymmetry value $\lambda(\Delta)$ such that the corresponding error operator ${E}(\Delta, \lambda(\Delta))$, defined the same as in the main text by
\begin{equation}\tag{\ref{eq: error operator}}
    {E}\big(\Delta,\lambda(\Delta)\big)=\Lambda \,\text{Env}_{\Delta/\lambda(\Delta),\Delta\lambda(\Delta)}\,\Lambda^{\dag},
\end{equation}
has logical error rate going to zero in the limit $\Delta \to 0$. 

Specifically, we consider the following circuit
\begin{equation}
    \begin{quantikz}
         \lstick{$\ket{\bar{\psi}_\text{in}}$} & \gate{{E}(\Delta,\lambda(\Delta))  } & \gate{\substack{\text{Ideal} \\ \text{QEC}}} & \rstick{$\bar{\rho}_\text{out}$}.
    \end{quantikz}
\end{equation}
Because the input and output of the circuit are contained within the ideal GKP codespace, we can define a logical channel $\mathcal{E}$ that acts from qubits to qubits defined such that $\rho_{\text{out}}=\mathcal{E}(\ket{\psi_{\text{in}}}\!\bra{\psi_{\text{in}}})$, which implicitly depends on the squeezing parameter $\Delta$ and the asymmetry $\lambda$. Here, $\mathcal{E}$ differs slightly from the logical channel defined in Eq.~\eqref{eq: logical channel equation}, because we assume that the QEC round following the $T$ gate is ideal. This assumption allows us to isolate the contribution of errors arising solely from the $T_{3}$ gate, independent of those originating from noisy QEC round. The proof can be straightforwardly extended to include noisy QEC rounds if desired because the noise arising from the noisy QEC round also goes to zero as $\Delta\rightarrow0$.

We define the average logical fidelity of the channel $\mathcal{E}$ gate, which also defines the ``logical'' average gate fidelity of the $T_3$ gate, as 
\begin{equation}\label{eq: bar F definition}
    \bar F(\Delta,\lambda) = \int d \psi~\bra{\psi} \mathcal{E}(\ket{ \psi}\!\bra{ \psi}) \ket{\psi}.
\end{equation}
To understand why $\bar{F}(\Delta,\lambda)$ characterizes the gate fidelity of the $T_3$ gate, recall that the operator ${E}$ represents the effect of propagating the envelope operator through the $T_3$ gate. The ideal logical action of ${E}$ is therefore the identity. Consequently, the average gate fidelity of ${E}$ with respect to the identity directly corresponds to the logical gate fidelity of the $T_3$ operation.

We now prove our fault-tolerance theorem.
\begin{theorem}\label{theorem: FT theorem T3 gate}
   Let $\varepsilon_{\rm target}$ be a desired logical average gate infidelity for the $T_3$ gate. Then there exists a threshold parameter $\Delta_{\rm th}$ such that
   \begin{equation}
       \inf_{\lambda>0}[1-\bar F(\Delta,\lambda)] \leq \varepsilon_{\rm target}
   \end{equation}
   for all $\Delta \leq \Delta_{\rm th}$, where $\bar{F}(\Delta,\lambda)$ is defined in \cref{eq: bar F definition}.
\end{theorem}

\begin{proof}

The proof proceeds in three main steps. First, we use the characteristic function $\chi_{{E}}(\bo v)$ of the propagated error operator ${E}$ (defined in Eq.~\eqref{eq: error operator}) to write down an exact expression for the average fidelity $\bar{F}(\Delta,\lambda)$. Second, we derive a lower bound $\bar{F}(\Delta,\lambda)_{\text{l.b.}}$ for this average fidelity $\bar F(\Delta,\lambda)$ that simplifies the expression found in the first step. Third, we show that for each $\Delta$ there exists a biasing $\lambda(\Delta)$ such that $\bar{F}(\Delta,\lambda(\Delta))_{\text{l.b.}}\rightarrow 1$ as $\Delta\to 0$. This proves that $\mathrm{inf}_{\lambda>0}\bar F(\Delta,\lambda) \to 1$ as $\Delta\rightarrow 0$, which is equivalent to the statement in Theorem~\ref{theorem: FT theorem T3 gate} via the $\epsilon$-$\delta$ definition of a limit.

%\textbf{Step 1: Characterizing the logical channel $\mathcal{E}$}
\textbf{Step 1: An exact expression for $\bar{F}(\Delta,\lambda)$}

As a starting point, the propagated error operator ${E}$ of the $T_3$ is given exactly by \cref{eq: main-text propagated error operator exact}, which we write here as
\begin{multline}\label{eq: reminder propagated error operator}
        {E}(\Delta,\lambda) = \int d \bo v \,\mathcal{N}\Big(2\tanh(\frac{\Delta^2}{2})\text{Diag}(\lambda^{-1},\lambda),\bo v\Big) \\
        \times W(\bo v) e^{i 2 \pi \big[ P_{\text{min},3}({q}/\sqrt{\pi}+\sqrt{2} v_q)-P_{\text{min},3}({q}/\sqrt{\pi} )\big] }.
\end{multline}
where $P_{\text{min},3}(x)$ is the minimal polynomial that implements the $T$ gate defined in Eq.~\eqref{eq: T3 gate}.

Now, our next step will be to evaluate the so-called $\chi$-matrix of the \textit{logical} channel $\mathcal{E}$, which is defined such that
\begin{equation}
    \mathcal{E}(\rho)=\sum_{\mu,\mu'=0}^{3}\chi_{\mathcal{E}}^{\sigma_\mu,\sigma_{\mu'}}\sigma_{\mu}\rho\sigma_{\mu'},
\end{equation}
where as a reminder the logical channel $\mathcal{E}$ is a qubit-to-qubit channel, and where $\sigma_{0},\sigma_{1},\sigma_{2},\sigma_{3}=\mathrm{Id},{X},{Y},{Z}$. This representation is convenient because the average gate fidelity of $\mathcal{E}$ is straightforward to obtain from the $\chi$-matrix via \cite{NIELSEN2002}
\begin{equation}\label{eq: average gate fidelity chi matrix}
    \bar{F}(\Delta,\lambda)=\frac{2}{3}\frac{\chi_{\mathcal{E}}^{\mathrm{Id},\mathrm{Id}}}{\sum_{\mu}\chi_{\mathcal{E}}^{\sigma_\mu,\sigma_\mu}}+\frac{1}{3}.
\end{equation}

The $\chi$-matrix elements of $\mathcal{E}$ can be obtained from the characteristic function of $E(\Delta,\lambda)$. As a stepping stone to this, we define the \textit{logical} characteristic function $\xi_{E}^{\sigma_{\mu}}(\vect{v})$ of $E(\Delta,\lambda)$ as
\begin{subequations}
\begin{align}
    \xi_{ E}^{\sigma_\mu}( \bo v) & = \text{Tr}[ {\Pi}_{\text{GKP}, \mu} {W}^{\dg}(\bo v ) {E}],\\
    {\Pi}_{\text{GKP}, \mu} & = {\Pi}_{\text{GKP,Id}}\bar{\sigma}_{\mu}{\Pi}_{\text{GKP,Id}}=\sum_{\mathclap{j,k \in \{0,1\}}}\; [\sigma_{\mu}]_{j,k} \ket{\bar j} \bra{\bar k}, 
\end{align}
\end{subequations}
which is based on the so-called ``outcome-dependent'' GKP projectors defined in Ref.~\cite{Baragiola2019}. With these logical characteristic functions we can directly write
\begin{equation}\label{eq: logical characteristic function to chi matrix}
    \chi_{E}^{\sigma_{\mu},\sigma_{\mu'}}=\int_{\square} d \bo v~ \xi_{ E}^{\sigma_\mu}(\bo v) \xi_{ E}^{\sigma_{\mu'}}(\bo v)^{*},
\end{equation}
where the integral is taken over the correctable error patch $\square$ defined as  
\begin{equation}
    \label{eq: correctable patch}
    \square=\big(-1/\sqrt{8},1/\sqrt{8}\big]^{\times 2}\subset\mathbb{R}^{2},
\end{equation}
which is simply a rescaling of~\cref{eq: correctable displacement}.

Ref.~\cite{shaw2022} presents an alternative way of calculating the logical characteristic function. First, we write the error operator $E(\Delta,\lambda)$ as
\begin{multline}
        \label{eq: decomposing propagated error into correctable, pauli, and stabilizers}
    E = \sum_{\bo n \in \mathbb{Z}^{2}} \sum_{\mu =0}^{3} \int_{\square} d \bo v\,f_{ E}(\bo v,\sigma_\mu,\bo n) {W}(\bo v)\bar{\sigma}_{\mu} {W}(\sqrt{2} \bo n).
\end{multline}
% (i^{r_q r_p} X^{r_q} Z^{r_p}) {W}(\sqrt{2} \bo n),
Since ${W}(\bo v)\bar{\sigma}_{\mu} {W}(\sqrt{2} \bo n)$ is a displacement operator (up to a geometric phase from~\cref{eq: displacement composition rule}), the function $f$ can be written
\begin{equation}\label{eq: f from characteristic function}
    f_{E}(\bo v,\sigma_{\mu},\bo n) = e^{i\theta(\vect{v},\sigma_{\mu},\vect{n})}\chi_{{E}}(\bo v+\sqrt{2}\vect{n} + \vect{\ell}_{\mu}),
\end{equation}
for some function $\theta$, where $\vect{\ell}_{\mu}=(0,0), (1/\sqrt{2},0),(1/\sqrt{2},1/\sqrt{2}), (0,1/\sqrt{2})$ respectively for $\mu=0,1,2,3$. After calculating $f$ from the characteristic function $\chi_{E}$, the logical characteristic function is simply
\begin{equation}
    \label{eq: logical characteristic in term of physical characteristic}
     \xi_{ {E}}^{\sigma_{\mu}}(\bo v)= \sum_{\bo n \in \mathbb{Z}^2} f_{ E}(\bo v,\sigma_{\mu},\bo n)
\end{equation}

From \cref{eq: average gate fidelity chi matrix,eq: logical characteristic function to chi matrix}, the average gate-fidelity is simply given by 
\begin{equation}\label{eq: average gate fidelity logical channel}
    \bar{F}(\Delta,\lambda) = \frac{2}{3} \frac{ \int_{\square}d \bo v |\xi_{E}^{\rm Id}(\bo v)|^2 }{\sum_{\sigma_\mu}  \int_{\square}d \bo v |\xi_{E}^{\sigma_\mu}(\bo v)|^2} + \frac{1}{3}, 
\end{equation}
which comes from Eq.~(B17) of Ref.~\cite{shaw2024logical}.

Our first task is to evaluate the denominator in the first term of \cref{eq: average gate fidelity logical channel} which we call the ``normalization constant'' of the logical channel $\mathcal{E}$. Computing each individual contribution of the normalization constant is a non-trivial task. However, we can use the following identity to simplify the computation 
\begin{equation}
    \label{eq: normalization identity}
    \sum_{\sigma_{\mu} } \int_{\square} d \bo v | \xi_{E}^{\sigma_\mu}(\bo v)|^2 = \xi_{{E}^{\dagger} {E}}^{\rm Id}( \bo 0). 
\end{equation}

To prove Eq.~\eqref{eq: normalization identity}, we start from the right-hand side (RHS) and make use of the decomposition given in Eq.~\eqref{eq: decomposing propagated error into correctable, pauli, and stabilizers}. Recall that $\mathcal{\xi}_{{E}^{\dg}{E}}^{\rm Id}(\bo 0)$ consists of contributions from displacement errors $W(\bo v)$ that act trivially on the GKP code space. Therefore, when we express the error operators ${E}^{\dg}$ and ${E}$ in terms of their decompositions from Eq.~\eqref{eq: decomposing propagated error into correctable, pauli, and stabilizers} and combine them, only the terms in $\mathcal{\xi}_{{E}^{\dg}{E}}^{\rm Id}(\bo 0)$ with opposite correctable displacements and identical logical Pauli operators contribute to $\mathcal{\xi}_{{E}^{\dg}{E}}^{\rm Id}(\bo 0)$. Consequently, we find that
\begin{equation}
    \xi_{{E}^{\dg} {E}}^{\rm Id}(\vect{0}) = \int_{\square} d \bo v\sum_{\bo n',\bo n \in \mathbb{Z}^2 } \sum_{\mu=0}^{3}  f(\bo v, \sigma_{\mu},\bo  n')^{*} f(\bo v,\sigma_{\mu}, \bo n)
\end{equation}
which is just the LHS of Eq.~\eqref{eq: normalization identity}. 

We will not compute $\xi_{{E}^{\dg} {E}}^{\rm Id}(\vect{0})$ directly, but use the following trick 
\begin{subequations}
    \begin{align}
        & \xi_{{E}^{\dg} E}^{\rm Id}( \bo 0)  = \text{Tr}[{\Pi}_{\rm GKP, \rm Id} ~  {E}^{\dg}{E}]\\
        & = \text{Tr}[  {T}_3^{\dg} {E}^{\dg}{E}~{\Pi}_{\rm GKP, \rm Id} {T}_3]\\
        &=  \text{Tr}[  {T}_3^{\dg} {E}^{\dg}{E} {T}_3~{\Pi}_{\rm GKP, \rm Id} ] \\
        & = \text{Tr}\Big[ \mathrm{Env}_{\Delta/\lambda,\Delta\lambda}^{\dg} \mathrm{Env}_{\Delta/\lambda,\Delta\lambda} {\Pi}_{\rm GKP, \rm Id}  \Big] \\
        & = \text{Tr}\Big[  \mathrm{Env}_{\sqrt{2}\Delta/\lambda,\sqrt{2}\Delta\lambda}  {\Pi}_{\rm GKP, \rm Id}     \Big],
    \end{align} 
\end{subequations}
where we have used the fact that ${T}_3$ commutes with the projector onto the GKP code $\Pi_{\rm GKP, Id}$ since it corresponds to a GKP logical gate. Therefore, the normalization constant is entirely dependent on the biased envelope operator. Using the characteristic function of the biased envelope operator given in Eq.~\eqref{eq: envelope characteristic function approx}, we find that the normalization constant is simply 
\begin{equation}
    \xi_{{E}^{\dg} E}^{\rm Id}( \bo 0)  =\frac{\coth(\Delta^2)}{2} \vartheta _3\Big(0,e^{-\frac{\pi  \coth \left( \Delta^2\right)}{\lambda }}\Big) \vartheta _3\Big(0,e^{-\pi  \lambda  \coth \left( \Delta^2 \right)}\Big)
\end{equation}
where $\vartheta_3$ is the 3rd Jacobi-Theta function, defined as 
\begin{equation}
    \vartheta_3(z,q)= 1+2 \sum_{n=1}^{\infty} q^{n^2} \cos(2 n z). 
\end{equation}

Our next step is to evaluate the numerator of the first term in \cref{eq: average gate fidelity logical channel}. We start by evaluating $\xi_{{E}}^{\rm Id}(\bo v)$, which involves calculating the characteristic function $\chi_{{E}}(\bo v)$. The first step is decomposing the shear error in Eq.~\eqref{eq: reminder propagated error operator} into displacement errors 
\begin{multline}\label{eq: decomposing shear error into displacement}
    e^{i 2 \pi \big[ P_{\text{min},3}(\frac{{q}+\sqrt{2\pi} v_q}{\sqrt{\pi}} )-P_{\text{min},3}(\frac{{q}}{\sqrt{\pi}} )\big]} \\
    = \frac{1}{\sqrt{2\pi}}\iint d u_p dx \bigg(
    e^{i 2 \pi \big[ P_{\text{min},3}(\frac{x+\sqrt{2\pi} v_q}{\sqrt{\pi}} )-P_{\text{min},3}(\frac{x}{\sqrt{\pi}} )\big]}\\
    \times e^{-i\sqrt{2\pi} u_p x}{W}(0,u_p)\bigg). 
\end{multline}
By combining Eq.~\eqref{eq: decomposing shear error into displacement} with Eq.~\eqref{eq: reminder propagated error operator}, through some tedious algebra, we find that $\chi_{ E}$ is 
\begin{multline}\label{eq: characteristic of P3 propagated error}
    \chi_{{E}}(\bo v)= \frac{ \mathcal{N}(\frac{2 \tanh(\Delta^2/2)}{\lambda},v_q)} {\sqrt{2\pi}} \int dx\bigg(e^{-i v_p \sqrt{2\pi} x}   \\
    \times e^{- \lambda \tanh(\frac{\Delta^2}{2})x^2  }  e^{i 2 \pi \big[ P_3(\frac{x}{\sqrt{\pi}} )-P_3(\frac{x-\sqrt{2\pi} v_q}{\sqrt{\pi}} )\big] } \bigg)
\end{multline}
The integral over $x$ can be performed exactly. Then, by taking the modulus squared and dividing by the normalization constant $\xi_{{E}^{\dg}{E}}(\vect{0})$, the result can be written in terms of 1D normal distribution functions
\begin{widetext} 
    \begin{align}
    \label{eq: absolute square of chi gives a probability function}
    \frac{|\chi_{{E}}(\bo v)|^2}{ \xi_{{E}^{\dg}{E}}^{\rm Id}(\vect{0})} = \frac{ \mathcal{N}(\tanh(\frac{\Delta^2}{2})/\lambda,v_q)\mathcal{N}(\frac{ v_q^2}{2 \lambda \tanh(\frac{\Delta^2}{2})}+\lambda \tanh(\frac{\Delta^2}{2}) ,v_p - \frac{v_q}{2} + \frac{v_q^2}{\sqrt{2}})  }{C(\Delta,\lambda)},
\end{align}
where
\begin{align}\label{eq: C Delta lambda}
    C(\Delta,\lambda)  = 2 \coth(\Delta^2)\tanh(\frac{\Delta^2}{2}) \vartheta _3\left(0,e^{-\frac{\pi  \coth \left(\Delta^2\right)}{\lambda }}\right) \vartheta _3\left(0,e^{-\pi  \lambda  \coth \left(\Delta^2 \right)}\right).
\end{align}
\end{widetext}
This is not quite equal to the fraction in \cref{eq: average gate fidelity logical channel}, but in the next step we will use \cref{eq: absolute square of chi gives a probability function} to lower-bound $\bar{F}(\Delta,\lambda)$.

\textbf{Step 2: Lower bound for the average fidelity $\bar F(\Delta,\lambda)$}

In this stage of the proof, our goal is to obtain a lower-bound for the numerator in \cref{eq: average gate fidelity logical channel}. From \cref{eq: f from characteristic function,eq: logical characteristic in term of physical characteristic}, we can write
\begin{subequations}
    \begin{align}
            |\xi_{{E}}^{\rm Id}(\bo v)|^2 & = \Big| |\chi_{{E}}(\bo v)| + \sum_{\substack{\bo n \in \mathbb{Z}^2  \\ \bo n \neq \bo 0 }} |\chi_{{E}}(\bo v + \sqrt{2}\bo n)| e^{i \theta(\bo v,\mathrm{Id},\bo n)} \Big|^2 \\
            & \equiv \Big| |\chi_{{E}}(\bo v)| + \chi_{{E}}^{\rm Id, rest}(\bo v) \Big|^2
    \end{align}
\end{subequations}
We have also defined $\chi_{{E}}^{\rm Id,rest}$ as the sum of all the contributions to $\xi^{\rm Id}_{{E}}$ which come from outside the correction patch located at the origin. Substituting $\arg(\chi_{{E}}^{\rm Id,rest}(\vect{v}))=\pi$ for all $\vect{v}$ gives a lower bound 
\begin{subequations}
    \label{eq: Cauchy-inequality 1}
    \begin{align}
            & \int_{\square} d \bo v | \xi_{{E}}^{\rm Id}(\bo v) |^2= \int_{\square} d \bo v \Big| |\chi_{E}(\bo v)|+\chi_{{E}}^{\rm Id,rest}(\bo v)  \Big|^2 \\
            & \geq   \int_{\square}d \bo v \Big(|\chi_{E}(\bo v)|^2-2 |\chi_{E}(\bo v)| ~ |\chi_{{E}}^{\rm Id,rest}(\bo v)|+|\chi_{{E}}^{\rm Id,rest}(\bo v)|^2 \Big) 
    \end{align}
\end{subequations}
At this step, we apply the Cauchy–Schwarz inequality successively 
\begin{subequations}
    \label{eq: Cauchy-inequality 2}
    \begin{align}
       & \int_{\square} d\bo v  |\chi_{E}(\bo v)| |\chi_{{E}}^{\rm Id,rest}(\bo v)|  \\
       & \leq \int dv_q \sqrt{\int dv_p |\chi_{{E}}(\bo v)|^2 } \sqrt{\int dv_p |\chi_{{E}}^{\rm Id,rest}(\bo v)|^2 } \\
       & \leq  \sqrt{ \iint dv_qdv_p |\chi_{{E}}(\bo v)|^2   } \sqrt{ \iint dv_qdv_p |\chi_{{E}}^{\rm Id,rest}(\bo v)|^2    } \\
       &= \sqrt{ \int_{\square} d \bo v |\chi_{{E}}(\bo v)|^2   } \sqrt{ \int_{\square} d\bo v |\chi_{{E}}^{\rm Id,rest}(\bo v)|^2    }.
\end{align}
\end{subequations}
The main idea behind applying the Cauchy–Schwarz inequality is that it allows us to complete the square and thereby obtain a new lower bound 
\begin{subequations}
    \begin{align}
        & \int_{\square} d \bo v | \xi_{{E}}^{\rm Id}(\bo v) |^2\\
        & \qquad \geq \left(\sqrt{\int_{\square}  d\bo v |\chi_{E}(\bo v)|^2    } - \sqrt{ \int_{\square} d\bo v |\chi_{{E}}^{\rm Id,rest}(\bo v)|^2 } \right)^2.
\end{align}
\end{subequations}

In the following, we find an upper-bound for the second term in the parentheses, which will gives us the final lower-bound to the fidelity. We note that $|\chi_{{E}}^{\rm Id,\rm rest}(\bo v)|$ is upper-bounded by 
\begin{equation}
    |\chi_{{E}}^{\rm Id,\rm rest}(\bo v)| \leq \sum_{\substack{\bo n \in \mathbb{Z}^2  \\ \bo n \neq \bo 0 }}|\chi_{E}(\bo v + \sqrt{2} \bo n) |.
\end{equation}
Given this upper bound, we find that 
\begin{subequations}
\begin{align}
        & \int_{\square}d \bo v | \chi_{E}^{\rm Id,rest}(\bo v)|^2 \leq \int_{\square }d \bo v \Big(  \sum_{\substack{\bo n \in \mathbb{Z}^2  \\ \bo n \neq \bo 0 }} |\chi_{E}(\bo v+\sqrt{2}\bo n)| \Big)^2 \\
        & \qquad\leq \Bigg( \sum_{\substack{\bo n \in \mathbb{Z}^2  \\ \bo n \neq \bo 0 }} \sqrt{ \int_{\square} d\bo v |\chi_{{E}}(\bo v + \sqrt{2} \bo n)|^2   }\Bigg)^2 
\end{align}
\end{subequations}
where in the final inequality, we follow the same steps as in Eq.~\eqref{eq: Cauchy-inequality 1} and Eq.~\eqref{eq: Cauchy-inequality 2}. In summary, our final lower bound is 
\begin{align}
    \label{eq: final lower bound to gate fidelity}
        & \frac{\int_{\square} d \bo v | \xi_{{E}}^{\rm Id}(\bo v) |^2}{\xi_{{E}^{\dag}{E}}^{\rm Id}(\bo 0)} \geq  \Big( \sqrt{\int_{\square}  \frac{d\bo v |\chi_{\hat E}(\bo v)|^2  }{\xi_{{E}^{\dag}{E}}^{\rm Id}(\bo 0)}  }  \\
        & \qquad \quad - \sum_{\bo n \neq \bo 0} \sqrt{ \int_{\square} d\bo v \frac{|\chi_{{E}}(\bo v + \sqrt{2} \bo n)|^2}{\xi^{\rm Id}_{{E}^{\dg}{E}}(\bo 0)}   }  \Big)^2
\end{align}

Let us denote
\begin{equation}
     \frac{p_{E}(\bo n)}{ C(\Delta,\lambda) } =\int_{\square} d \bo v\frac{|\chi_{{E}}(\bo v +  \bo n)|^2}{\xi^{\rm Id}_{{E}^{\dg}{E}}(\bo 0)}  ,
\end{equation}
so that we rewrite Eq.~\eqref{eq: final lower bound to gate fidelity} as
\begin{equation}
    \label{eq: lower bound written in term of probability}
    \frac{\int_{\square} d \bo v | \xi_{{E}}^{\rm Id}(\bo v) |^2}{\xi_{{E}^{\dag}{E}}^{\rm Id}(\bo 0)} \geq \Big( \frac{p_{E}(\bo 0)}{C(\Delta,\lambda)}  - \sum_{\substack{\bo n \in \mathbb{Z}^2  \\ \bo n \neq \bo 0 }} \frac{p_{E}(\sqrt{2} \bo n)}{C(\Delta,\lambda)} \Big)^2.
\end{equation}
Based on Eq.~\eqref{eq: absolute square of chi gives a probability function}, we can interpret $p_{{E}}(\vect{n})$ as the probability that a displacement error $\bo v$ lies inside a square patch $\square$ centered at $\vect{n}$. Hence, using  Eq.~\eqref{eq: absolute square of chi gives a probability function}, we find that
\begin{equation}
    \label{eq: sum of probability}
    \sum_{\bo n} p_{E}\Big(\frac{\bo n}{\sqrt{2}} \Big) = 1.
\end{equation}

In the next step, we will make use of the following property of Gaussian integrals to obtain a lower bound for $p_{{E}}(\bo 0)$
\begin{equation}
    \label{eq: fact about Gaussian integral}
    \int_{-D}^{D} dx~\mathcal{N}(\Delta^2,x-\mu_1) \geq \int_{-D}^{D} dx~\mathcal{N}(\Delta^2,x-\mu_2), 
\end{equation}
for all $D$ and $\Delta^2$, provided that $|\mu_1| \leq |\mu_2|$. Intuitively, Eq.~\eqref{eq: fact about Gaussian integral} expresses the fact that the integral over the interval $(-D,D)$ represents the probability that the random variable $x$ lies within that range. This probability is maximized when the mean of the Gaussian $\mu$ is at the center of the interval, i.e.~at the origin $\mu=0$. As the mean $\mu$ moves away from the origin, the Gaussian distribution shifts, and because its tails decay exponentially, the probability of finding $x$ within $(-D,D)$ decreases accordingly.  For $\bo v \in \square$, we find that
\begin{equation}
    \max_{\bo v \in \square} |\frac{v_q}{2}-\frac{v_q^2}{\sqrt{2}}| = \frac{3}{4\sqrt{8}}
\end{equation}
obtained at $v_{q}=-1/\sqrt{8}$. Consequently, by using Eq.~\eqref{eq: fact about Gaussian integral}, we can lower bound $p_{{E}}(\bo 0)$ by 
\begin{subequations}
    \begin{align}
    & p_{E}(\bo 0)  \geq \int_{\square} d\bo v ~ \mathcal{N}\Big(\tanh(\frac{\Delta^2}{2})/\lambda,v_q\Big) \nonumber\\
    &  \quad \times   \mathcal{N}\Big(\frac{ v_q^2}{2 \lambda \tanh(\frac{\Delta^2}{2})}+\lambda \tanh(\frac{\Delta^2}{2}) ,v_p-\frac{3}{4\sqrt{8}} \Big)\nonumber \\
    & \geq \int_{-1/\sqrt{8}}^{1/\sqrt{8}} dv_{q} ~ \mathcal{N}\Big(\tanh(\frac{\Delta^2}{2})/\lambda,v_q\Big) \\
    &  \times \int_{\frac{1}{2\sqrt{8}}}^{\frac{1}{\sqrt{8}}} dv_p~    \mathcal{N}\Big(\frac{ v_q^2}{2 \lambda \tanh(\frac{\Delta^2}{2})}+\lambda \tanh(\frac{\Delta^2}{2}) ,v_p-\frac{3}{4\sqrt{8}} \Big) \\
    & =  \int_{-1/\sqrt{8}}^{1/\sqrt{8}} dv_{q} ~ \mathcal{N}\Big(\tanh(\frac{\Delta^2}{2})/\lambda,v_q\Big)\nonumber\\
    &\quad  \times \text{Erf}\Big( \frac{\sqrt{\pi}}{4\sqrt{8} \sqrt{\frac{ v_q^2}{2 \lambda \tanh(\frac{\Delta^2}{2})}+\lambda \tanh(\frac{\Delta^2}{2})} }  \Big)
    \end{align}
\end{subequations}

In the next step, we limit the integral over $v_q$ to the interval
\begin{equation}
    v_q \in \Big(-\sqrt[4]{\tanh(\frac{\Delta^2}{2})/\lambda},\sqrt[4]{\tanh(\frac{\Delta^2}{2})/\lambda} \Big), 
\end{equation}
under the assumption that
\begin{equation}
    \label{eq: condition for valid limitting integral}
    \sqrt[4]{\tanh(\frac{\Delta^2}{2})/\lambda} \leq \frac{1}{\sqrt{8}}.
\end{equation}
We will later derive a condition that specifies for which values of $\Delta$ and $\lambda$ this assumption holds. By limiting the interval of the integral, we can use $v_{q}^{2}\leq \sqrt{\mathrm{tanh}(\Delta^{2}/2)/\lambda}$ to obtain our final lower bound to $p_{{E}}(\bo 0)$ as
\begin{multline}
        \label{eq: lower bound p0 with Erf function}
        p_{E}(\bo 0) \geq \text{Erf}\left(\frac{\sqrt{\pi}}{\Big( \tanh(\Delta^2/2)/\lambda  \Big)^{1/4}}\right) \\
        \times \text{Erf}\left( \frac{\sqrt{\pi}}{(4\sqrt{8}) \sqrt{\frac{ \sqrt{\tanh(\Delta^2/2)/\lambda} }{2 \lambda \tanh(\frac{\Delta^2}{2})}+\lambda \tanh(\frac{\Delta^2}{2})} }  \right).
\end{multline}

\textbf{Step 3: Choosing the asymmetry value $\lambda(\Delta)$}

Based on the argument presented in the main text, we know that the shear error dominates in the small-$\Delta$ limit. The second $\mathrm{Erf}$ term can be identified as the contribution associated with the shear error. Therefore, we choose $\lambda(\Delta)$ to be the value of $\lambda$ that maximizes the second $\mathrm{Erf}$ function in \cref{eq: lower bound p0 with Erf function}, which amounts to minimizing 
\begin{equation}
    \min_{\lambda} \Big(  \frac{ \sqrt{\tanh(\Delta^2/2)/\lambda} }{2 \lambda \tanh(\frac{\Delta^2}{2})}+\lambda \tanh(\frac{\Delta^2}{2}) \Big).
\end{equation}
Through some simple algebra, we find that 
\begin{subequations}
    \label{eq: ansatz for optimal ansatz}
    \begin{align}
            \lambda(\Delta) & = \frac{3^{2/5}}{2^{4/5}} \Big( \tanh(\Delta^2/2) \Big)^{-3/5} \\
            & \propto \Big( \tanh(\Delta^2/2) \Big)^{-3/5}.
    \end{align}
\end{subequations}
Substituting back we find that 
\begin{subequations}
    \label{eq: asymtotic scaling in the FT proof}
    \begin{gather}
        \tanh(\Delta^2/2)/\lambda(\Delta) \propto \tanh(\Delta^2/2)^{8/5}\\
         \frac{ \sqrt{\tanh(\Delta^2/2)/\lambda(\Delta)} }{2 \lambda(\Delta) \tanh(\frac{\Delta^2}{2})}+\lambda(\Delta) \tanh(\frac{\Delta^2}{2}) \propto \tanh(\Delta^2/2)^{2/5}. 
    \end{gather}
\end{subequations}
Given our choice of $\lambda(\Delta)$, we can also derive the condition on $\Delta$ under which Eq.~\eqref{eq: lower bound p0 with Erf function} is valid, namely
\begin{equation}
    \Delta \leq \sqrt{2} \sqrt{\text{arctanh} \left(\frac{1}{16}\sqrt[4]{\frac{3}{2}}\right)} \approx 0.372 .
\end{equation}

Again using our choice of $\lambda(\Delta)$ given in Eq.~\eqref{eq: ansatz for optimal ansatz}, by substituting Eq.~\eqref{eq: asymtotic scaling in the FT proof} back into Eq.~\eqref{eq: lower bound p0 with Erf function}, we find that the RHS of Eq.~\eqref{eq: lower bound p0 with Erf function} approaches one as $\Delta \to 0$. Due to Eq.~\eqref{eq: sum of probability}, we can conclude that 
\begin{equation}
    \label{eq: limit of pE0}
    \lim_{\Delta \to 0} p_{{E}}(\bo 0) = 1
\end{equation}
and 
\begin{equation}
    \lim_{\Delta \to 0} p_{{E}}(\sqrt{2} \bo n) = 0 \quad \forall \bo n \in \mathbb{Z}^2 \backslash \{ \mathbf{0} \} . 
\end{equation}
One final observation is that that the function $C(\Delta,\lambda(\Delta))$ from~\cref{eq: C Delta lambda} is a strictly decreasing function as $\Delta$ decreases and 
\begin{equation}
    \label{eq: limit of the constant CDelta}
   \lim_{\Delta \to 0 } C(\Delta,\lambda(\Delta)) = 1
\end{equation}

By combining Eq.~\eqref{eq: limit of pE0} and Eq.~\eqref{eq: limit of the constant CDelta}, we find that the RHS of Eq.~\eqref{eq: lower bound written in term of probability} approaches one in the limit $\Delta \to 0$. Consequently, the gate fidelity $\bar{F}$ also approaches one in this limit. Hence, it implies that for any given target gate infidelity $\epsilon_{\rm target}$, there exists a threshold $\Delta_{\rm th}$ such that
\begin{equation}
    \Delta \leq \Delta_{\rm th} \Rightarrow 1-F(\Delta,\lambda(\Delta)) \leq \epsilon_{\rm target}. 
\end{equation}
which proves Theorem~\ref{theorem: FT theorem T3 gate}. 
    
\end{proof}

\section{Numerical Methods}\label{sec: numerical methods}

In this appendix we detail the numerical methods we used to produce the plots in \cref{sec: numerics} and provide some additional plots that we did not include in the main text.

\subsection{Polynomial Phase Gates}

We will begin by detailing the numerical methods we use to model the effective logical channel $\mathcal{E}$ that we introduced in \cref{sec: numerics}. To recap, here we define the effective logical channel as
\begin{equation}
\tag{\ref{eq: logical channel equation}}
\raisebox{-.43\height}{\includegraphics{figures/Logical_channel_equation.pdf}},
\end{equation}
where $\Delta_{q}=\Delta/\lambda$ and $\Delta_{p}=\Delta\lambda$ as usual. Our definition of $\mathcal{E}$ here differs from the definition in \cref{appendix: proof of fault-tolerance theorem} by including a non-biased noise channel to account for noise in the QEC cycle that follows.

Our numerics are state vector simulations in the Fock basis with a variable truncation dimension that increases over the course of the simulation. In particular, the key subroutine of our simulation takes the following inputs: an input qubit state $\ket{\psi}$, a polynomial phase gate with polynomial $P(x)$, a GKP codestate approximation level $\Delta$, an asymmetry $\lambda$, and a measured Pauli operator ${P}$. The subroutine then outputs the expectation value $\mathrm{tr}\big({P}\mathcal{E}(\ket{\psi}\!\bra{\psi})\big)$, where $\mathcal{E}$ is the effective logical channel of the gate. We will explain later how this subroutine can be used to determine the average gate fidelity and state fidelity.

Before proceeding, note that \cref{eq: logical channel equation} can be rewritten using \cref{eq: relation between biased square GKP to non-biased rectangular GKP} as
\begin{equation}
\label{eq: logical channel equation 2}
\raisebox{-.43\height}{\includegraphics{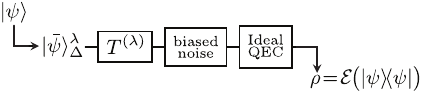}}.
\end{equation}
Note that this is \textit{not} the same as the on-demand morphing scheme described in \cref{appendix: mid-circuit morphing} because the noise preceding the ideal QEC box is biased. \cref{eq: logical channel equation 2} is preferable numerically because the non-biased envelope of the input state $\ket{\bar{\psi}}_{\Delta}^{\lambda}$ lowers the average photon number of the state being simulated, reducing the truncation error.

With the rewrite in \cref{eq: logical channel equation 2}, we can describe the simulation in three steps:
\begin{enumerate}
    \item State preparation of a logical rectangular GKP codestate with a non-biased envelope $\ket{\bar{\psi}}_{\Delta}^{\lambda}$ as defined in \cref{eq: rectangular codewords};
    \item The polynomial phase gate acting on the rectangular GKP code, given by $\mathrm{exp}\big(2\pi i P({q}/(\sqrt{\lambda \pi}))\big)$; and
    \item A biased Gaussian random displacement channel $\mathcal{G}_{\Sigma}$ with covariance matrix $\Sigma = \mathrm{tanh}(\Delta^{2}/2)\mathrm{diag}(\lambda,\lambda^{-1})$, followed by ideal QEC and a trivial unencoding into a logical qubit density matrix. The noise, QEC, and unencoding will all be simulated together.
\end{enumerate}

We now proceed in deriving the equations that we use in our simulation, starting with the preparation of the $\ket{\bar{\psi}}_{\Delta}^{\lambda}$ state. For notational clarity we denote the vacuum state as $\ket{\mathrm{vac}}$, single-mode displacement operators as ${W}(\vect{v})={W}(v_{1},v_{2})$, and coherent states as displaced vacuum states ${W}(\vect{v})\ket{\mathrm{vac}}$. Working towards an expression in the Fock basis $\{\ket{n}_{\text{Fock}}\}$ that we can numerically evaluate, we begin with a decomposition of ideal GKP codestates into coherent states following Ref.~\cite{GKP2001}. In particular, the $\ket{\bar{0}}^{\lambda}$ state is the unique simultaneous $+1$ eigenstate of the rectangular GKP stabiliser ${S}_{q}={W}(\sqrt{2\lambda},0)$ and the logical operator ${Z}={W}(0,1/\sqrt{2\lambda})$. Therefore, we can write the projector onto the $\ket{\bar{0}}^{\lambda}$ state as~\cite{GKP2001}
\begin{equation}\label{eq: 0 state projector}
    \ket{\bar{0}}^{\lambda}\bra{\bar{0}}^{\lambda}=\sum_{m,n\in\mathbb{Z}}{S}_{q}^{m}{Z}^{n},
\end{equation}
since any other simultaneous eigenstate of ${S}_{q}$ and ${Z}$ has at least one of these eigenvalues not equal to one, and will therefore be a 0-eigenstate of \cref{eq: 0 state projector}.
Applying this projector to the vacuum state gives
\begin{equation}\label{eq: logical 0 coherent states}
    \ket{\bar{0}}^{\lambda}=\sum_{m,n\in\mathbb{Z}}(-1)^{mn}{W}\big(m\sqrt{2\lambda},n/\sqrt{2\lambda}\big)\ket{\mathrm{vac}},
\end{equation}
and applying ${X}={W}(\sqrt{\lambda/2},0)$ to \cref{eq: logical 0 coherent states} gives
\begin{equation}
    \ket{\bar{1}}^{\lambda}=\sum_{m,n\in\mathbb{Z}}i^{2mn-n}{W}\big((2m+1)\sqrt{\lambda/2},n/\sqrt{2\lambda}\big)\ket{\mathrm{vac}}.
\end{equation}

Next, we need an expression for the approximate codestates $e^{-\Delta^{2}{a}^{\dag}{a}}$. This is simple to derive from the identity
\begin{multline}
    e^{-\Delta^{2}{a}^{\dag}{a}}{W}(v_{q},v_{p})e^{\Delta^{2}{a}^{\dag}{a}}\\
    =e^{-\pi(v_{q}^{2}+v_{p}^{2})(1-e^{-2\Delta^{2}})/2}\hspace{2 cm}\\
    \times{W}(e^{-\Delta^{2}}v_{q},e^{-\Delta^{2}}v_{p})\\
    \times e^{2\sqrt{\pi}(-v_{1}+iv_{2})\sinh(\Delta^{2}){a}},
\end{multline}
which itself can be derived from the Baker-Campbell-Haussdorf formula. Since $e^{-\Delta^{2}{a}^{\dag}{a}}\ket{\mathrm{vac}}=e^{\alpha {a}}\ket{\mathrm{vac}}=\ket{\mathrm{vac}}$, this gives
\begin{subequations}\label{eq: approx logical codestates coherent states}
    \begin{align}
        \ket{\bar{0}}^{\lambda}_{\Delta}&=\sum_{m,n\in\mathbb{Z}}c_{2m,n}(-1)^{mn}{W}\big(m\sqrt{2\lambda},n/\sqrt{2\lambda}\big)\ket{\mathrm{vac}},\\
        \ket{\bar{1}}^{\lambda}_{\Delta}&=\sum_{m,n\in\mathbb{Z}}\bigg(c_{2m+1,n}i^{2mn-n}\nonumber\\
        &\hspace{0.5 cm}\times{W}\big((2m{+}1)\sqrt{\lambda/2},n/\sqrt{2\lambda}\big)\ket{\mathrm{vac}}\bigg),
    \end{align}
\end{subequations}
where we have defined $c_{m,n}=\exp\big(-\pi(m^{2}\lambda+n^{2}/\lambda)(1-e^{-2\Delta^{2}})/4\big)$ for convenience. Meanwhile, coherent states themselves can be represented in the Fock basis using the standard expression, which in our notation is
\begin{equation}\label{eq: coherent state Fock basis}
    {W}(v_{q},v_{p})\ket{\mathrm{vac}}=e^{-\pi(v_{q}^{2}+v_{p}^{2})/2}\sum_{n=0}^{\infty}\frac{\sqrt{\pi}^{n}(v_{q}{+}iv_{p})^{n}}{\sqrt{n!}}\ket{n}_{\text{Fock}}.
\end{equation}
Using \cref{eq: approx logical codestates coherent states,eq: coherent state Fock basis} allows us to generate GKP codestates in the Fock basis numerically, truncated to a desired Fock truncation dimension $d_{\text{trunc,init}}$; in this work, we generate the codestates at $d_{\text{trunc,init}}=400$.

Some care needs to be taken to ensure the state generation method is stable at high Fock dimensions. In particular, calculating $(v_{q}+iv_{p})^{n}/\sqrt{n!}$ for large $n$ can be numerically unstable and can return large or infinite results. This is undesirable because coherent states with large $|v_{q}+iv_{p}|$ are exponentially suppressed in the GKP codestates themselves due to the $c_{m,n}$ coefficient. To combat this, we simply check whether the coherent state has returned finite values, otherwise we set it to zero to avoid corrupting the remaining GKP codestate. The exclusion of some of these coherent states will cause some numerical error in addition to the usual truncation error, although in our tests we found this to be lower than $10^{-6}$ across the full range of $\Delta$ that we probed.

One important subtlety that must be taken into account is that the codestates $\ket{\bar{0}}^{\lambda}_{\Delta}$ and $\ket{\bar{1}}^{\lambda}_{\Delta}$ are not orthogonal. This is a problem because it would mean that the effective logical ``channel'' $\mathcal{E}$ would not be a channel because it is not trace-preserving. To rectify this, we follow the same procedure as in Ref.~\cite{shaw2022} and orthonormalize the codewords in \cref{eq: approx logical codestates coherent states}. First, the output states from \cref{eq: approx logical codestates coherent states} are not normalized so we begin by normalizing them. Then, to orthogonalize two arbitrary, normalized vectors $\ket{\psi}$ and $\ket{\phi}$, we first calculate the angle $\theta=\mathrm{arg}(\langle\psi|\phi\rangle)$. Then, the superpositions $\ket{+}=N_{+}(\ket{\psi}+e^{-i\theta}\ket{\phi})$ and $\ket{-}=N_{-}(\ket{\psi}-e^{-i\theta}\ket{\phi})$ are guaranteed to be orthogonal. Moreover, we choose the constants $N_{\pm}$ such that the states $\ket{\pm}$ are normalized. We can then construct orthogonalized states $\ket{\psi}_{\text{orth}}=(\ket{+}+\ket{-})/\sqrt{2}$ and $\ket{\phi}_{\text{orth}}=e^{i\theta}(\ket{+}-\ket{-})/\sqrt{2}$. This procedure coincides with the L\"owdin orthogonalization~\cite{Lowdin50} in the special case of having two 2D vectors. The additional error introduced by this procedure is negligible for $\bar{n}>2$ ($\Delta<0.45$), see for example Figure 11 of Ref.~\cite{shaw2022}. With this, we can obtain arbitrary logical codestates by taking superpositions of the orthonormal codestates.

With the GKP codestate $\ket{\bar{\psi}}_{\Delta}^{\lambda}$ now prepared in the Fock basis, we can apply the polynomial phase gate. The truncated Fock basis representation of the polynomial phase gate $\mathrm{exp}\big(2\pi i P({q}/(\lambda\sqrt{\pi}))\big)$ is straight-forward to obtain simply by calculating ${q}=(a+a^{\dag})/\sqrt{2}$ in the truncated Fock basis and taking the appropriate matrix exponential. However, there are two problems with this simple approach. The first is that performing the matrix exponential over the truncated Fock basis does not coincide with truncating the exponent of the infinite-dimensional operator, leading to additional numerical errors. This can be solved by performing the matrix exponential over a temporary, higher truncation dimension, say $d_{\text{trunc,temp}}=3d_{\text{trunc,init}}$, and then truncating to $d_{\text{trunc,init}}$ afterwards. The second problem is that even if the input state has no support on Fock basis states above $d_{\text{trunc,init}}$, the gate may increase its photon number such that the output state after the gate does have support on Fock basis states above $d_{\text{trunc,init}}$, leading to additional truncation error. We solve this second problem by only truncating the \textit{input} dimension of the matrix, so that the gate is represented by a $d_{\text{trunc,out}}\times d_{\text{trunc,init}}=3d_{\text{trunc,init}}\times d_{\text{trunc,init}}=1200\times400$ matrix. This way, the output state has a higher truncation dimension than the input state, avoiding any significant truncation error due to the gate.

Finally, we explain how to model the biased noise channel $\mathcal{G}_{\Sigma}$, the final ideal QEC round, and the unencoding into the logical qubit subspace. We start by explaining how to model the ideal error correction and unencoding without the preceding biased noise channel; we will explain how to incorporate $\mathcal{G}_{\Sigma}$ afterwards. To begin, note that applying ideal QEC to any CV state $\rho$ maps it to an ideal GKP codestate $\bar{\tau}$ whose logical state is given by~\cite{shaw2024logical}
\begin{equation}\label{eq: partial trace with P_ms}
    {\tau} = \frac{1}{2}\Big({I}\mathrm{tr}({\rho})+{X}\mathrm{tr}({\rho}{X}_{\text{m}}^{\lambda})+{Y}\mathrm{tr}({\rho}{Y}_{\text{m}}^{\lambda})+{Z}\mathrm{tr}({\rho}{Z}_{\text{m}}^{\lambda})\Big).
\end{equation}
Here, the operators ${X}_{\text{m}}^{\lambda}$, ${Y}_{\text{m}}^{\lambda}$ and ${Z}_{\text{m}}^{\lambda}$ operators are the \textit{ideal Pauli measurement operators} introduced in Ref.~\cite{shaw2022}. In particular, the ${Z}_{\text{m}}^{\lambda}$ is defined such that it acts as a logical ${Z}$ on any ideal GKP codestate that has been displaced by a \textit{correctable} displacement error~\cref{eq: correctable displacement}; that is, it acts as
\begin{subequations}
\begin{align}
    {Z}_{\text{m}}^{\lambda}{W}(v_{q},v_{p})\ket{\bar{0}}^{\lambda}&=+{W}(v_{q},v_{p})\ket{\bar{0}}^{\lambda},\\
    {Z}_{\text{m}}^{\lambda}{W}(v_{q},v_{p})\ket{\bar{1}}^{\lambda}&=-{W}(v_{q},v_{p})\ket{\bar{1}}^{\lambda},
\end{align}
\end{subequations}
for all $-\sqrt{\lambda/8}<v_{q}\leq\sqrt{\lambda/8}$ and $-1/\sqrt{8\lambda}<v_{p}\leq 1/\sqrt{8\lambda}$. Similar expressions also exist for the other Pauli operators ${X}_{\text{m}}$ and ${Y}_{\text{m}}$ acting on their respective correctable displaced GKP logical eigenstates.

Our task is, therefore, to obtain the Fock basis representation of the ideal Pauli measurement operators. In Appendix D of Ref.~\cite{shaw2024logical} it was shown that these ideal Pauli measurement operators can be written as an infinite sum of displacement operators given by
\begin{subequations}\label{eq: ideal Pauli measurements}
    \begin{align}
        {X}_{\text{m}}^{\lambda}&=\frac{1}{\pi}\sum_{n\in\mathbb{Z}}\frac{(-1)^{n}}{n+\frac{1}{2}}{W}\big((2n+1)\sqrt{\lambda/2},0\big),\\
        {Z}_{\text{m}}^{\lambda}&=\frac{1}{\pi}\sum_{n\in\mathbb{Z}}\frac{(-1)^{n}}{n+\frac{1}{2}}{W}\big(0,(2n+1)/\sqrt{2\lambda}\big).
    \end{align}
\end{subequations}
The displacement operators themselves can be obtained by taking a matrix exponential, again we use a higher temporary Fock truncation dimension to perform the truncation to avoid numerical errors.
At this stage in the simulation, our Fock state vector has dimension $d_{\text{trunc,out}}=3d_{\text{trunc,init}}=1200$; therefore, the matrix exponential is performed at a dimension of $d_{\text{trunc,temp}}=3d_{\text{trunc,out}}=3600$ before truncating back to a $d_{\text{trunc,out}}\times d_{\text{trunc,out}}=1200\times1200$ matrix. With this, we can obtain the ${X}_{\text{m}}^{\lambda}$ and ${Z}_{\text{m}}^{\lambda}$ operators by truncating the sum over the displacement operators \cref{eq: ideal Pauli measurements} such that $(2n+1)$ runs over all odd numbers from $-59$ to $+59$, which we found to be sufficient to avoid significant errors from this truncation. Finally, the ${Y}_{\text{m}}^{\lambda}$ operator is calculated by
\begin{equation}
    {Y}_{\text{m}}^{\lambda}=\frac{1}{2}\big(i{X}_{\text{m}}^{\lambda}{Z}_{\text{m}}^{\lambda}-i{Z}_{\text{m}}^{\lambda}{X}_{\text{m}}^{\lambda}\big),
\end{equation}
with the symmetric form of this expression chosen to improve numerical stability.

Now, we return to the issue of the Gaussian random displacement channel (arising from the noisy stabilizer measurements) that must be applied before taking the expectation value with the ideal Pauli measurement operators. In particular, Gaussian random displacement channels obey the relation
\begin{equation}
    \mathrm{tr}\big({O}\mathcal{G}_{\Sigma}({\rho})\big)=\mathrm{tr}\big(\mathcal{G}_{\Sigma}({O}){\rho}\big)
\end{equation}
for any operator $O$. Therefore, we can rewrite the ideal QEC and decoding steps in \cref{eq: partial trace with P_ms} to incorporate $\mathcal{G}_{\Sigma}$ as
\begin{subequations}
\begin{align}
    {\tau} &= \frac{1}{2}\Big({I}\mathrm{tr}\big(\mathcal{G}_{\Sigma}({\rho})\big)+{X}\mathrm{tr}\big(\mathcal{G}_{\Sigma}({\rho}){X}_{\text{m}}^{\lambda}\big)\nonumber\\
    &\qquad\qquad+{Y}\mathrm{tr}\big(\mathcal{G}_{\Sigma}({\rho}){Y}_{\text{m}}^{\lambda}\big)+{Z}\mathrm{tr}\big(\mathcal{G}_{\Sigma}({\rho}){Z}_{\text{m}}^{\lambda}\big)\Big)\\
    &= \frac{1}{2}\Big({I}\mathrm{tr}\big({\rho}\big)+{X}\mathrm{tr}\big({\rho}{X}_{\text{m},\Sigma}^{\lambda}\big)\nonumber\\
    &\qquad\qquad\qquad+{Y}\mathrm{tr}\big({\rho}{Y}_{\text{m},\Sigma}^{\lambda}\big)+{Z}\mathrm{tr}\big({\rho}{Z}_{\text{m},\Sigma}^{\lambda}\big)\Big),
\end{align}
\end{subequations}
where we have introduced the \textit{approximate} Pauli measurement operators defined as ${P}_{\text{m},\Sigma}^{\lambda}=\mathcal{G}_{\Sigma}({P}_{\text{m}}^{\lambda})$. The expressions for the approximate Pauli measurement operators can be derived from their ideal counterparts~\cref{eq: ideal Pauli measurements} via the equation
\begin{equation}\label{eq: GRD on displacement}
    \mathcal{G}_{\Sigma}\big({W}(v_{q},v_{p})\big)=\mathrm{exp}\big({-}\pi\vect{v}^{T}\Omega^{T}\Sigma\Omega\vect{v}\big){W}(v_{q},v_{p}),
\end{equation}
where $\Omega = \begin{bmatrix}0&1\\-1&0\end{bmatrix}$ is the single-mode version of \cref{eq: Omega matrix} and $\vect{v}=\begin{bmatrix}v_{q}\\ v_{p}\end{bmatrix}$. This amounts simply to a rescaling of the coefficients in \cref{eq: ideal Pauli measurements}.
Therefore, applying \cref{eq: GRD on displacement} to \cref{eq: ideal Pauli measurements} gives the approximate Pauli measurement operators in the Fock basis, allowing us to model the approximate QEC and unencoding steps at the same time in our numerics. 
Moreover, this has the advantage that these approximate Pauli measurement operators can be applied to pure states, avoiding the need to compute and store a CV density matrix.

To briefly summarise, we can calculate an expectation value $\mathrm{tr}\big({P}\mathcal{E}(\ket{\psi}\!\bra{\psi})\big)$ by first obtaining the $\ket{\bar{\psi}}_{\Delta}^{\lambda}$ state using superpositions of the orthonormalised versions of the states \cref{eq: approx logical codestates coherent states} truncated to $d_{\mathrm{trunc,init}}=400$, then applying the polynomial phase gate truncated to a $d_{\mathrm{trunc,out}}\times d_{\mathrm{trunc,init}}=1200\times400$ matrix, then taking the expectation value of this output state with the approximate Pauli measurement operator ${P}_{\mathrm{m},\Sigma}^{\lambda}$ (truncated to $d_{\mathrm{trunc,out}}\times d_{\mathrm{trunc,out}}=1200\times1200$) with $\Sigma=\mathrm{tanh}(\Delta^{2}/2)\mathrm{diag}(\lambda,\lambda^{-1})$, using \cref{eq: GRD on displacement,eq: ideal Pauli measurements}.

With these ingredients we can model the effective logical channel applied to any pure input qubit state. To calculate the average gate fidelity efficiently we utilise the formula from Ref.~\cite{NIELSEN2002}, giving
\begin{equation}
    F(\mathcal{E},{U})=\frac{1}{3}+\frac{1}{12}\sum_{j=0}^{3}\mathrm{tr}\big({U}{\sigma}_{j}{U}^{\dag}\mathcal{E}({\sigma}_{j})\big),
\end{equation}
where the ${\sigma}_{j}$ operators are the ${I},{X},{Y}$ and ${Z}$ Pauli matrices respectively for $j=0,1,2,3$. After some rearrangements we obtain
\begin{equation}\label{eq: simple formula rearranged}
    F(\mathcal{E},{U})=\frac{1}{3}+\frac{1}{12}\sum_{k,\ell=0}^{3}\mathrm{tr}\big({U}{V}_{k}{U}^{\dag}{\sigma}_{\ell}\big)\mathrm{tr}\big({\sigma}_{\ell}\mathcal{E}(\ket{\psi_{k}}\!\bra{\psi_{k}})\big),
\end{equation}
where $\{\ket{\psi_{k}}\}$ is a set of (qubit) pure states such that $\ket{\psi_{k}}\!\bra{\psi_{k}}$ is a basis for the space of qubit operators, and ${V}_{k}=\sum_{j}\alpha_{jk}{\sigma}_{j}/2$ where $\alpha_{jk}$ are coefficients such that ${\sigma}_{j}=\sum_{k}\alpha_{jk}\ket{\psi_{k}}\!\bra{\psi_{k}}$. In this work we choose $\ket{\psi_{k}}\in\{\ket{1},\ket{+},\ket{i},\ket{0}\}$ where $\ket{i}=(\ket{0}+i\ket{1})/\sqrt{2}$, although we expect that different choices would give identical numerical performance. \cref{eq: simple formula rearranged} is convenient because the latter trace can be evaluated simply by evaluating the relevant expectation value of the ideal Pauli operator as described above, instead of explicitly evaluating $\mathcal{E}$.

A similar method can also be used to extract the \textit{state} fidelity in \cref{fig:T_states}. In particular, we wish to calculate
\begin{equation}
    F=\bra{T}\mathcal{E}\big(\ket{+}\!\bra{+}\big)\ket{T}.
\end{equation}
To do this, note that $\ket{T}\!\bra{T}={I}/2+({X}+{Y})/(2\sqrt{2})$, and therefore we can write
\begin{equation}
    F=\frac{1}{2}+\frac{1}{2\sqrt{2}}\sum_{j=1}^{2}\mathrm{tr}\big({\sigma}_{j}\mathcal{E}(\ket{+}\!\bra{+})\big),
\end{equation}
which again can be easily evaluated using the Pauli measurement operators ${X}_{\text{m},\Sigma}^{\lambda}$ and ${Y}_{\text{m},\Sigma}^{\lambda}$.

For each polynomial phase gate, we calculated the average gate fidelity for input values of $\Delta=1/\sqrt{2\bar{n}+1}$ where $\bar{n}$ ranges from 2 to 20 in increments of 0.5, and for 32 input values of the asymmetry $\lambda$ ranging from 1 to 6.5. We also collected the same data for the state fidelity of the produced $\ket{T}$ state for each polynomial phase gate that we tested that implements a ${T}$ gate. The slowest stage of the simulation is the generation of the ideal Pauli measurement operators, so we pre-computed these for each value of $\lambda$ and $\Delta$ and re-used it for each gate. 

\begin{figure*}
\centering
    \includegraphics[]{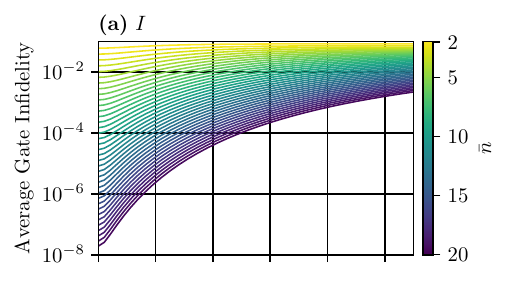}\hfill\;\\[-0.3 cm]
    \includegraphics[]{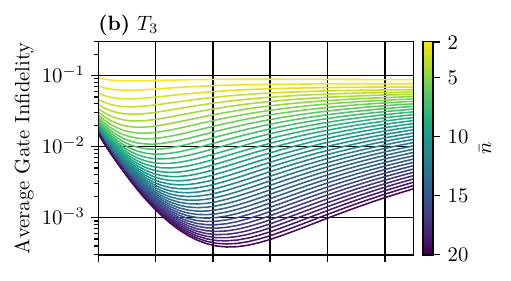}\hfill 
    \includegraphics[]{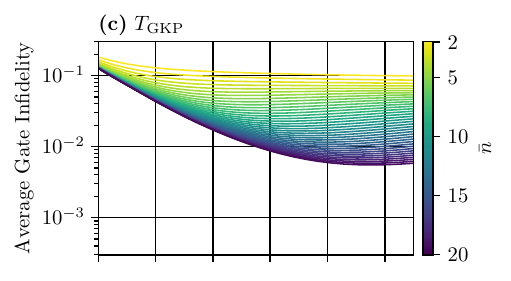}\\[-0.35 cm]
    \includegraphics[]{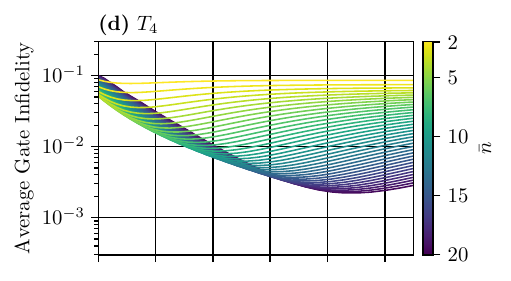}\hfill 
    \includegraphics[]{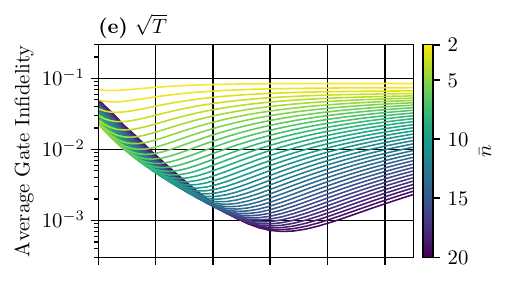}\\[-0.35 cm]
    \includegraphics[]{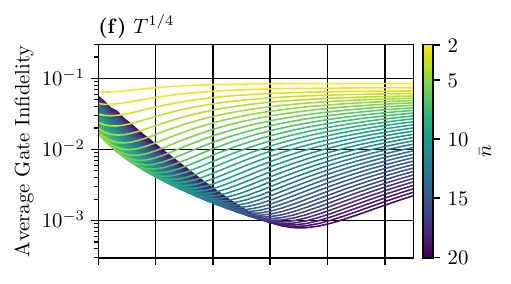}\hfill 
    \includegraphics[]{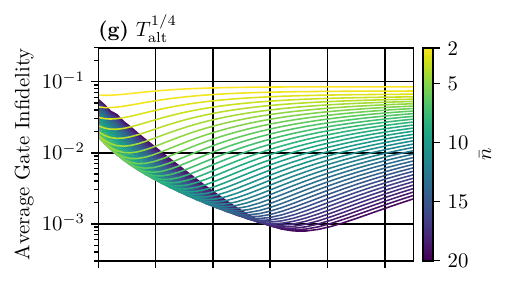}\\[-0.35 cm]
    \includegraphics[]{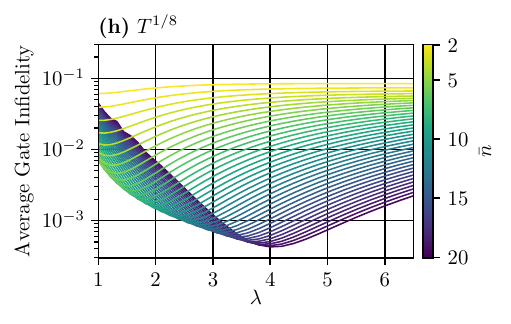}\hfill 
    \includegraphics[]{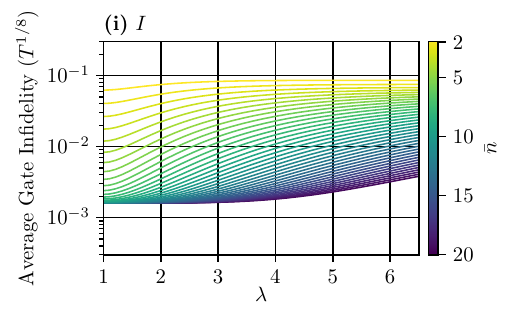}
    \caption{Average gate infidelity plots for the polynomial phase gates in \cref{tab: simulation polynomials,eq: 4rt_T_minus}, for each value of asymmetry and $\Delta$ as described in the text. Note that in subfigure (i), we implement the identity gate on the CV mode but calculate the average gate infidelity with respect to the $T^{1/8}$ gate.}\label{fig:all_gates}
\end{figure*}

\begin{figure*}
    \centering
    \includegraphics[]{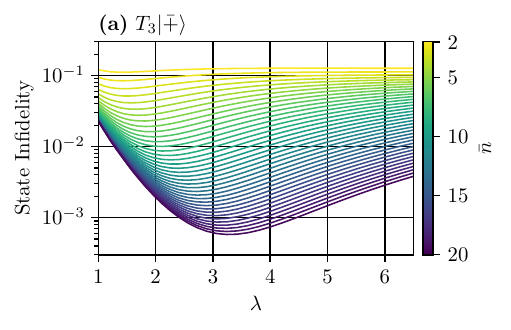}\hfill \includegraphics[]{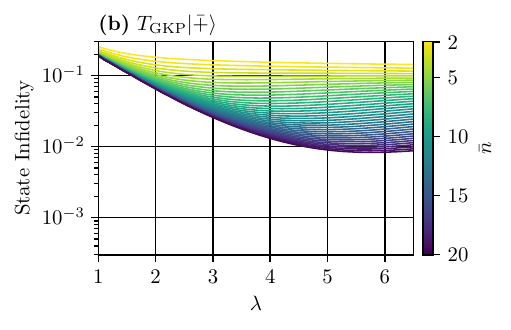}\\
    \includegraphics[]{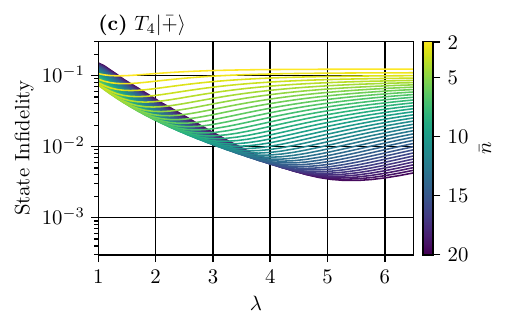}
    \caption{State infidelity plots for the polynomial phase gates in \cref{tab: simulation polynomials} that implement a $T$ gate, for each value of asymmetry and $\Delta$ as described in the text.}\label{fig:all_states}
\end{figure*}

The full results are presented in \cref{fig:all_gates,fig:all_states}, from which the plots in the main text are produced. Note that there are no error bars here since the simulations are matrix calculations instead of Monte Carlo simulations. One can see here that achieving the precisely optimal asymmetry is not crucial to achieve good performance. Moreover, note that for the identity gate, the optimal asymmetry is always 1, indicating that under idling noise the square GKP code outperforms all rectangular GKP codes as expected. Compared to the main text we also tested two additional gates. The first is an alternative polynomial phase gate implementation of the ${T}^{1/4}$ gate given by
\begin{equation}\label{eq: 4rt_T_minus}
    P(x)=-x^{5}/240-x^{4}/96+x^{3}/48+x^{2}/24-x/60,
\end{equation}
which is related to the polynomial for the ${T}^{1/4}$ gate given in \cref{tab: simulation polynomials} by the transformation $x\leftrightarrow -x$. Both polynomials are minimal under the lexicographic ordering because they differ only in the absolute value of each of the coefficients, the two implementations perform identically up to numerical error. The second is an ``implementation'' of the ${T}^{1/8}$ qubit gate that at the CV level implements an identity gate ${I}$ (that is, the polynomial is $P(x)=0$). In other words, we implement the identity gate ${I}$ at the CV level, but then evaluate the average gate infidelity with respect to the ${T}^{1/8}$ gate. The purpose of this is to evaluate a trivial benchmark in terms of average gate infidelity that the non-trivial implementation of the $T^{1/8}$ gate must beat.

\begin{figure}
    \includegraphics[]{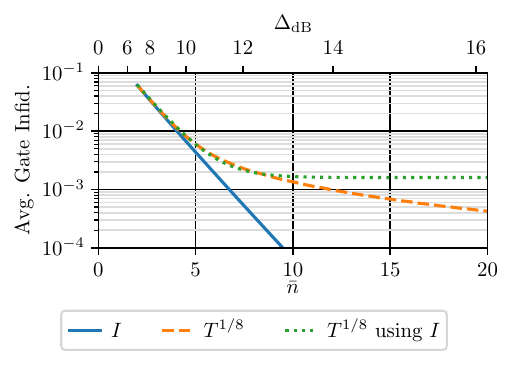}
    \caption{Comparison of the trivial (green dotted) and non-trivial (orange dashed) implementations of the $T^{1/8}$ gate. For reference, we also plot the identity gate ${I}$ evaluated with the average gate infidelity with respect to the identity gate (blue solid).}\label{fig:8rt_T_gates}
\end{figure}

In \cref{fig:8rt_T_gates} we compare the trivial and non-trivial implementations of the ${T}^{1/8}$ gates. Note that the trivial and non-trivial implementations have almost identical performance until the trivial implementation reaches an infidelity floor at $\sim1.6\times10^{-3}$, which corresponds to the qubit-level average gate infidelity between ${T}^{1/8}$ and ${I}$ as explained in the main text. To achieve this, one requires GKP codestates of higher quality than $\Delta = 0.25$ ($\bar{n}=8$, $\Delta_{\text{dB}}=12$).

\subsection{Vacuum State Method}

Finally, we explain how we simulate the vacuum state method, following the same method as in Ref.~\cite{Baragiola2019}.

The scheme that we wish to simulate consists of the following steps:
\begin{enumerate}
    \item Begin with a vacuum state $\ket{\mathrm{vac}}$,
    \item Measure the stabilisers with a round of square GKP QEC (with noisy measurements), and
    \item Based on the stabiliser measurement outcomes, perform a logical Clifford correction to make the state as close to the $\ket{{T}}$ state as possible.
\end{enumerate}
After this procedure, we average over the possible measurement outcomes that could have occurred (weighted appropriately by the probability of obtaining each outcome) giving a CV density matrix ${\rho}_{\text{CV}}$. The round of QEC is modeled with noisy measurements but projects the state into the ideal GKP subspace, which follows the same noise model as the QEC round 2 in the circuit $\mathcal{C}$ as explained in the main text. This means that ${\rho}_{\text{CV}}$ is in the ideal GKP subspace and can be trivially converted into a qubit density matrix $\rho$ which is the output of this procedure. We can then quantify the quality of the output state by taking the state fidelity between ${\rho}$ and $\ket{T}=T\ket{+}$.

In our simulations, the only noise arises from the noisy QEC measurements, which we model as a non-biased Gaussian random displacement function $\mathcal{G}_{\mathrm{tanh}(\Delta^{2}/2)I}$ acting on the vacuum state prior to an ideal GKP QEC round. Conveniently, a non-biased Gaussian random displacement function applied to the vacuum state corresponds exactly to a \textit{thermal} state ${\rho}_{\mathrm{th}}(\bar{n})$ with average photon number in our case given by $\bar{n}=\mathrm{tanh}(\Delta^{2}/2)$. This allows us to directly use the equations given in Ref.~\cite{Baragiola2019}, which give us the output qubit state conditioned on a given pair of stabilizer measurement outcomes $(e^{i\phi_{1}},e^{i\phi_{2}})$ along with the probability of these outcomes occurring. To perform our simulations, we sample the continuous set of possible stabiliser measurement outcomes from a $500\times500$ grid. For each output state, we calculate the state fidelity between it and each of the 12 states that are Clifford-equivalent to a $\ket{T}$ state, and return the maximum of these 12 fidelities. Finally, we sort the outcomes from highest fidelity to lowest.

In the case of no post-selection ($p=100\%$), we take the average of all these fidelities weighted according to the probability of obtaining the corresponding measurement outcome. With post-selection ($0\%<p<100\%$) we take the (weighted) average of the best fidelities up until the post-selection fraction has been reached. To obtain the infidelity lower-bound ($p=0\%$) we simply return the maximum fidelity (minimum infidelity) that was obtained in the grid, which in each case corresponded to obtaining $+1$ outcomes for both stabilisers. This is sufficient to produce~\cref{fig:T_states} in the main text.

\bibliography{report}

\end{document}